\documentclass{amsart}
\usepackage{amsmath,amsfonts,amssymb}
\usepackage{graphicx,amsmath, amscd}
\usepackage{color,comment}
\usepackage[T1]{fontenc}
\usepackage[latin1]{inputenc}

\usepackage[lmargin=2cm,rmargin=2cm,tmargin=2cm,bmargin=2cm]{geometry}
\usepackage[toc,page]{appendix}

\usepackage{tikz}

\usepackage{booktabs}
\usepackage{diagbox}
\usepackage{array}
\newcolumntype{C}[1]{>{\centering\let\newline\\\arraybackslash\hspace{0pt}}m{#1}}

\newtheorem{theorem}{Theorem}[section]
\newtheorem{lemma}[theorem]{Lemma}

\newtheorem{proposition}[theorem]{Proposition}
\newtheorem{definition}[theorem]{Definition}
\newtheorem{remark}[theorem]{Remark}
\newtheorem{defi/prop}[theorem]{Definition/Proposition}
\newtheorem{example}[theorem]{Example}

\newtheorem{fact}[theorem]{Fact}

\newcommand{\N}{\mathbf{N}}

\newcommand{\R}{\mathbf{R}}
\newcommand{\C}{\mathbf{C}}
\renewcommand{\P}{\mathbf{P}}

\newcommand{\e}{\varepsilon}
\renewcommand{\leq}{\leqslant}
\renewcommand{\geq}{\geqslant}

\newcommand{\st}{\  : \ }

\newcommand{\Id}{\mathrm{Id}}
\newcommand{\id}{\mathrm{id}}

\newcommand{\A}{\mathrm{A}}
\newcommand{\B}{\mathrm{B}}

\DeclareMathOperator{\vrad}{vrad}

\DeclareMathOperator{\im}{Im}
\DeclareMathOperator{\tr}{Tr}
\DeclareMathOperator{\Sym}{Sym}
\DeclareMathOperator{\E}{\mathbf{E}}

\newcommand{\ketbra}[2]{| #1 \rangle\!\langle #2 |}

\newcommand{\ket}[1]{| #1 \rangle}

\author{C\'ecilia Lancien}

\address{\textbf{C\'{e}cilia Lancien:} Institut Camille Jordan, Universit\'{e} Claude Bernard Lyon 1, 69622 Villeurbanne Cedex, France \& F\'{\i}sica Te\`{o}rica: Informaci\'{o} i Fen\`{o}mens Qu\`{a}ntics, Universitat Aut\`{o}noma de Barcelona, 08193 Bellaterra (Barcelona), Spain.} 

\email{lancien@math.univ-lyon1.fr}

\title{$k$-extendibility of high-dimensional bipartite quantum states}

\begin{document}

\begin{abstract}
The idea of detecting the entanglement of a given bipartite state by searching for symmetric extensions of this state was first proposed by Doherty, Parrilo and Spedialeri. The complete family of separability tests it generates, often referred to as the hierarchy of \textit{$k$-extendibility tests}, has already proved to be most promising. The goal of this paper is to try and quantify the efficiency of this separability criterion in typical scenarios. For that, we essentially take two approaches. First, we compute the average width of the set of $k$-extendible states, in order to see how it scales with the one of separable states. And second, we characterize when random-induced states are, depending on the ancilla dimension, with high probability violating or not the $k$-extendibility test, and compare the obtained result with the corresponding one for entanglement \textit{vs} separability. The main results can be precisely phrased as follows: on $\C^d\otimes\C^d$, when $d$ grows, the average width of the set of $k$-extendible states is equivalent to $(2/\sqrt{k})/d$, while random states obtained as partial traces over an environment $\C^s$ of uniformly distributed pure states are violating the $k$-extendibility test with probability going to $1$ if $s<((k-1)^2/4k)d^2$. Both statements converge to the conclusion that, if $k$ is fixed, $k$-extendibility is asymptotically a weak approximation of separability, even though any of the other well-studied separability relaxations is outperformed by $k$-extendibility as soon as $k$ is above a certain (dimension independent) value.
\end{abstract}

\maketitle

\section{Introduction}

Deciding whether a given bipartite quantum state is entangled or separable (or even just close to separable) is known to be a computationally hard task (see \cite{Gurvits} and \cite{Gharibian}). Several much more easily checkable necessary conditions for separability do exist though, the most famous and widely used ones being perhaps the positivity of partial transpose criterion \cite{Peres}, the realignment criterion \cite{CW} or the $k$-extendibility criterion \cite{DPS}. All of them have in common that verifying if a given state fulfils them or not may be cast as a Semi-Definite Program (SDP) and hence be efficiently solved (see e.g.~the quite extensive review \cite{Doherty} for much more on that topic).

We focus here on a relaxation of the notion of separability of quite different kind: the so-called $k$-extendibility criterion for separability, which was introduced in \cite{DPS}. It is especially appealing because it provides a hierarchy of increasingly powerful separability tests (expressible as SDPs of increasing dimension), which is additionally complete, meaning that any entangled state is guaranteed to fail a test after some finite number of steps in the hierarchy. Let us be more precise.

\begin{definition}
Let $k\in\N$. A state $\rho_{\A\B}$ on a bipartite Hilbert space $\mathrm{A}\otimes\mathrm{B}$ is $k$-extendible with respect to $\mathrm{B}$ if there exists a state $\rho_{\A\B^k}$ on $\mathrm{A}\otimes\mathrm{B}^{\otimes k}$ which is invariant under any permutation of the $\mathrm{B}$ subsystems and such that $\rho_{\A\B}=\tr_{\B^{k-1}}\rho_{\A\B^k}$.
\end{definition}

\begin{theorem}[The complete family of $k$-extendibility criteria for separability, \cite{DPS}] \label{th:k-extendibility-criterion}
A state on a bipartite Hilbert space $\mathrm{A}\otimes\mathrm{B}$ is separable if and only if it is $k$-extendible with respect to $\mathrm{B}$ for all $k\in\N$.
\end{theorem}

Note that one direction in Theorem \ref{th:k-extendibility-criterion} is obvious, namely that a separable state on some bipartite system is necessarily $k$-extendible for all $k\in\N$ (with respect to both subsystems). Indeed, if $\rho=\sum_{x}p_x\sigma_x\otimes\tau_x$ is separable, then $\sum_{x}p_x\sigma_x^{\otimes k}\otimes\tau_x$ and $\sum_{x}p_x\sigma_x\otimes\tau_x^{\otimes k}$ are symmetric extensions of $\rho$ to $k$ copies of the first and second subsystems respectively. The other direction in Theorem \ref{th:k-extendibility-criterion} follows from the quantum finite de Finetti theorem (see e.g.~\cite{KR,CKMR} for the seminal statements). The latter establishes, roughly speaking, that starting from a permutation-invariant state on some tensor power system and tracing out all except a few of the subsystems, one gets a state that may be well-approximated by a convex combination of tensor power states (with a vanishing error as the initial number of subsystems increases).

It is easy to see that if a state is $k$-extendible for some $k\in\N$, then it is automatically $k'$-extendible for all $k'\leq k$. Hence, the necessary and sufficient condition for separability provided by Theorem \ref{th:k-extendibility-criterion} actually decomposes into a series of increasingly constraining necessary conditions for separability, which are only asymptotically also sufficient (see Figure \ref{fig:k-ext}). In real life however, checks can only be done up to a finite level in this hierarchy. It thus makes sense to ask, given a finite $k\in\N$, how ``powerful'' the $k$-extendibility test is to detect entanglement.

\begin{figure}[h] \caption{The nested and converging sequence of $k$-extendibility relaxations of separability, $k\in\N$}
\label{fig:k-ext}
\begin{center}
\includegraphics[width=6cm]{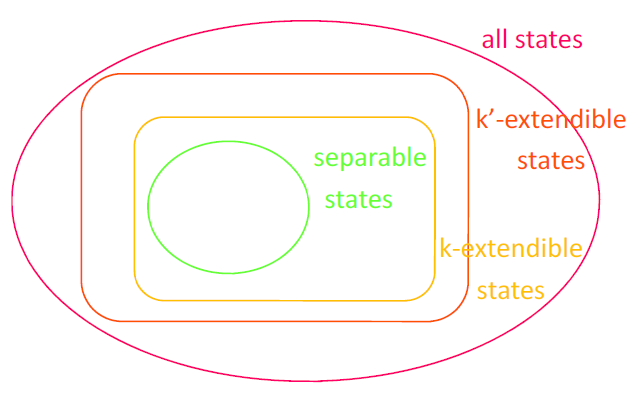}
\end{center}
\end{figure}

Actually, various more quantitative versions of Theorem \ref{th:k-extendibility-criterion} do exist, that put bounds on how far a $k$-extendible state can be from separable. Let us mention two quite different statements in that direction. The original result, appearing in \cite{CKMR}, establishes that a state on $\mathrm{A}\otimes\mathrm{B}$ which is $k$-extendible with respect to $\mathrm{B}$ is at distance at most $2d_{\B}^2/k$, in $1$-norm, from the set of separable states. It is a direct consequence of one of the quantitative versions of the quantum finite De Finetti theorem. A more recent result, proved essentially in \cite{BCY} and improved in \cite{BH}, stipulates that such a state is at distance at most $\sqrt{2\ln d_{\A}/k}$, in $\mathbf{LOCC}^{\rightarrow}$-norm, from the set of separable states (see e.g.~\cite{AL} for a precise definition of the operational one-way-LOCC norm). It relies on the observation that a $k$-extendible state has a small squashed entanglement, and therefore cannot be distinguished well from a separable state by local observers. The main problem of such estimates is that they become non-trivial only when $k\gg d_{\B}$ or $k\gg \ln d_{\A}$. So in the case where $d_{\A},d_{\B}$ are ``big'', can anything interesting still be said for a ``not too big'' $k$? On the other hand, these bounds valid for any $k$-extendible state are known to be close from optimal (there are examples of $k$-extendible states whose closest separable state is at distance of order $d_{\B}/k$ in $1$-norm or of order $\sqrt{\ln d_{\A}/k}$ in $\mathbf{LOCC}^{\rightarrow}$-norm). Consequently, one may only hope to make stronger statements about average behaviours.

This is precisely the general question we address here, being especially interested in the case of high-dimensional bipartite quantum systems. We try and quantify in two distinct ways the typical efficiency of the $k$-extendibility criterion for separability in this asymptotic regime.

The first approach consists in estimating a specific size parameter (known as the \textit{mean width}) of the set of $k$-extendible states when the dimension of the underlying Hilbert space goes to infinity. Comparing the obtained value with the known asymptotic estimate for the mean width of the set of separable states then tells us how the sizes of these two sets of states scale with one another. The computation is carried out in Section \ref{section:w(k-ext)} (where all needed notions related to high-dimensional convex geometry are properly defined as well) and ends with the concluding Theorem \ref{th:w(k-ext)}, some technical parts being relegated to Appendix \ref{appendix:gaussian}. In Section \ref{section:w-comparison}, the result is commented and comparisons are made between the mean-widths of, on the one hand, $k$-extendible states, and on the other, separable or PPT states. Besides, a smaller upper bound is derived, in Section \ref{section:w(k-ext-ppt]} and its companion technical Appendix \ref{appendix:gaussian-gamma}, on the mean width of the set of $k$-extendible states whose extension is required to be PPT (precise definitions and motivations to look at this set of states appear there).

The second approach consists in looking at random mixed states which are obtained by partial tracing over an ancilla space a uniformly distributed pure state, and characterizing when these are, with overwhelming probability as the dimension of the system grows, $k$-extendible or not. Again, comparing the obtained result with the known one for separability provides some information on how powerful the $k$-extendibility test is to detect entanglement. Section \ref{section:random-states} introduces all required material regarding the considered model of random-induced states and one possible way of detecting their non-$k$-extendibility. The adopted strategy is next seen through in Section \ref{section:non-k-ext-witness}, relying on technical statements put in Appendix \ref{appendix:wishart}, and concludes as Theorem \ref{th:not-kext}. The determined environment dimension below which random-induced states are with high probability violating the $k$-extendibility criterion is then compared, in Section \ref{section:threshold-comparison}, with the previously established ones for violating other separability criteria, and for actually not being separable.

Finally, generalizations to the unbalanced case are stated in Section \ref{section:unbalanced} (Theorems \ref{th:unbalanced} and \ref{th:unbalanced'}), while Section \ref{section:miscellaneous} exposes miscellaneous concluding remarks and loose ends.

The reader may have a look at Table \ref{table:comparison} for a sample corollary of this study.

\begin{table}[h] \caption{Comparison of the average and typical case performance of the $k$-extendibility criterion with that of the PPT and realignment criteria}
\label{table:comparison}
\begin{center}
\renewcommand{\arraystretch}{1.7}
\begin{tabular}{| C{8cm} | C{2.5cm} | C{2.5cm} |}
\hline
\diagbox[width=8cm]{$k$-extendibility ``beats''}{from the point of view of} & mean width of the set & entanglement detection of random states\\
\hline
PPT & for $k\geq 11$ & for $k\geq 17$ \\
\hline
realignment & ? & for $k\geq 5$ \\
\hline
\end{tabular}
\end{center}
\end{table}

Appendix \ref{appendix:combinatorics} gathers a bunch of standard definitions and facts about the combinatorics of permutations and partitions which are necessary for our purposes. All employed notation on that matter are also introduced there. In Appendix \ref{appendix:wick}, the connection is made between computing moments of GUE or Wishart matrices and counting permutations having a certain genus. These general observations play a key role in the moments' derivations of Appendices \ref{appendix:gaussian}, \ref{appendix:gaussian-gamma} and \ref{appendix:wishart}, which are, as for them, specifically the ones that we need to obtain our various statements. To get tractable expressions, though, a formula relating the number of cycles in some specific permutations is additionally required, whose proof is detailed in Appendix \ref{appendix:technical}. Appendix \ref{appendix:geodesics}, finally, is devoted to establishing the last crucial ingredient in most of our reasonings, namely bounding the number of non-geodesic permutations (in terms of the number of geodesic ones) in some particular instances which are of interest to us. Aside, Appendix \ref{appendix:convergence} is dedicated to proving more precise results than the ones which are strictly needed on the convergence of the studied random matrix ensembles, and in generalizing the developed method to establish the asymptotic freeness of certain gaussian random matrices.

\subsection*{Notation} For any Hilbert space $\mathrm{H}\equiv\C^n$, we shall denote by $\mathcal{H}(\mathrm{H})\equiv\mathcal{H}(n)$ the set of Hermitian operators on $\mathrm{H}$, and by $\mathcal{H}_+(\mathrm{H})\equiv\mathcal{H}_+(n)$ the subset of positive operators on $\mathrm{H}$. For each $p\in\N$, we define $\|\cdot\|_p$ as the Schatten $p$-norm on $\mathcal{H}(\mathrm{H})$, i.e. $\|\cdot\|_p=\left(\tr\left[|\cdot|^p\right]\right)^{1/p}$, and $\|\cdot\|_{\infty}=\lim_{p\rightarrow+\infty}\|\cdot\|_p$. Particular instances of interest are the trace class norm $\|\cdot\|_1$, the Hilbert--Schmidt norm $\|\cdot\|_2$, and the operator norm $\|\cdot\|_{\infty}$. We shall also denote by $\mathcal{D}(\mathrm{H})\equiv\mathcal{D}(n)$ the set of states on $\mathrm{H}$ (positive and trace $1$ operators).

We will in fact mostly consider the case where $\mathrm{H}=\mathrm{A}\otimes\mathrm{B}\equiv\C^d\otimes\C^d$ is a (balanced) bipartite Hilbert space. And we introduce the additional notations $\mathcal{S}(\mathrm{A}:\mathrm{B})\equiv\mathcal{S}(d\times d)$ for the set of separable states on $\mathrm{H}$, $\mathcal{P}(\mathrm{A}:\mathrm{B})\equiv\mathcal{P}(d\times d)$ for the set of PPT states on $\mathrm{H}$ (in both cases in the cut $\mathrm{A}:\mathrm{B}$), and for each $k\in\N$, $\mathcal{E}_k(\mathrm{A}:\mathrm{B})\equiv\mathcal{E}_k(d\times d)$ for the set of $k$-extendible states on $\mathrm{H}$ (in the cut $\mathrm{A}:\mathrm{B}$ and with respect to $\mathrm{B}$).

\subsection*{Preliminary technical lemma}

It will be essential for us in the sequel to express in a more tractable way the quantity $\sup_{\sigma\in\mathcal{E}_k(\A{:}\B)}\mathrm{Tr}(M\sigma)$, for any given Hermitian $M$ on $\A\otimes\B$. Such amenable expression is provided by Lemma \ref{lemma:sup-k-ext} below.

\begin{lemma} \label{lemma:sup-k-ext}
Let $k\in\N$. For any $M_{\A\B}\in\mathcal{H}(\mathrm{A}\otimes \mathrm{B})$, we have
\[ \underset{\sigma_{\A\B}\in\mathcal{E}_k(\A{:}\B)}{\sup}\ \tr\big(M_{\A\B}\sigma_{\A\B}\big) = \left\|\frac{1}{k}\underset{j=1}{\overset{k}{\sum}}\widetilde{M}_{\A\B^k}(j)\right\|_{\infty}, \]
where for each $1\leq j\leq k$, denoting by $\Id_{\widehat{\B}^k_j}$ the identity on $\B_1\otimes\cdots\otimes \B_{j-1}\otimes \B_{j+1}\otimes\cdots \otimes \B_k$, we defined $\widetilde{M}_{\A\B^k}(j)=M_{\A\B_j}\otimes\Id_{\widehat{\B}^k_j}$.
\end{lemma}

Before proving Lemma \ref{lemma:sup-k-ext}, let us introduce once and for all the following notation, which we shall later use on several occasions: for any $M_{\A\B^k}\in\mathcal{H}(\mathrm{A}\otimes \mathrm{B}^{\otimes k})$, we define its symmetrisation with respect to $\mathrm{B}^{\otimes k}$ as
\[ \Sym_{\A{:}\B^k}(M_{\A\B^k}) = \frac{1}{k!}\underset{\pi\in\mathfrak{S}(k)}{\sum}\big(\Id_\A\otimes U(\pi)_{\B^k}\big) M_{\A\B^k}\big(\Id_\A\otimes U(\pi)_{\B^k}\big)^{\dagger}, \]
where for each permutation $\pi\in\mathfrak{S}(k)$, $U(\pi)_{\B^k}$ denotes the associated permutation unitary on $\mathrm{B}^{\otimes k}$ (see e.g.~\cite{Harrow} for further details).

\begin{proof}
By definition, the condition $\sigma_{\A\B}\in\mathcal{E}_k(\A{:}\B)$ is equivalent to the condition $\sigma_{\A\B}=\tr_{\B^{k-1}}\Sym_{\A{:}\B^k}(\sigma_{\A\B^k})$ for some $\sigma_{\A\B^k}\in\mathcal{D}(\mathrm{A}\otimes \mathrm{B}^{\otimes k})$. Hence, for any $M_{\A\B}\in\mathcal{H}(\mathrm{A}\otimes \mathrm{B})$, we have
\begin{align*}
\underset{\sigma_{\A\B}\in\mathcal{E}_k(\A{:}\B)}{\sup}\ \tr_{\A\B}\big[M_{\A\B}\sigma_{\A\B}\big] &=\underset{\sigma_{\A\B^k}\in\mathcal{D}(\mathrm{A}\otimes\mathrm{B}^{\otimes k})}{\sup}\ \tr_{\A\B} \big[M_{\A\B} \tr_{\B^{k-1}}\Sym_{\A{:}\B^k}(\sigma_{\A\B^k})\big]\\
&= \underset{\sigma_{\A\B^k}\in\mathcal{D}(\mathrm{A}\otimes\mathrm{B}^{\otimes k})}{\sup}\ \tr_{\A\B^{k}} \big[\big(M_{\A\B}\otimes\Id_{\B^{k-1}}\big) \Sym_{\A{:}\B^k}(\sigma_{\A\B^k})\big]\\
&= \underset{\sigma_{\A\B^k}\in\mathcal{D}(\mathrm{A}\otimes\mathrm{B}^{\otimes k})}{\sup}\ \tr_{\A\B^{k}} \big[ \Sym_{\A{:}\B^k}\big(M_{\A\B}\otimes\Id_{\B^{k-1}}\big)\sigma_{\A\B^k}\big]\\
&= \left\|\Sym_{\A{:}\B^k} \big(M_{\A\B}\otimes\Id_{\B^{k-1}}\big)\right\|_{\infty}.
\end{align*}
Now, for each $\pi\in\mathfrak{S}(k)$, $\big(\Id_\A\otimes U(\pi)_{\B^k}\big) \big(M_{\A\B}\otimes\Id_{\B^{k-1}}\big) \big(\Id_\A\otimes U(\pi)_{\B^k}\big)^{\dagger}= M_{\A\B_{\pi(1)}}\otimes \Id_{\widehat{\B}^k_{\pi(1)}}$. Therefore, grouping together, for each $1\leq j\leq k$, the permutations $\pi\in\mathfrak{S}(k)$ such that $\pi(1)=j$, we get
\[ \Sym_{\A{:}\B^k} \big(M_{\A\B}\otimes\Id_{\B^{k-1}}\big)=\frac{1}{k}\sum_{j=1}^k M_{\A\B_j}\otimes\Id_{\widehat{\B}^k_j}, \]
and hence the advertised result.
\end{proof}

\section{Mean width of the set of $k$-extendible states for ``small'' $k$}
\label{section:w(k-ext)}

\subsection{Preliminaries on convex geometry}

Let us introduce a few notions coming from classical convex geometry which we shall need in the sequel. For any $K\subset\mathcal{H}(\C^n)$ and any $M\in\mathcal{H}(\C^n)$ having unit Hilbert--Schmidt norm, we define the width of $K$ in the direction $M$ as
\[ w(K,M) = \underset{\Delta\in K}{\sup}\mathrm{Tr}(M\Delta) . \]
The mean width of $K$ is then defined as the average of $w(K,\cdot)$ over the whole Hilbert--Schmidt unit sphere $S_{HS}(\C^n)$ of $\mathcal{H}(\C^n)$ (equipped with the Haar measure $\sigma$) i.e.
\[ w(K)= \int_{M\in S_{HS}(\C^n)}w(K,M)\mathrm{d}\sigma(M) = \int_{M\in S_{HS}(\C^n)} \left[\underset{\Delta\in K}{\sup}\mathrm{Tr}(M\Delta)\right]\mathrm{d}\sigma(M). \]
This average width $w$ is an interesting size parameter, on its own, but also because it is related to other important geometric quantities, such as e.g.~the volume radius $\mathrm{vrad}$, which is defined as the radius of the Euclidean ball having same volume (i.e.~Lebesgue measure). For instance, we have for any convex body $K$ the Urysohn inequality $w(K)\geq \mathrm{vrad}(K)$, and for most of the convex bodies $K$ we shall be considering a ``reverse'' Urysohn inequality $w(K)\leq \mu\,\mathrm{vrad}(K)$ for some $\mu\geq 1$. These connections and the precise formulation of these convex geometry results are exemplified in Section \ref{section:w-comparison}.

In order to compute the quantity $w(K)$, it is often convenient to re-express it as a Gaussian rather than spherical averaging. We thus denote by $GUE(n)$ the Gaussian Unitary Ensemble on $\C^n$, which is the standard Gaussian vector in $\mathcal{H}(\C^n)$ (equivalently, $G\sim GUE(n)$ if $G=(H+H^{\dagger})/\sqrt{2}$ with $H$ a $n\times n$ matrix having independent complex normal entries). And we define the Gaussian mean width of $K$ as
\[ w_G(K) = \E_{G\sim GUE(n)}\left[\underset{\Delta\in K}{\sup}\mathrm{Tr}(G\Delta)\right] .\]
Just observing that for $G\sim GUE(n)$, $G/\|G\|_2$ is uniformly distributed over $S_{HS}(\C^n)$, and $G/\|G\|_2$, $\|G\|_2$ are independent random variables, we get that the link between both quantities is, setting $\gamma(n)=\E_{G\sim GUE(n)}\|G\|_2$, which is known to satisfy $\gamma(n)\sim_{n\rightarrow +\infty}n$ (see e.g.~\cite{AGZ}, Chapter 2, for a proof),
\begin{equation} \label{eq:w-wG} w(K)= \frac{1}{\gamma(n)}w_G(K). \end{equation}

\begin{remark} All the sets $K$ that we will consider in the sequel will actually be subsets of $\mathcal{D}(\C^n)$, hence living in the hyperplane of $\mathcal{H}(\C^n)$ composed of trace $1$ elements, i.e.~in a space of real dimension $n^2-1$, rather than $n^2$. It would thus seem more natural to define their mean width $w(K)$ as an average width over a $n^2-2$, rather than $n^2-1$, dimensional Euclidean unit sphere. The Gaussian mean width $w_G(K)$, on the other hand, is an intrinsic notion that does not depend on the ambient dimension (because marginals of standard Gaussian vectors are themselves standard Gaussian vectors). As a consequence, we see from equation \eqref{eq:w-wG} that computing the mean width of $K$ as if it was a $n^2$ dimensional set is asymptotically equivalent to computing it taking into account that it is in fact a $n^2-1$ dimensional set. We may therefore serenely forget about this issue.
\end{remark}


Our aim is now to estimate, for any fixed $k\in\N$, the mean width of the set of $k$-extendible states on $\A\otimes\B$ when $\A\equiv\B\equiv\C^d$ and $d\rightarrow+\infty$. By the definitions above, we have
\[ w\big(\mathcal{E}_k(\A{:}\B)\big) 
= \frac{1}{\gamma(d^2)} \mathbf{E}_{G_{\A\B}\sim GUE(d^2)} \left[ \underset{\sigma_{\A\B}\in \mathcal{E}_k(\A{:}\B)}{\sup}\mathrm{Tr}\big(G_{\A\B}\sigma_{\A\B}\big) \right]. \]
Using the result of Lemma \ref{lemma:sup-k-ext}, and the notation introduced there, we thus get,
\begin{equation} \label{eq:w(k-ext)-definition} w\big(\mathcal{E}_k(\A{:}\B)\big) = \frac{1}{\gamma(d^2)} \mathbf{E}_{G_{\A\B}\sim GUE(d^2)} \left\| \frac{1}{k} \underset{j=1}{\overset{k}{\sum}}\widetilde{G}_{\A\B^k}(j) \right\|_{\infty}. \end{equation}

\subsection{An operator-norm estimate}

As justified above, to obtain the mean width of the set of $k$-extendible states on $\C^d\otimes\C^d$, what we need is to compute the average operator-norm $\mathbf{E} \left\|\sum_{j=1}^k\widetilde{G}_{\A\B^k}(j)\right\|_{\infty}$, for $G$ a GUE matrix on $\C^d\otimes\C^d$. We will show that the following asymptotic estimate holds.

\begin{proposition} \label{prop:gaussian-infty}
Fix $k\in\N$. Then,
\[ \mathbf{E}_{G_{\A\B}\sim GUE(d^2)} \left\| \underset{j=1}{\overset{k}{\sum}}\widetilde{G}_{\A\B^k}(j) \right\|_{\infty} \underset{d\rightarrow+\infty}{\sim} 2\sqrt{k}d .\]
\end{proposition}

As a preliminary step towards estimating the sup-norm $\mathbf{E} \left\|\sum_{j=1}^k\widetilde{G}_{\A\B^k}(j)\right\|_{\infty}$, we will look at the $2p$-order moments $\mathbf{E}\, \mathrm{Tr}\left[\left(\sum_{j=1}^k\widetilde{G}_{\A\B^k}(j)\right)^{2p}\right]$, $p\in\mathbf{N}$, and show that they can be expressed in terms of the $2p$-order moments of a centered semicircular distribution of appropriate parameter.

So let us recall first a few required definitions. For any $\sigma>0$, we shall denote by $\mu_{SC(\sigma^2)}$ the centered semicircular distribution of variance parameter $\sigma^2$, whose density is given by
\[ \mathrm{d}\mu_{SC(\sigma^2)}(x)=\frac{1}{2\pi\sigma^2}\sqrt{4\sigma^2-x^2} \mathbf{1}_{[-2\sigma,2\sigma]}(x)\mathrm{d}x .\]
We shall also denote, for each $p\in\N$, by $\mathrm{M}_{SC(\sigma^2)}^{(p)}$ its $p$-order moment, i.e.~ $\mathrm{M}_{SC(\sigma^2)}^{(p)}=\int_{-\infty}^{+\infty}x^{p}\mathrm{d}\mu_{SC(\sigma^2)}(x)$. It is well-known that
\[ \forall\ p\in\N,\ \mathrm{M}_{SC(\sigma^2)}^{(2p-1)}=0\ \ \text{and}\ \ \mathrm{M}_{SC(\sigma^2)}^{(2p)}=\sigma^{2p}\,\mathrm{Cat}_p, \]
where $\mathrm{Cat}_p$ is the $p^{\text{th}}$ Catalan number defined in Lemma \ref{lemma:catalan}.

\begin{proposition} \label{prop:gaussian-p} Fix $k\in\N$. Then, when $d\rightarrow+\infty$, the random matrix $\left(\sum_{j=1}^k\widetilde{G}_{\A\B^k}(j)\right)/d$ converges in moments towards a centered semicircular distribution of parameter $k$. Equivalently, this means that, for any $p\in\N$,
\[ \mathbf{E}_{G_{\A\B}\sim GUE(d^2)}\, \mathrm{Tr}\left[ \left( \underset{j=1}{\overset{k}{\sum}}\widetilde{G}_{\A\B^k}(j) \right)^{2p-1} \right]=0\ \ \text{and}\ \ \mathbf{E}_{G_{\A\B}\sim GUE(d^2)}\, \mathrm{Tr}\left[ \left( \underset{j=1}{\overset{k}{\sum}}\widetilde{G}_{\A\B^k}(j) \right)^{2p} \right] \underset{d\rightarrow+\infty}{\sim} \mathrm{M}_{SC(k)}^{(2p)}d^{2p+k+1} .\]
\end{proposition}

\begin{remark}
Stronger convergence results than the one established in Proposition \ref{prop:gaussian-p} may in fact be proved, as discussed in Appendix \ref{appendix:convergence}.
\end{remark}

\begin{proof} [Proof of Proposition \ref{prop:gaussian-p}]
Let $p\in\N$. Computing the value of the $2p$-order moment $\mathbf{E}\, \mathrm{Tr}\left[\left(\sum_{j=1}^k\widetilde{G}_{\A\B^k}(j)\right)^{2p}\right]$ may be done using the Gaussian Wick formula (see Lemma \ref{lemma:wick} for the statement and Appendix \ref{appendix:wick-gaussian} for a succinct summary of how to derive moments of GUE matrices from it). In our case, what we get by the computations carried out in Appendix \ref{appendix:gaussian} and summarized in Proposition \ref{prop:gaussian-p-preliminary} is that, for any $d\in\N$, denoting by $\sharp(\cdot)$ the number of cycles in a permutation,
\begin{align}
\mathbf{E}\, \mathrm{Tr}\left[\left(\sum_{j=1}^k\widetilde{G}_{\A\B^k}(j)\right)^{2p}\right] & =  \sum_{f:[2p]\rightarrow[k]} \mathbf{E}\, \mathrm{Tr}\left[\underset{i=1}{\overset{2p}{\overrightarrow{\prod}}}\widetilde{G}_{\A\B^k}(f(i))\right]\\
& =  \underset{f:[2p]\rightarrow[k]}{\sum} \underset{\lambda\in\mathfrak{P}^{(2)}(2p)}{\sum} d^{\sharp(\gamma^{-1}\lambda)+\sharp(\gamma_f^{-1}\lambda) + k - |\im(f)|},
\label{eq:p-moment-gaussian} \end{align}
where we defined on $\{1,\ldots,2p\}$, $\mathfrak{P}^{(2)}(2p)$ as the set of pair partitions, $\gamma=(2p\,\ldots\,1)$ as the canonical full cycle, and for each $f:[2p]\rightarrow[k]$, $\gamma_f=\gamma_{f=1}\cdots\gamma_{f=k}$ as the product of the canonical full cycles on each of the level sets of $f$.

We now have to understand which $\lambda\in\mathfrak{P}^{(2)}(2p)$ and $f:[2p]\rightarrow[k]$ contribute to the dominating term in the moment expansion \eqref{eq:p-moment-gaussian}, i.e.~are such that the quantity $\sharp(\gamma^{-1}\lambda)+\sharp(\gamma_f^{-1}\lambda) + k - |\im(f)|$ is maximal.

First of all, for any $\lambda\in\mathfrak{P}^{(2)}(2p)$, we have
\begin{equation} \label{eq:lambda-lambda_f'}
\sharp(\lambda)+\sharp(\gamma^{-1}\lambda) = 4p - \left(|\lambda|+|\gamma^{-1}\lambda|\right) \leq 4p-|\gamma^{-1}| = 4p-(2p-1)=2p+1,
\end{equation}
where the first equality is by Lemma \ref{lemma:cycles-transpositions}, while the second inequality is by equation \eqref{eq:geodesic'} in Lemma \ref{lemma:distance} and is an equality if and only if the pair-partition $\lambda$ is non-crossing. Next, for any $\lambda\in\mathfrak{P}^{(2)}(2p)$ and $f:[2p]\rightarrow[k]$, we have
\begin{equation} \label{eq:lambda-lambda_f''}
\sharp(\lambda)+\sharp(\gamma_f^{-1}\lambda) = 4p-\left(|\lambda|+|\gamma_f^{-1}\lambda|\right) \leq 4p - |\gamma_f^{-1}|= 4p - \left(2p-|\im(f)|\right) =2p+|\im(f)|,
\end{equation}
where the first equality is again by Lemma \ref{lemma:cycles-transpositions}, while the second inequality is by equation \eqref{eq:geodesic''} in Lemma \ref{lemma:distance} and is an equality if and only if the pair-partition $\lambda$ is non-crossing and is finer than the partition of $\{1,\ldots,2p\}$ induced by $\gamma_f$ (i.e.~$f$ takes the same value on elements belonging to the same pair-block of $\lambda$).

Putting equations \eqref{eq:lambda-lambda_f'} and \eqref{eq:lambda-lambda_f''} together, we get that for any $\lambda\in\mathfrak{P}^{(2)}(2p)$ and $f:[2p]\rightarrow[k]$ (just keeping in mind that necessarily $\sharp(\lambda)=p$),
\begin{equation} \label{eq:lambda-lambda_f} \sharp(\gamma^{-1}\lambda)+\sharp(\gamma_f^{-1}\lambda) + k - |\im(f)|\leq 2p+k+1, \end{equation}
with equality if and only if $\lambda\in NC^{(2)}(2p)$ and $f\circ\lambda=f$. Since it is well-known that there are $\mathrm{Cat}_p$ elements in $NC^{(2)}(2p)$, and for each of these there are $k^p$ functions which are constant on each of its $p$ pair-blocks, we indeed get the asymptotic estimate announced in Proposition \ref{prop:gaussian-p}, namely
\[ \mathbf{E}\, \mathrm{Tr}\left[ \left( \underset{j=1}{\overset{k}{\sum}}\widetilde{G}_{\A\B^k}(j) \right)^{2p} \right] \underset{d\rightarrow+\infty}{\sim} k^p\,\mathrm{Cat}_p\,d^{2p+k+1} =\mathrm{M}_{SC(k)}^{(2p)}d^{2p+k+1} . \qedhere \]
\end{proof}

\begin{proof} [Proof of Proposition \ref{prop:gaussian-infty}]
The convergence in moments stated in Proposition \ref{prop:gaussian-p} implies 
that, asymptotically, the matrix $\left(\sum_{j=1}^k\widetilde{G}_{\A\B^k}(j)\right)/d$ has a smallest eigenvalue and a largest eigenvalue which are, on average, at most the lower-edge and at least the upper-edge of the support of $\mu_{SC(k)}$, i.e.~$-2\sqrt{k}$ and $2\sqrt{k}$. Indeed, convergence of all moments of the empirical spectral distribution of $\left(\sum_{j=1}^k\widetilde{G}_{\A\B^k}(j)\right)/d$ entails convergence of all polynomial functions, and therefore of all continuous functions with bounded support, when integrated against it. And this in turn entails (when applied to continuous functions with support strictly included in $[-2\sqrt{k},2\sqrt{k}]$) that the extreme eigenvalues of $\left(\sum_{j=1}^k\widetilde{G}_{\A\B^k}(j)\right)/d$ cannot be, on average, strictly bigger than $-2\sqrt{k}$ or smaller $2\sqrt{k}$. The reader is referred to \cite{AGZ}, Chapter 2, for all the technical details of the argument. Hence in other words, Proposition \ref{prop:gaussian-p} guarantees that there exist positive constants $c_d\rightarrow_{d\rightarrow+\infty}1$ such that
\begin{equation} \label{eq:lower-bound-gaussian} \mathbf{E} \left\| \underset{j=1}{\overset{k}{\sum}}\widetilde{G}_{\A\B^k}(j) \right\|_{\infty} 
\geq c_d\, 2\sqrt{k}d .\end{equation}

In the opposite direction, Proposition \ref{prop:gaussian-p} only guarantees that the matrix $\left(\sum_{j=1}^k\widetilde{G}_{\A\B^k}(j)\right)/d$ asymptotically has, on average, no strictly positive fraction of eigenvalues strictly below $-2\sqrt{k}$ or above $2\sqrt{k}$. So to show that the reverse inequality to \eqref{eq:lower-bound-gaussian} holds too, a little more care is required. Indeed, to say it roughly, we have to make sure that in the moment's expression \eqref{eq:p-moment-gaussian}, the permutations contributing to the non-dominating terms (in $d$) are not too numerous.

For $d\in\N$ fixed, it holds thanks to Jensen's inequality and monotonicity of Schatten norms that
\begin{equation} \label{eq:infty-p-gaussian} \forall\ p\in\N,\ \mathbf{E} \left\|\sum_{j=1}^k\widetilde{G}_{\A\B^k}(j)\right\|_{\infty} \leq \left( \mathbf{E} \left\|\sum_{j=1}^k\widetilde{G}_{\A\B^k}(j)\right\|_{\infty}^{2p}\right)^{1/2p} \leq \left(\mathbf{E}\, \mathrm{Tr}\left[\left(\sum_{j=1}^k\widetilde{G}_{\A\B^k}(j)\right)^{2p}\right] \right)^{1/2p} .\end{equation}
So let us fix $d\in\N$ and $p\in\N$, and rewrite \eqref{eq:p-moment-gaussian} explicitly as an expansion in powers of $d$, keeping in the sum the permutations not saturating equation \eqref{eq:lambda-lambda_f}. Being cautious only with the permutations not saturating equation \eqref{eq:lambda-lambda_f''}, and not with those not saturating equation \eqref{eq:lambda-lambda_f'}, we get
\begin{equation} \label{eq:momentsG-defect}
\mathbf{E}\, \mathrm{Tr}\left[\left(\sum_{j=1}^k\widetilde{G}_{\A\B^k}(j)\right)^{2p}\right] \leq \left(\underset{f:[2p]\rightarrow[k]}{\sum} \sum_{\delta=0}^{\lfloor(p+k)/2\rfloor} \left|\mathfrak{P}^{(2)}_{f,\delta}(2p)\right| d^{-2\delta}\right) d^{2p+k+1},
\end{equation}
where we defined, for each $f:[2p]\rightarrow[k]$ and each $0\leq\delta\leq \lfloor(p+k)/2\rfloor$,
\[ \mathfrak{P}^{(2)}_{f,\delta}(2p)=\left\{ \lambda\in\mathfrak{P}^{(2)}(2p) \st \sharp(\gamma_f^{-1}\lambda) =p+|\im(f)|-2\delta \right\}. \]
In words, $\mathfrak{P}^{(2)}_{f,\delta}(2p)$ is nothing else than the set of permutations which have a defect $2\delta$ from lying on the geodesics between the identity and the product of the canonical full cycles on each of the level sets of $f$. This justifies in particular \textit{a posteriori} why the summation in \eqref{eq:momentsG-defect} is only over even defects (see the parity argument in Lemma \ref{lemma:distance-general}).

Now, by Lemma \ref{lemma:number-functions,pairings-defect}, we know that, if $0\leq\delta\leq\lfloor p/2\rfloor$, then
\[ \left|\left\{ (f,\lambda) \st \lambda\in\mathfrak{P}^{(2)}_{f,\delta}(2p) \right\}\right| \leq k^p\,\mathrm{Cat}_p \times\left(\frac{kp^2}{2}\right)^{2\delta}. \]
And if $\lceil p/2\rceil\leq\delta\leq \lfloor (p+k)/2\rfloor$, then trivially
\[ \left|\left\{ (f,\lambda) \st \lambda\in\mathfrak{P}^{(2)}_{f,\delta}(2p) \right\}\right| \leq k^{2p}\,\frac{(2p)!}{2^pp!} \leq k^p\,\mathrm{Cat}_p\times \left(\frac{kp^2}{2}\right)^p. \]
Putting everything together, we therefore get,
\[ \mathbf{E}\, \mathrm{Tr}\left[\left(\sum_{j=1}^k\widetilde{G}_{\A\B^k}(j)\right)^{2p}\right] \leq k^p\,\mathrm{Cat}_p \left(1+\sum_{\delta=1}^{\lfloor p/2\rfloor}\left(\frac{kp^2}{2d}\right)^{2\delta}+\frac{k}{2}\left(\frac{kp^2}{2d}\right)^p\right)d^{2p+k+1}. \]
Yet, $\max\left\{\left(kp^2/2d\right)^{2\delta} \st 1\leq\delta\leq \lceil p/2\rceil\right\}$ is attained for $\delta=1$, provided $p\leq(2d/k)^{1/2}$. So if such is the case,
\[ \sum_{\delta=1}^{\lfloor p/2\rfloor}\left(\frac{kp^2}{2d}\right)^{2\delta}+\frac{k}{2}\left(\frac{kp^2}{2d}\right)^p \leq \frac{p+k}{2} \frac{k^2p^4}{4d^2} \leq \frac{k^2p^5}{4d^2}, \]
where the last inequality holds as long as $p\geq k$. And hence, under all the previous assumptions,
\[ \mathbf{E}\, \mathrm{Tr}\left[\left(\sum_{j=1}^k\widetilde{G}_{\A\B^k}(j)\right)^{2p}\right] \leq\, \mathrm{M}_{SC(k)}^{(2p)}\left(1 +\frac{k^2p^5}{4d^2}\right) d^{2p+k+1}. \]

So set $p_d=(2d/k)^{(2-\epsilon)/5}$ for some $0<\epsilon<1$ (which is indeed smaller than $(2d/k)^{1/2}$ and bigger than $k$ for $d$ big enough, in particular bigger than $k^{7/2}/2$). And using inequality \eqref{eq:infty-p-gaussian} in the special case $p=p_d$, we eventually get
\begin{equation} \label{eq:upper-bound-gaussian} \mathbf{E} \left\|\sum_{j=1}^k\widetilde{G}_{\A\B^k}(j)\right\|_{\infty}  \leq \left(\mathrm{M}_{SC(k)}^{(2p_d)} \left(1 + \frac{k^2p_d^5}{4d^2}\right)\right)^{1/2p_d} d^{1+(k+1)/2p_d} \underset{d\rightarrow+\infty}{\sim} 2\sqrt{k}d. \end{equation}
Combining the lower bound in equation \eqref{eq:lower-bound-gaussian} and the upper bound in equation \eqref{eq:upper-bound-gaussian} yields Proposition \ref{prop:gaussian-infty}.
\end{proof}

\subsection{Conclusion}

Combining Proposition \ref{prop:gaussian-infty} with equation \eqref{eq:w(k-ext)-definition}, we straightforwardly obtain the estimate we were looking for, which is stated in Theorem \ref{th:w(k-ext)} below.

\begin{theorem} \label{th:w(k-ext)}
Let $k\in\N$. The mean width of the set of $k$-extendible states on $\C^d\otimes\C^d$ satisfies
\[ w\big(\mathcal{E}_k(\C^d{:}\C^d)\big) \underset{d\rightarrow +\infty}{\sim} \frac{2}{\sqrt{k}}\frac{1}{d} .\]
\end{theorem}

\section{Discussion and comparison with the mean width of the set of PPT states}
\label{section:w-comparison}

It was shown in \cite{AS} that the mean width of the set of separable states on $\C^d\otimes\C^d$ is of order $1/d^{3/2}$. And we just showed in Theorem \ref{th:w(k-ext)} that, for $k\in\N$ fixed, the mean width of the set of $k$-extendible states on $\C^d\otimes\C^d$ is of order $1/d$, so that, for $d$ large,
\[ w\big(\mathcal{S}(\C^d{:}\C^d)\big) \ll w\big(\mathcal{E}_k(\C^d{:}\C^d)\big). \]
This result is not surprising: it just means that, when $d$ grows, if $k$ does not grow in some way too, then the set of $k$-extendible states becomes an increasingly poor approximation of the set of separable states on $\C^d\otimes\C^d$. There had been several evidences, already, in that direction, with examples of highly-extendible, though entangled, states (see e.g.~\cite{BCY} and \cite{NOP1}).

It is well-known that the exact same feature is actually exhibited by the set of PPT states on $\C^d\otimes\C^d$, whose mean width is of order $1/d$ too. Let us be more precise.

\begin{proposition} \label{prop:w(ppt)}
There exist positive constants $c_d,C_d\rightarrow_{d\rightarrow+\infty}1$ such that the mean width of the set of PPT states on $\C^d\otimes\C^d$ satisfies
\[ c_d\,\frac{e^{-1/2}}{d} \leq w\big(\mathcal{P}(\C^d{:}\C^d)\big) \leq C_d\,\frac{2}{d}. \]
\end{proposition}

\begin{proof} Proposition \ref{prop:w(ppt)} was basically established in \cite{AS}, but not stated in this exact way and with these exact constants, so we briefly recall the argument here for the sake of completeness.

To get the asymptotic upper bound, we just use
\[ w\big(\mathcal{P}(\C^d{:}\C^d)\big) \leq w\big(\mathcal{D}(\C^d\otimes\C^d)\big) \underset{d\rightarrow+\infty}{\sim} \frac{2}{d}. \]
The last equivalence is a consequence of Wigner's semicircle law (see e.g.~\cite{AGZ}, Chapter 2, for a proof) from which it follows that
\[ w_G\big(\mathcal{D}(\C^d\otimes\C^d)\big) = \E_{G\sim GUE(d^2)} \underset{\sigma\in\mathcal{D}(\C^d\otimes\C^d)}{\sup}\mathrm{Tr}(G\sigma) = \E_{G\sim GUE(d^2)} \|G\|_{\infty} \underset{d\rightarrow+\infty}{\sim}2d. \]

To get the asymptotic lower-bound, we will make use of two results from classical convex geometry. Before stating them, we need one more definition: For any convex body $K$, we denote by $\vrad(K)$ its volume radius, which is defined as the radius of the Euclidean ball having the same volume (i.e.~Lebesgue measure) as $K$.\\
$\bullet$ Urysohn inequality (see e.g.~\cite{Pisier}, Corollary 1.4): For any convex body $K$,
\begin{equation} \label{eq:urysohn} w(K)\geq\vrad(K). \end{equation}
$\bullet$ Milman-Pajor inequality (see \cite{MP}, Corollary 3): For any convex bodies $K,L$ having the same center of gravity,
\begin{equation} \label{eq:milman-pajor} \vrad(K\cap L)\vrad(K-L)\geq\vrad(K)\vrad(L). \end{equation}
Combining equations \eqref{eq:urysohn} and \eqref{eq:milman-pajor}, we get that if $K,L$ are convex bodies having the same center of gravity, then
\[ \vrad(K\cap L)\geq\frac{\vrad(K)\vrad(L)}{\vrad(K-L)} \geq\frac{\vrad(K)\vrad(L)}{w(K-L)} \geq\frac{\vrad(K)\vrad(L)}{w(K)+w(L)}. \]
In our case, denoting by $\Gamma$ the partial transposition, we have $\mathcal{P}(d\times d)=\mathcal{D}(d^2)\cap\mathcal{D}(d^2)^{\Gamma}$, with $\mathcal{D}(d^2)$ and $\mathcal{D}(d^2)^{\Gamma}$ both having the maximally mixed state $\Id/d^2$ as center of gravity. Hence,
\[ w\big(\mathcal{P}(d\times d)\big) \geq \vrad\big(\mathcal{P}(d\times d)\big) \geq \frac{\vrad\big(\mathcal{D}(d^2)\big)^2}{2w\big(\mathcal{D}(d^2)\big)}, \]
the first inequality being by the Urysohn inequality, and the second being by the Milman-Pajor inequality, after noticing that $\vrad\big(\mathcal{D}(d^2)^{\Gamma}\big)=\vrad\big(\mathcal{D}(d^2)\big)$ and $w\big(\mathcal{D}(d^2)^{\Gamma}\big)=w\big(\mathcal{D}(d^2)\big)$. Now, we just argued that $w\big(\mathcal{D}(d^2)\big)\sim_{d\rightarrow+\infty} 2/d$, while it was shown in \cite{SZ} that $\vrad\big(\mathcal{D}(d^2)\big)\sim_{d\rightarrow+\infty} e^{-1/4}/d$. Therefore,
\[ w\big(\mathcal{P}(d\times d)\big) \geq \frac{\vrad\big(\mathcal{D}(d^2)\big)^2}{2w\big(\mathcal{D}(d^2)\big)} \underset{d\rightarrow+\infty}{\sim} \frac{e^{-1/2}}{d}. \qedhere \]
\end{proof}

As a straightforward consequence of Theorem \ref{th:w(k-ext)} and Proposition \ref{prop:w(ppt)}, we have, roughly speaking, that for $k\geq 11$, the set of $k$-extendible states becomes asymptotically a ``better'' approximation of the set of separable states than the set of PPT states, on average. Indeed, if $k\geq 11$, then $2/\sqrt{k}<e^{-1/2}$, so that for $d$ large enough
\[ w\big(\mathcal{E}_k(\C^d{:}\C^d)\big) < w\big(\mathcal{P}(\C^d{:}\C^d)\big). \]

\section{Adding the PPT constraint on the extension}
\label{section:w(k-ext-ppt]}

The hierarchy of SDPs originally proposed in \cite{DPS} to detect entanglement was in fact slightly different from the one that would be derived from Theorem \ref{th:k-extendibility-criterion}. Indeed, for a given bipartite state $\rho_{\A\B}$, the $k^{\text{th}}$ test would here consist in looking for a symmetric extension $\rho_{\A\B^k}$ of $\rho_{\A\B}$, while in \cite{DPS} it was additionally imposed that this extension had to be PPT in any cut of the $k+1$ subsystems. This of course increased quite considerably the size of the SDP to be solved at each step, but with the hope that it would at the same time decrease dramatically the number of steps an entangled state would pass.

Another hierarchy of SDPs was later proposed in \cite{NOP1} and \cite{NOP2}, built on the exact same ideas as those in \cite{DPS}. It was noticed there that only demanding that the (Bose) symmetric extension of the state be PPT in one fixed (even) cut of the $k+1$ subsystems already implied a noticeable speed-up in the convergence of the algorithm. It therefore seems worth taking a closer look at the set of states arising from these constraints. The latter is properly defined as follows.

\begin{definition}
Let $k\in\N$. A state $\rho_{\A\B}$ on a bipartite Hilbert space $\mathrm{A}\otimes\mathrm{B}$ is $k$-PPT-extendible with respect to $\mathrm{B}$ if there exists a state $\rho_{\A\B^k}$ on $\mathrm{A}\otimes\mathrm{B}^{\otimes k}$ which is PPT in the cut $\mathrm{A}\otimes\mathrm{B}^{\otimes \lfloor k/2\rfloor}:\mathrm{B}^{\otimes \lceil k/2\rceil}$, invariant under any permutation of the $\mathrm{B}$ subsystems and such that $\rho_{\A\B}=\tr_{\B^{k-1}}\rho_{\A\B^k}$. We denote by $\mathcal{E}_k^{PPT}(\A{:}\B)$ the set of $k$-PPT-extendible states on $\mathrm{A}\otimes\mathrm{B}$ (in the cut $\A{:}\B$ and with respect to $\mathrm{B}$).
\end{definition}

\begin{theorem} \label{th:w(k-PPT-ext)}
Let $k\in\N$. There exist positive constants $C_d\rightarrow_{d\rightarrow+\infty}1$ such that the mean width of the set of $k$-PPT-extendible states on $\C^d\otimes\C^d$ satisfies
\[ w\big(\mathcal{E}_k^{PPT}(\C^d{:}\C^d)\big) \leq C_d\,\frac{\sqrt{2}}{\sqrt{k}}\frac{1}{d} .\]
\end{theorem}

\begin{proof}
Using the notation introduced in Lemma \ref{lemma:sup-k-ext}, we start from the simple observation that, for any $M_{\A\B}\in\mathcal{H}(\mathrm{A}\otimes \mathrm{B})$,
\begin{align*}
\underset{\sigma_{\A\B}\in\mathcal{E}_k^{PPT}(\A{:}\B)}{\sup}\ \tr\big[M_{\A\B}\sigma_{\A\B}\big]
&= \underset{\sigma_{\A\B^k}\in\mathcal{P}(\mathrm{A}\mathrm{B}^{\lfloor k/2\rfloor}:\mathrm{B}^{\lceil k/2\rceil})}{\sup}\ \tr \big[ \Sym_{\A{:}\B^k}\big(M_{\A\B}\otimes\Id_{\B^{k-1}}\big)\sigma_{\A\B^k}\big]\\
&\leq \min\left( \big\|\Sym_{\A{:}\B^k} \big(M_{\A\B}\otimes\Id_{\B^{k-1}}\big)\big\|_{\infty}, \big\|\left[\Sym_{\A{:}\B^k} \big(M_{\A\B}\otimes\Id_{\B^{k-1}}\big)\right]^{\Gamma}\big\|_{\infty} \right),
\end{align*}
where $\Gamma$ stands here for the partial transposition over the $\lceil k/2\rceil$ last $\mathrm{B}$ subsystems, so that in fact
\[ \left[\Sym_{\A{:}\B^k} \big(M_{\A\B}\otimes\Id_{\B^{k-1}}\big)\right]^{\Gamma} = \frac{1}{k} \left( \sum_{j=1}^{\lfloor k/2\rfloor} M_{\A\B_j}\otimes\Id_{\widehat{\B}^k_j} + \sum_{j=\lfloor k/2\rfloor +1}^k M_{\A\B_j}^{\Gamma}\otimes\Id_{\widehat{\B}^k_j} \right), \]
where $\Gamma$ now stands for the partial transposition over $\mathrm{B}$.

The upper bound in Theorem \ref{th:w(k-PPT-ext)} will thus be a direct consequence of the sup-norm estimate
\[ \mathbf{E}_{G_{\A\B}\sim GUE(d^2)} \left\| \sum_{j=1}^{k}\widetilde{G}_{\A\B^k}(j)^{\Gamma} \right\|_{\infty} \underset{d\rightarrow+\infty}{\sim} \sqrt{2}\sqrt{k}d .\]
The latter is proved in the exact same way as Proposition \ref{prop:gaussian-infty}, i.e.~by first showing that for any $p\in\N$,
\begin{equation} \label{eq:moment-Gamma} \mathbf{E}_{G_{\A\B}\sim GUE(d^2)}\, \mathrm{Tr}\left[ \left( \sum_{j=1}^{k}\widetilde{G}_{\A\B^k}(j)^{\Gamma} \right)^{2p} \right] \underset{d\rightarrow+\infty}{\sim} 2\mathrm{M}_{SC(k/2)}^{(2p)}d^{2p+k/2+1} ,\end{equation}
and second arguing that also $\mathbf{E}\left\| \sum_{j=1}^{k} \widetilde{G}_{\A\B^k}(j)^{\Gamma}\right\|_{\infty}\sim_{d\rightarrow+\infty}\lim_{p\rightarrow+\infty} \left( \mathbf{E}\, \mathrm{Tr}\left[ \left( \sum_{j=1}^{k} \widetilde{G}_{\A\B^k}(j)^{\Gamma}\right)^{2p} \right] \right)^{1/2p}$. This last step will be omitted here since the argument is very similar to the one appearing in the proof of Proposition \ref{prop:gaussian-infty}. Concerning the moment estimate \eqref{eq:moment-Gamma}, it is first of all proved in Appendix \ref{appendix:gaussian-gamma} that
\[ \E\, \mathrm{Tr}\left[\left(\sum_{j=1}^k\widetilde{G}_{\A\B^k}(j)^{\Gamma}\right)^{2p}\right]
 \underset{d\rightarrow+\infty}{\sim} \underset{\lambda\in\mathfrak{P}^{(2)}(2p)}{\sum}\, \underset{f:[2p]\rightarrow\left[\lfloor k/2\rfloor\right]\,\text{or}\,\left[\lceil k/2\rceil\right]}{\sum} d^{\sharp(\gamma^{-1}\lambda)+\sharp(\gamma_f^{-1}\lambda)+k-|\im(f)|}. \]
And by the same arguments as the in the proof of Proposition \ref{prop:gaussian-p}, we can then identify which $\lambda$ and $f$ actually contribute to the dominant order in the latter expression, yielding
\begin{align*}
\E \,\mathrm{Tr}\left[\left(\sum_{j=1}^k\widetilde{G}_{\A\B^k}(j)^{\Gamma}\right)^{2p}\right]
& \underset{d\rightarrow+\infty}{\sim} \underset{\lambda\in NC^{(2)}(2p)}{\sum} \left( \underset{\underset{f\circ\lambda=f}{f:[2p]\rightarrow\left[\lfloor k/2\rfloor\right]}}{\sum} d^{2p+\lfloor k/2\rfloor+1} + \underset{\underset{f\circ\lambda=f}{f:[2p]\rightarrow\left[\lceil k/2\rceil\right]}}{\sum} d^{2p+\lceil k/2\rceil+1} \right) \\
& \underset{d\rightarrow+\infty}{\sim} \mathrm{Cat}_p\left(\lfloor k/2\rfloor^pd^{2p+\lfloor k/2\rfloor+1} + \lceil k/2\rceil^pd^{2p+\lceil k/2\rceil+1}\right),
\end{align*}
which is the announced moment estimate \eqref{eq:moment-Gamma}.
\end{proof}

Comparing Theorem \ref{th:w(k-PPT-ext)} to Theorem \ref{th:w(k-ext)}, we see that the asymptotic mean width of the set of $k$-PPT-extendible states is at least $\sqrt{2}$ smaller than the asymptotic mean width of the set of $k$-extendible states. For instance, the set of $2$-PPT-extendible states is, on average, asymptotically smaller than the set of $4$-extendible states.
This however does not really shed light on why adding the constraint, at each step in the sequence of tests, that the symmetric extension is PPT across one fixed (even) cut would make the entanglement detection notably faster.

\section{Preliminaries on random-induced states and witnesses}
\label{section:random-states}

We will employ the notation $\rho\sim\mu_{n,s}$ to mean that $\rho=\tr_{\C^s}\ketbra{\psi}{\psi}$ with $\ket{\psi}$ a random Haar-distributed pure state on $\C^n\otimes\C^s$ (i.e.~$\rho$ describes an $n$-dimensional system which is obtained by partial-tracing over an $s$-dimensional ancilla space a uniformly distributed pure state on the global ``system+ancilla'' space). An equivalent mathematical characterization of such random state model is $\rho=W/\mathrm{Tr}W$ with $W\sim\mathcal{W}_{n,s}$ an $(n,s)$-Wishart matrix, i.e.~$W=GG^{\dagger}$ with $G$ a $n\times s$ matrix having independent complex normal entries (see e.g.~\cite{ZS}).

Let $K\subset\mathcal{D}(\C^n)$ be a convex body. For any $\rho\in\mathcal{D}(\C^n)$, a standard way of showing that $\rho\notin K$ is to produce a ``not belonging to $K$ witness'', i.e.~some $M\in\mathcal{H}_+(\C^n)$ which is such that
\[ \underset{\sigma\in K}{\sup}\mathrm{Tr}(M\sigma)<\mathrm{Tr}(M\rho). \]
By testing $\rho$ itself as possible such ``not belonging to $K$ witness'', we have
\begin{equation} \label{eq:witness} \underset{\sigma\in K}{\sup}\mathrm{Tr}(\rho\sigma)<\mathrm{Tr}(\rho^2)\ \Rightarrow\ \rho\notin K. \end{equation}

Crucially for the applications we have in mind, the functions $\rho\mapsto\mathrm{Tr}(\rho^2)$ and $\rho\mapsto\sup_{\sigma\in K}\mathrm{Tr}(\rho\sigma)$ both have nice concentration properties around their average. More precisely, we have the two following results.

\begin{proposition} \label{prop:concentration1}
Let $n,s\in\N$. Then, there exist universal constants $c,c'>0$ such that, for any $\eta>0$, first of all
\[ \mathbf{P}_{\rho\sim\mu_{n,s}}\left(\left|\mathrm{Tr}(\rho^2)- \mathbf{E}_{\tau\sim\mu_{n,s}}\big[\mathrm{Tr}(\tau^2)\big]\right|\geq\eta\right) \leq e^{-cs} + e^{-c'n^3s\eta^2}, \]
and second of all, for any convex body $K\subset\mathcal{D}(\C^n)$,
\[ \mathbf{P}_{\rho\sim\mu_{n,s}}\left(\left|\underset{\sigma\in K}{\sup}\mathrm{Tr}(\rho\sigma)-\mathbf{E}_{\tau\sim\mu_{n,s}}\left[\underset{\sigma\in K}{\sup}\mathrm{Tr}(\tau\sigma)\right]\right|\geq\eta\right) \leq e^{-cs} + e^{-c'n^2s\eta^2}. \]
\end{proposition}

\begin{proof} To show Proposition \ref{prop:concentration1}, we will make essential use of a local version of Levy's Lemma, namely (see \cite{ASY}, Lemma 3.4, for a proof): Let $\Omega\subset S^{m-1}$ be a subset of the Euclidean unit sphere of $\R^m$ satisfying $\P(\Omega)\geq 7/8$. Let also $f:S^{m-1}\rightarrow\R$ be a function whose restriction to $\Omega$ is $L$-Lipschitz and $M$ be a central value for $f$ (i.e.~$\P(\{f\geq M\})\geq 1/4$ and $\P(\{f\leq M\})\geq 1/4$). Then, for any $\eta>0$,
\[ \P\left(\left\{|f-M|\geq\eta\right\}\right) \leq \P\left(S^{m-1}\setminus\Omega\right) + e^{-c_0m\eta^2/L^2}, \]
where $c_0>0$ is a universal constant.

It is well-known (see e.g.~\cite{ZS} for a proof) that $\rho\sim\mu_{n,s}$ is equivalent to $\rho=XX^{\dagger}$ with $X$ uniformly distributed over the Hilbert--Schmidt unit sphere of $n\times s$ complex matrices, and the latter can be identified with the real Euclidean unit sphere $S^{2ns-1}$. Therefore, one may apply Levy's lemma above with $\Omega=\left\{X\in S^{2ns-1} \st \|X\|_{\infty}\leq 3/\sqrt{n}\right\}$, which is such that $\P(S^{2ns-1}\setminus\Omega)\leq e^{-cs}$ for some universal constant $c>0$ (see e.g.~\cite{ASW}, Lemma 6 and Appendix B, for a proof).

Consider first $f:X\in S^{2ns-1}\mapsto\mathrm{Tr}((XX^{\dagger})^2)$, which is $36/n$-Lipschitz on $\Omega$.  Indeed, for any $X,Y\in\Omega$,
\[ \left|f(X)-f(Y)\right| \leq \left\|\left(XX^{\dagger}\right)^2-\left(YY^{\dagger}\right)^2\right\|_1 \leq \left(\|XX^{\dagger}\|_{\infty}+\|YY^{\dagger}\|_{\infty}\right)\left(\|X\|_{2}+\|Y\|_{2}\right)\|X-Y\|_2 \leq \frac{36}{n} \|X-Y\|_2. \]
The second inequality is just by H\"{o}lder's inequality (more specifically $\|ABC\|_1\leq\|\A\|_{\infty}\|B\|_2\|C\|_2$) and the triangle inequality, after noticing that $\left(XX^{\dagger}\right)^2-\left(YY^{\dagger}\right)^2 = XX^{\dagger}\Delta+\Delta YY^{\dagger}$ with $\Delta=X(X^{\dagger}-Y^{\dagger})+(X-Y)Y^{\dagger}$. And the third inequality is because, by assumption, for any $Z\in\Omega$, $\|Z\|_2=1$ and $\|ZZ^{\dagger}\|_{\infty}=\|Z\|_{\infty}^2\leq 9/n$.

Now, the fact that $\P(S^{2ns-1}\setminus\Omega)\leq e^{-cs}$, combined with the fact that $|f|$ is bounded by $1$ on $S^{2ns-1}$, implies that the average of $|f|$ on $S^{2ns-1}\setminus\Omega$ is bounded by $e^{-cs}$, which tends to $0$ when $s$ tends to infinity. While the Lipschitz estimate for $f$ on $\Omega$ implies that the average of $f$ on $\Omega$ differs from its median by at most $C/n^{3/2}s^{1/2}$, which also tends to $0$ when $n,s$ tend to infinity. We can therefore conclude that the average of $f$ is a central value of $f$ for $n,s$ big enough. Hence, taking $M=\E f$ as central value for $f$, we get the concentration estimate
\[ \mathbf{P}_X\left(\left|\mathrm{Tr}\left(\left(XX^{\dagger}\right)^2\right)- \mathbf{E}_Y\mathrm{Tr}\left(\left(YY^{\dagger}\right)^2\right)\right|\geq\eta\right) \leq e^{-cs} + e^{-c'n^3s\eta^2}. \]

Take next $f:X\in S^{2ns-1}\mapsto\sup_{\sigma\in K}\mathrm{Tr}(XX^{\dagger}\sigma)$, which is $6/\sqrt{n}$-Lipschitz on $\Omega$. Indeed, for any $X,Y\in\Omega$,
\[ \left|f(X)-f(Y)\right| \leq \left| \sup_{\sigma\in K}\mathrm{Tr}\left(\left(XX^{\dagger}-YY^{\dagger}\right)\sigma\right)\right| \leq \left\|XX^{\dagger}-YY^{\dagger}\right\|_{\infty} \leq \left(\|X\|_{\infty}+\|Y\|_{\infty}\right) \left\|X-Y\right\|_{\infty} \leq \frac{6}{\sqrt{n}}\|X-Y\|_2. \]
The second inequality is just by duality, since $K$ is contained in the unit ball for the $1$-norm. The third inequality is by the triangle inequality, after noticing that $XX^{\dagger}-YY^{\dagger}=(X-Y)X^{\dagger}+Y(X^{\dagger}-Y^{\dagger})$. And the fourth inequality is by the norm inequality $\|\cdot\|_{\infty}\leq\|\cdot\|_2$ and because, by assumption, for any $Z\in\Omega$, $\|Z\|_{\infty}\leq 3/\sqrt{n}$.

Arguing as before, we see that the average of $f$ is a central value of $f$ for $n,s$ big enough (this time, the average of $|f|$ on $S^{2ns-1}\setminus\Omega$ is bounded by $e^{-cs}$ while the average of $f$ on $\Omega$ differs from its median by at most $C/ns^{1/2}$). Hence, taking $M=\E f$ as central value for $f$, we get the concentration estimate
\[ \P_X\left(\left\{\left|\underset{\sigma\in K}{\sup}\mathrm{Tr}(XX^{\dagger}\sigma)-\E_Y\underset{\sigma\in K}{\sup}\mathrm{Tr}(YY^{\dagger}\sigma)\right|\geq\eta\right\}\right) \leq e^{-cs} + e^{-c'n^2s\eta^2}. \]

Hence, we indeed have the two announced deviation probability bounds.
\end{proof}

Combining the two statements in Proposition \ref{prop:concentration1}, together with equation \eqref{eq:witness}, we get as a consequence: Let $K\subset\mathcal{D}(\C^n)$ be a convex body. Then, for any $\eta>0$,
\begin{equation} \label{eq:proba-notinK'} \mathbf{E}_{\rho\sim\mu_{n,s}}\big[\mathrm{Tr}(\rho^2)\big] - \mathbf{E}_{\rho\sim\mu_{n,s}}\left[\underset{\sigma\in K}{\sup}\mathrm{Tr}(\rho\sigma)\right] > \eta
\ \Rightarrow\ \mathbf{P}_{\rho\sim\mu_{n,s}}\big(\rho\notin K\big)\geq 1-e^{-cs\min(1,n^2\eta^2)}, \end{equation}
where $c>0$ is a universal constant.

From now on, we will in fact consider random-induced states on the bipartite space $\C^d\otimes\C^d$. So let $K\subset\mathcal{D}(\C^d\otimes\C^d)$ (such as e.g.~$\mathcal{P}(\C^d{:}\C^d)$ or $\mathcal{E}_k(\C^d{:}\C^d)$, $k\in\N$). It follows from equation \eqref{eq:proba-notinK'} that, for any $\eta>0$,
\begin{equation} \label{eq:proba-notinK} \mathbf{E}_{\rho\sim\mu_{d^2,s}}\big[\mathrm{Tr}(\rho^2)\big] - \mathbf{E}_{\rho\sim\mu_{d^2,s}}\left[\underset{\sigma\in K}{\sup}\mathrm{Tr}(\rho\sigma)\right] > \eta
\ \Rightarrow\ \mathbf{P}_{\rho\sim\mu_{d^2,s}}\big(\rho\notin K\big)\geq 1-e^{-cs\min(1,d^4\eta^2)}, \end{equation}
where $c>0$ is a universal constant.

\section{Non $k$-extendibility of random-induced states for ``small'' $k$}
\label{section:non-k-ext-witness}

\subsection{Strategy}

Our goal in the sequel will be to identify a range of environment size $s$ for which random-induced states on $\C^d\otimes\C^d$ are, with high-probability, not $k$-extendible. In view of equation \eqref{eq:proba-notinK}, this may be done by characterizing
\[ \left\{ s\in\N \st  \mathbf{E}_{\rho\sim\mu_{d^2,s}}\left[\underset{\sigma\in\mathcal{E}_k(\C^d{:}\C^d)}{\sup}\mathrm{Tr}(\rho\sigma)\right] < \mathbf{E}_{\rho\sim\mu_{d^2,s}}\big[\mathrm{Tr}(\rho^2)\big] \right\}. \]

Yet by Lemma \ref{lemma:sup-k-ext}, and using the notation introduced there, we have that for any state $\rho_{\A\B}$ on $\mathrm{A}\otimes\mathrm{B}$,
\[ \underset{\sigma_{\A\B}\in\mathcal{E}_k(\A{:}\B)}{\sup}\mathrm{Tr}\big(\rho_{\A\B}\sigma_{\A\B}\big) =\left\|\underset{j=1}{\overset{k}{\sum}}\widetilde{\rho}_{\A\B^k}(j)\right\|_{\infty} .\]

\subsection{An operator-norm estimate}

As explained above, to know when random-induced states on $\A\otimes\B$ are not $k$-extendible, what we need first is to compute the average operator-norm $\mathbf{E} \left\|\sum_{j=1}^k\widetilde{W}_{\A\B^k}(j)\right\|_{\infty}$, for $W$ a $(d^2,s)$-Wishart matrix. We will proceed in a very similar way to what was done in Section \ref{section:w(k-ext)}, and establish what can be seen as the analogues of Propositions \ref{prop:gaussian-infty} and \ref{prop:gaussian-p} but for Wishart instead of GUE matrices.

\begin{proposition} \label{prop:wishart-infty} Fix $k\in\N$ and $c>0$. Then,
\[ \mathbf{E}_{W_{\A\B}\sim\mathcal{W}_{d^2,cd^2}} \left\|\underset{j=1}{\overset{k}{\sum}}\widetilde{W}_{\A\B^k}(j)\right\|_{\infty} \underset{d\rightarrow+\infty}{\sim} (\sqrt{ck}+1)^2d^2. \]
\end{proposition}

As a preliminary step towards estimating the sup-norm $\mathbf{E} \left\|\sum_{j=1}^k\widetilde{W}_{\A\B^k}(j)\right\|_{\infty}$, we will look at the $p$-order moments $\mathbf{E}\, \mathrm{Tr}\left[\left(\sum_{j=1}^k\widetilde{W}_{\A\B^k}(j)\right)^p\right]$, $p\in\mathbf{N}$, and show that they can be expressed in terms of the $p$-order moments of a Mar\v{c}enko-Pastur distribution of appropriate parameter.

So let us recall first a few required definitions. For any $\lambda>0$, we shall denote by $\mu_{MP(\lambda)}$ the Mar\v{c}enko-Pastur distribution of parameter $\lambda$, whose density is given by
\[ \mathrm{d}\mu_{MP(\lambda)}(x) = \begin{cases} f_{\lambda}(x)\mathrm{d}x\ \text{if}\ \lambda>1 \\
\left(1-\lambda\right)\delta_0 + \lambda f_{\lambda}(x)\mathrm{d}x\ \text{if}\ \lambda\leq 1 \end{cases} ,\]
where, setting $\lambda_{\pm}=(\sqrt{\lambda}\pm 1)^2$, we defined the function $f_{\lambda}$ by
\[ f_{\lambda}(x)=\frac{\sqrt{(\lambda_+-x)(x-\lambda_-)}}{2\pi\lambda x}\mathbf{1}_{[\lambda_-,\lambda_+]}(x). \]
We shall also denote, for each $p\in\N$, by $\mathrm{M}_{MP(\lambda)}^{(p)}$ its $p$-order moment, i.e.~ $\mathrm{M}_{MP(\lambda)}^{(p)}=\int_{-\infty}^{+\infty}x^p\mathrm{d}\mu_{MP(\lambda)}(x)$. It is well-known that
\[ \forall\ p\in\N,\ \mathrm{M}_{MP(\lambda)}^{(p)}=\sum_{m=1}^p\lambda^m\mathrm{Nar}_p^m, \]
where $\mathrm{Nar}_p^m$ is the $(p,m)^{th}$ Narayana number defined in Lemma \ref{lemma:catalan}. In particular,  $\mathrm{M}_{MP(1)}^{(p)}=\mathrm{Cat}_p$, the $p^{\text{th}}$ Catalan number defined in Lemma \ref{lemma:catalan} as well.

\begin{proposition}\label{prop:wishart-p} Fix $k\in\N$ and $c>0$. Then, when $d\rightarrow+\infty$, the random matrix $\left(\sum_{j=1}^k\widetilde{W}_{\A\B^k}(j)\right)/d^2$ converges in moments towards a Mar\v{c}enko-Pastur distribution of parameter $ck$. Equivalently, this means that, for any $p\in\N$,
\[ \mathbf{E}_{W_{\A\B}\sim\mathcal{W}_{d^2,cd^2}}\, \mathrm{Tr}\left[\left(\underset{j=1}{\overset{k}{\sum}}\widetilde{W}_{\A\B^k}(j)\right)^p\right] \underset{d\rightarrow+\infty}{\sim} \mathrm{M}_{MP(ck)}^{(p)}d^{2p+k+1}. \]
\end{proposition}

\begin{remark}
Stronger convergence results than the one established in Proposition \ref{prop:wishart-p} may in fact be proved, as discussed in Appendix \ref{appendix:convergence}.
\end{remark}

\begin{proof} [Proof of Proposition \ref{prop:wishart-p}]
Let $p\in\N$. Computing the value of the $p$-order moment $\mathbf{E}\, \mathrm{Tr}\left[\left(\sum_{j=1}^k\widetilde{W}_{\A\B^k}(j)\right)^p\right]$ may be done using the Gaussian Wick formula (see Lemma \ref{lemma:wick} for the statement and Appendix \ref{appendix:wick-wishart} for a succinct summary of how to derive moments of Wishart matrices from it). In our case, we get by the computations carried out in Appendix \ref{appendix:wishart} and summarized in Proposition \ref{prop:wishart-p-preliminary} that, for any $d,s\in\N$, denoting by $\sharp(\cdot)$ the number of cycles in a permutation,
\begin{align*} \mathbf{E}_{W_{\A\B}\sim\mathcal{W}_{d^2,s}}\, \mathrm{Tr}\left[\left(\sum_{j=1}^k\widetilde{W}_{\A\B^k}(j)\right)^p\right] = & \sum_{f:[p]\rightarrow[k]}\mathbf{E}_{W_{\A\B}\sim\mathcal{W}_{d^2,s}}\, \mathrm{Tr}\left[\underset{i=1}{\overset{p}{\overrightarrow{\prod}}}\widetilde{W}_{\A\B^k}(f(i))\right] \\
= & \underset{f:[p]\rightarrow[k]}{\sum}\underset{\alpha\in\mathfrak{S}(p)}{\sum} d^{\sharp(\gamma^{-1}\alpha)+\sharp(\gamma_f^{-1}\alpha)+k-|\im(f)|}s^{\sharp(\alpha)}, \end{align*}
where we defined on $\{1,\ldots,p\}$, $\mathfrak{S}(p)$ as the set of permutations, $\gamma=(p\,\ldots\,1)$ as the canonical full cycle, and for each $f:[p]\rightarrow[k]$, $\gamma_f=\gamma_{f=1}\cdots\gamma_{f=k}$ as the product of the canonical full cycles on each of the level sets of $f$.

Hence, in the case where $s=cd^2$, for some constant $c>0$, we have
\begin{equation}
\label{eq:p-moment}
\mathbf{E}_{W_{\A\B}\sim\mathcal{W}_{d^2,cd^2}}\, \mathrm{Tr}\left[\left(\sum_{j=1}^k\widetilde{W}_{\A\B^k}(j)\right)^p\right] = \sum_{f:[p]\rightarrow[k]} \sum_{\alpha\in\mathfrak{S}(p)} c^{\sharp(\alpha)} d^{2\sharp(\alpha)+\sharp(\gamma^{-1}\alpha)+\sharp(\gamma_f^{-1}\alpha)+k-|\im(f)|}.
\end{equation}

We now have to understand which $\alpha\in\mathfrak{S}(p)$ and $f:[p]\rightarrow[k]$ contribute to the dominating term in the moment expansion \eqref{eq:p-moment}, i.e.~are such that the quantity $2\sharp(\alpha)+\sharp(\gamma^{-1}\alpha)+\sharp(\gamma_f^{-1}\alpha)+k-|\im(f)|$ is maximal.

First of all, for any $\alpha\in\mathfrak{S}(p)$, we have
\begin{equation} \label{eq:alpha} \sharp(\alpha)+\sharp(\gamma^{-1}\alpha)= 2p-(|\alpha|+|\gamma^{-1}\alpha|)\leq 2p-|\gamma| =p+\sharp(\gamma) =p+1, \end{equation}
where the first equality is by Lemma \ref{lemma:cycles-transpositions}, whereas the second inequality is by equation \eqref{eq:geodesic'} in Lemma \ref{lemma:distance} and is an equality if and only if $\alpha\in NC(p)$. Next, for any $\alpha\in\mathfrak{S}(p)$ and $f:[p]\rightarrow[k]$, we have
\begin{equation} \label{eq:alpha_f} \sharp(\alpha)+\sharp(\gamma_f^{-1}\alpha)= 2p-(|\alpha|+|\gamma_f^{-1}\alpha|)\leq 2p-|\gamma_f| =p+\sharp(\gamma_f) =p+|\im(f)|, \end{equation}
where the first equality is once more by Lemma \ref{lemma:cycles-transpositions}, whereas the second inequality is by equation \eqref{eq:geodesic''} in Lemma \ref{lemma:distance} and is an equality if and only if $\alpha\in NC(p)$ and $f\circ\alpha=f$. So equations \eqref{eq:alpha} and \eqref{eq:alpha_f} together yield that, for any $\alpha\in\mathfrak{S}(p)$ and $f:[p]\rightarrow[k]$,
\begin{equation} \label{eq:alpha-alpha_f} 2\sharp(\alpha)+\sharp(\gamma^{-1}\alpha)+\sharp(\gamma_f^{-1}\alpha)+k-|\im(f)|\leq 2p+k+1, \end{equation}
with equality if and only if $\alpha\in NC(p)$ and $f\circ\alpha=f$.

We thus get the asymptotic estimate
\[ \mathbf{E}\, \mathrm{Tr}\left[\left(\sum_{j=1}^k\widetilde{W}_{\A\B^k}(j)\right)^p\right] \underset{d\rightarrow+\infty}{\sim} \left(\sum_{\alpha\in NC(p)} \sum_{\underset{f\circ\alpha=f}{f:[p]\rightarrow[k]}}c^{\sharp(\alpha)}\right) d^{2p+k+1}. \]
Yet, a function $f$ satisfying $f\circ\alpha=f$ is fully characterized by its value on each of the $\sharp(\alpha)$ cycles of $\alpha$. So there are $k^{\sharp(\alpha)}$ such functions. Hence in the end, the asymptotic estimate
\[ \mathbf{E}\, \mathrm{Tr}\left[\left(\sum_{j=1}^k\widetilde{W}_{\A\B^k}(j)\right)^p\right] \underset{d\rightarrow+\infty}{\sim} \left(\sum_{\alpha\in NC(p)}(ck)^{\sharp(\alpha)}\right) d^{2p+k+1} = \mathrm{M}_{MP(ck)}^{(p)} d^{2p+k+1}, \]
the last equality being because, for any $\lambda>0$, $\sum_{\alpha\in NC(p)}\lambda^{\sharp(\alpha)} = \sum_{m=1}^p\lambda^m\mathrm{Nar}_p^m = \mathrm{M}_{MP(\lambda)}^{(p)}$.
\end{proof}

\begin{proof} [Proof of Proposition \ref{prop:wishart-infty}]
The argument will follow the exact same lines as the one used to derive Proposition \ref{prop:gaussian-infty} from Proposition \ref{prop:gaussian-p}.

As pointed out there, showing the inequality ``$\geq$'' in Proposition \ref{prop:wishart-infty} is easy. Indeed, the convergence in moments established in Proposition \ref{prop:wishart-p} implies 
that, asymptotically, the matrix $\left(\sum_{j=1}^k\widetilde{W}_{\A\B^k}(j)\right)/d^2$ has a largest eigenvalue which is, on average, at least the upper-edge of the support of $\mu_{MP(ck)}$, i.e.~$(\sqrt{ck}+1)^2$. In other words, it guarantees that there exist positive constants $c_d\rightarrow_{d\rightarrow+\infty}1$ such that
\begin{equation} \label{eq:lower-bound} \mathbf{E} \left\|\sum_{j=1}^k\widetilde{W}_{\A\B^k}(j)\right\|_{\infty} 
\geq c_d\,(\sqrt{ck}+1)^2d^2. \end{equation}

Let us now turn to the more tricky part, which is showing the inequality ``$\leq$'' in Proposition \ref{prop:wishart-infty}. For $d\in\N$ fixed, it holds thanks to Jensen's inequality and monotonicity of Schatten norms that
\begin{equation} \label{eq:infty-p-wishart} \forall\ p\in\N,\ \mathbf{E} \left\|\sum_{j=1}^k\widetilde{W}_{\A\B^k}(j)\right\|_{\infty} \leq \left( \mathbf{E} \left\|\sum_{j=1}^k\widetilde{W}_{\A\B^k}(j)\right\|_{\infty}^p \right)^{1/p} \leq \left(\mathbf{E}\, \mathrm{Tr}\left[\left(\sum_{j=1}^k\widetilde{W}_{\A\B^k}(j)\right)^p\right] \right)^{1/p} .\end{equation}
So let us fix $d\in\N$ and $p\in\N$, and rewrite \eqref{eq:p-moment} explicitly as an expansion in powers of $d$, keeping in the sum the permutations not saturating equation \eqref{eq:alpha-alpha_f}. Being cautious only regarding the permutations not saturating equation \eqref{eq:alpha_f}, and not regarding those not saturating equation \eqref{eq:alpha}, we thus get the upper bound
\begin{equation} \label{eq:momentsW-defect}
\mathbf{E}\, \mathrm{Tr}\left[\left(\sum_{j=1}^k\widetilde{W}_{\A\B^k}(j)\right)^p\right]\leq \left(\underset{f:[p]\rightarrow[k]}{\sum} \sum_{\delta=0}^{\lfloor(p+k)/2\rfloor}\sum_{m=1}^p \left|\mathfrak{S}_{f,\delta,m}(p)\right|c^m d^{-2\delta}\right) d^{2p+k+1},
\end{equation}
where we defined, for each $f:[p]\rightarrow[k]$, each $1\leq m\leq p$ and each $0\leq\delta\leq \lfloor(p+k)/2\rfloor$,
\[ \mathfrak{S}_{f,\delta,m}(p)=\left\{ \alpha\in\mathfrak{S}(p) \st \sharp(\alpha)=m\ \ \text{and}\ \ \sharp(\alpha)+\sharp(\gamma_f^{-1}\alpha) =p+|\im(f)|-2\delta \right\}. \]
$\mathfrak{S}_{f,\delta,m}(p)$ is thus nothing else than the set of permutations which are composed of $m$ cycles and have a defect $2\delta$ from lying on the geodesics between the identity and the product of the canonical full cycles on each of the level sets of $f$. This justifies in particular \textit{a posteriori} why the summation in \eqref{eq:momentsW-defect} is only over even defects (see the parity argument in Lemma \ref{lemma:distance-general}). Note that the definition of $\mathfrak{S}_{f,\delta,m}(p)$ can actually be extended to all $m\in\N$, with $\left|\mathfrak{S}_{f,\delta,m}(p)\right|=0$ if $m\geq p+|\im(f)|-2\delta$, which we shall do in what follows for writing convenience.

Now, by Lemma \ref{lemma:number-functions,permutations-defect}, we know that, if $0\leq\delta\leq\lfloor p/2\rfloor$, then for any $1\leq m\leq p$,
\[ \left|\big\{ (f,\alpha) \st \alpha\in\mathfrak{S}_{f,\delta,m}(p) \big\}\right| \leq \left(\sum_{\epsilon=0}^{2\delta}k^{m-\epsilon}\mathrm{Nar}_p^{m-\epsilon}\right)\times\left(2k^2p^2\right)^{2\delta}. \]
And if $\lceil p/2\rceil\leq\delta\leq \lfloor (p+k)/2\rfloor$, then trivially for any $1\leq m\leq p$,
\[ \left|\big\{ (f,\alpha) \st \alpha\in\mathfrak{S}_{f,\delta,m}(p) \big\}\right| \leq k^p\,p! \leq \left(\sum_{\epsilon=0}^{p}k^{m-\epsilon}\mathrm{Nar}_p^{m-\epsilon}\right)\times\left(2k^2p^2\right)^{p}. \]
What is more, for a given $0\leq\delta\leq\lceil p/2\rceil$, we have, making the change of summation index $m\mapsto m-\epsilon$,
\begin{align*}
\sum_{m=1}^pc^m\left(\sum_{\epsilon=0}^{2\delta}k^{m-\epsilon}\mathrm{Nar}_p^{m-\epsilon}\right) & = \left(\sum_{\epsilon=0}^{2\delta}c^{-\epsilon}\right) \left(\sum_{m=1}^p(ck)^{m}\mathrm{Nar}_p^{m}\right) \\
& \leq (1+c)^{2\delta} \sum_{m=1}^p(ck)^{m}\mathrm{Nar}_p^{m}\\
& = (1+c)^{2\delta} \mathrm{M}_{MP(ck)}^{(p)}.
\end{align*}
Putting everything together, we therefore get,
\[ \mathbf{E}\, \mathrm{Tr}\left[\left(\sum_{j=1}^k\widetilde{W}_{\A\B^k}(j)\right)^{p}\right] \leq \mathrm{M}_{MP(ck)}^{(p)} \left(1 + \sum_{\delta=1}^{\lfloor p/2\rfloor}\left(\frac{2(1+c)k^2p^2}{d}\right)^{2\delta} +\frac{k}{2}\left(\frac{2(1+c)k^2p^2}{d}\right)^p\right) d^{2p+k+1}.\]
Yet, $\max\left\{\left(2(1+c)k^2p^2/d\right)^{2\delta} \st 1\leq\delta\leq \lceil p/2\rceil\right\}$ is attained for $\delta=1$, provided $p\leq(d/2(1+c)k^2)^{1/2}$. So if such is the case,
\[ \sum_{\delta=1}^{\lfloor p/2\rfloor}\left(\frac{2(1+c)k^2p^2}{d}\right)^{2\delta}+\frac{k}{2}\left(\frac{2(1+c)k^2p^2}{d}\right)^p \leq \frac{p+k}{2} \frac{4(1+c)^2k^4p^4}{d^2} \leq \frac{4(1+c)^2k^4p^5}{d^2}, \]
where the last inequality holds as long as $p\geq k$. And hence, under all the previous assumptions,
\[ \mathbf{E}\, \mathrm{Tr}\left[\left(\sum_{j=1}^k\widetilde{W}_{\A\B^k}(j)\right)^{p}\right] \leq \mathrm{M}_{MP(ck)}^{(p)} \left(1 +\frac{4(1+c)^2k^4p^5}{d^2}\right) d^{2p+k+1}. \]

So set $p_d=(d/2(1+c)k^2)^{(2-\epsilon)/5}$ for some $0<\epsilon<1$ (which is indeed smaller than $(d/2(1+c)k^2)^{1/2}$ and bigger than $k$ for $d$ big enough, in particular bigger than $2(1+c)k^{9/2}$). And using inequality \eqref{eq:infty-p-wishart} in the special case $p=p_d$, we eventually get
\begin{equation} \label{eq:upper-bound} \mathbf{E} \left\|\sum_{j=1}^k\widetilde{W}_{\A\B^k}(j)\right\|_{\infty} \leq \left(\mathrm{M}_{MP(ck)}^{(p_d)} \left(1 + \frac{4(1+c)^2p_d^4}{d^2}\right)\right)^{1/p_d} d^{2+(k+1)/p_d}
\underset{d\rightarrow+\infty}{\sim} (\sqrt{ck}+1)^2d^2.
\end{equation}

Combining the lower bound in equation \eqref{eq:lower-bound} and the upper bound in equation \eqref{eq:upper-bound} yields Proposition \ref{prop:wishart-infty}.
\end{proof}

\subsection{Conclusion} Having at hand the operator-norm estimate from Proposition \ref{prop:wishart-infty}, we can now easily answer our initial question. It is the content of Theorem \ref{th:not-kext} below.

\begin{theorem}
\label{th:not-kext}
Let $k\in\N$, and for any $0<\epsilon<1/2$ define $c_{\epsilon}(k)=\frac{(k-1)^2}{4k}(1-\epsilon)$. Then, there exists a constant $C_{k,\epsilon}>0$ such that
\[ \mathbf{P}_{\rho\sim\mu_{d^2,c_{\epsilon}(k)d^2}} \big(\rho\notin\mathcal{E}_k(\C^d{:}\C^d)\big) \geq 1-e^{-C_{k,\epsilon}d^2}. \]
One can take $C_{k,\epsilon}=C\epsilon^2/k$ for some universal constant $C>0$.
\end{theorem}

\begin{proof}
As a direct consequence of Proposition \ref{prop:wishart-infty}, we have
\[ \mathbf{E}_{W_{\A\B}\sim\mathcal{W}_{d^2,cd^2}} \left[\underset{\sigma_{\A\B}\in\mathcal{E}_k(\A{:}\B)}{\sup}\ \tr\big(W_{\A\B}\sigma_{\A\B}\big)\right] \underset{d\rightarrow+\infty}{\sim} \frac{(\sqrt{ck}+1)^2}{k}d^2. \]
And since $\mathbf{E}_{W\sim\mathcal{W}_{d^2,s}}\mathrm{Tr}W\sim_{d,s\rightarrow +\infty} d^2s$ (see e.g.~\cite{CN} or Appendix \ref{appendix:wick-wishart}), the result we eventually come to after renormalizing by $\mathrm{Tr}W$ is
\begin{equation}
\label{eq:k-ext}
\mathbf{E}_{\rho\sim\mu_{d^2,cd^2}} \left[\underset{\sigma\in\mathcal{E}_k(\C^d{:}\C^d)}{\sup}\ \mathrm{Tr}(\rho\sigma)\right] \underset{d\rightarrow+\infty}{\sim} \frac{1}{d^2cd^2}\frac{(\sqrt{ck}+1)^2}{k}d^2 = \frac{(\sqrt{ck}+1)^2}{ck}\frac{1}{d^2}.
\end{equation}
On the other hand, $\mathbf{E}_{W\sim\mathcal{W}_{d^2,s}}\mathrm{Tr}(W^2)\sim_{d,s\rightarrow +\infty} d^2s^2 + (d^2)^2s$ (see e.g.~\cite{CN} or Appendix \ref{appendix:wick-wishart}), so we also have
\[ \mathbf{E}_{\rho\sim\mu_{d^2,cd^2}} \big[\mathrm{Tr}(\rho^2)\big] \underset{d\rightarrow+\infty}{\sim} \frac{d^2(cd^2)^2+(d^2)^2cd^2}{(d^2cd^2)^2} = \left(1+\frac{1}{c}\right)\frac{1}{d^2}. \]
Now, if $c=(1-\epsilon)(k-1)^2/4k$ for some $0<\epsilon<1/2$, then
\[ \frac{(\sqrt{ck}+1)^2}{ck}-\left(1+\frac{1}{c}\right)=\frac{4\epsilon}{(k-1)(1-\epsilon)} < \frac{8\epsilon}{k-1}. \]
So by equation \eqref{eq:proba-notinK}, we have in such case
\[ \mathbf{P}_{\rho\sim\mu_{d^2,cd^2}} \big(\rho\notin\mathcal{E}_k(\C^d{:}\C^d)\big) \geq 1-\exp\left(-Ckd^2d^4(\epsilon/kd^2)^2\right)= 1 -\exp\left(-Cd^2\epsilon^2/k\right), \]
for some universal constant $C>0$.
\end{proof}

\section{Discussion and comparison with other separability criteria}
\label{section:threshold-comparison}

For each $k\in\N$, define $c_{k-ext}$ as the smallest constant $c$ such that a random state $\rho$ on $\C^d\otimes\C^d$ induced by an environment of dimension $cd^2$ is not $k$-extendible with high probability when $d$ is large. That is,
\[ c_{k-ext}= \inf\left\{ c \st \P_{\rho\sim\mu_{d^2,cd^2}}\left(\rho\notin\mathcal{E}_k(\C^d{:}\C^d)\right) \underset{d\rightarrow+\infty}{\rightarrow} 1 \right\}. \]
What we established in Theorem \ref{th:not-kext} is that $c_{k-ext}\leq(k-1)^2/4k$.

Yet, we know from \cite{ASY} that for $c>0$ fixed, $\rho\sim\mu_{d^2,cd^2}$ is with high probability entangled when $d\rightarrow+\infty$: the threshold for $\rho\sim\mu_{d^2,s(d)}$ being with high probability either entangled or separable occurs for some $s(d)=s_0(d)$ with $d^3\lesssim s_0(d)\lesssim d^3\log^2 d$. So what we proved is that if $c<(k-1)^2/4k$, i.e.~if $k>2c+2\sqrt{c(c+1)}+1$, then this generic entanglement will be generically detected by the $k$-extendibility test.

Furthermore, it is well-known (see e.g.~\cite{ZS}) that $\rho\sim\mu_{d^2,d^2}$ is equivalent to $\rho$ being uniformly distributed on the set of mixed states on $\C^d\otimes\C^d$ (for the Haar measure induced by the Hilbert--Schmidt distance). As just mentioned, when $d\rightarrow+\infty$, such states are typically not separable. Now, for $k\geq 6$, $(k-1)^2/4k>1$, so such states are also typically not $k$-extendible. Hence, entanglement of uniformly distributed mixed states on $\C^d\otimes\C^d$ is typically detected by the $k$-extendibility test for $k\geq 6$.

Let us define, in a similar way to what was done for the $k$-extendibility criterion, $c_{ppt}$, resp.~$c_{ra}$, as the smallest constant $c$ such that a random state $\rho$ on $\C^d\otimes\C^d$ induced by an environment of dimension $cd^2$ is, with probability tending to one when $d$ tends to infinity, not satisfying the PPT, resp.~realignment, criterion.
We know from \cite{Aubrun1} that $c_{ppt}=4$, whereas we know from \cite{AN} that $c_{ra}=\left(8/3\pi\right)^2$. Now, for $k\geq 17$, $(k-1)^2/4k>4$, and for $k\geq 5$, $(k-1)^2/4k>\left(8/3\pi\right)^2$. So roughly speaking, this means that the $k$-extendibility criterion for separability becomes ``better'' than the PPT one at most for $k\geq 17$, and ``better'' than the realignment one at most for $k\geq 5$. This is to be taken in the following sense: if $k\geq 17$, resp.~$k\geq 5$, then there is a range of environment dimensions for which random-induced states have a generic entanglement which is generically detected by the $k$-extendibility test but not detected by the PPT, resp.~realignment, test.

Note also that for the reduction criterion \cite{HH}, it was established in \cite{JLN} that the threshold for a random-induced state on $\C^d\otimes\C^d$ either passing or failing it with high probability occurs at an environment dimension $d$, hence much smaller than for all previously mentioned criteria.

\section{The unbalanced case}
\label{section:unbalanced}

For the sake of simplicity, we previously focussed on the case where $\mathrm{H}=\mathrm{A}\otimes\mathrm{B}$ is a balanced bipartite Hilbert space. One may now wonder what happens, more generally, when $\mathrm{A}\equiv\C^{d_\A}$ and $\mathrm{B}\equiv\C^{d_\B}$ with $d_\A$ and $d_\B$ being possibly different. It is easy to see that the results from Theorems \ref{th:not-kext} and \ref{th:w(k-ext)} straightforwardly generalize to the case where $d_\A$ and $d_\B$ both tend to infinity (but possibly at different rates). The corresponding statements appear in Theorem \ref{th:unbalanced} below.

\begin{theorem} \label{th:unbalanced}
Let $k\in\N$ and let $d_\A,d_\B\in\N$. The mean width of the set of $k$-extendible states on $\C^{d_\A}\otimes\C^{d_\B}$ (with respect to $\C^{d_\B}$) satisfies
\[ w\big(\mathcal{E}_k(\C^{d_\A}{:}\C^{d_\B})\big) \underset{d_\A,d_\B\rightarrow+\infty}{\sim} \frac{2}{\sqrt{k}}\frac{1}{\sqrt{d_\A d_\B}} .\]
Also, when $d_\A,d_\B\rightarrow+\infty$, a random state on $\C^{d_\A}\otimes\C^{d_\B}$ which is sampled from $\mu_{d_\A d_\B,c d_\A d_\B}$, with $c<(k-1)^2/4k$, is with high probability not $k$-extendible (with respect to $\C^{d_\B}$).
\end{theorem}

Oppositely, when one of the two subsystems has a fixed dimension and the other one only has an increasing dimension, the sets of $k$-extendible states with respect to either the smaller or the bigger subsystem exhibit different size scalings. This is made precise in Theorem \ref{th:unbalanced'} below.

\begin{theorem} \label{th:unbalanced'}
Let $k\in\N$ and let $d_\A,d_\B\in\N$. If $d_\A$ is fixed, the mean width of the set of $k$-extendible states on $\C^{d_\A}\otimes\C^{d_\B}$ (with respect to $\C^{d_B}$) satisfies
\[ w\big(\mathcal{E}_k(\C^{d_\A}{:}\C^{d_\B})\big) \underset{d_\B\rightarrow+\infty}{\sim} \frac{2}{\sqrt{k}}\frac{1}{\sqrt{d_\A d_\B}} .\]
Whereas if $d_\B$ is fixed, the mean width of the set of $k$-extendible states on $\C^{d_\A}\otimes\C^{d_\B}$ (with respect to $\C^{d_\B}$) satisfies
\[ w\big(\mathcal{E}_k(\C^{d_\A}{:}\C^{d_\B})\big) \underset{d_\A\rightarrow+\infty}{\sim} \frac{2\,C(d_\B,k)}{\sqrt{k}}\frac{1}{\sqrt{d_\A d_\B}} ,\]
with $C(d_\B,k)\geq \left(1+(k-1)/d_\B^2\right)^{1/4}$.
\end{theorem}

\begin{proof}
Using the same notation as in the proof of Proposition \ref{prop:gaussian-p}, we start in both cases from the exact expression for the $2p$-order moment (slightly generalizing Proposition \ref{prop:gaussian-p-preliminary})
\begin{equation} \label{eq:p-dAdB} \mathbf{E}_{G_{\A\B}\sim GUE(d_\A d_\B)}\, \mathrm{Tr}\left[\left(\sum_{j=1}^k\widetilde{G}_{\A\B^k}(j)\right)^{2p}\right] =\underset{f:[2p]\rightarrow[k]}{\sum} \underset{\lambda\in\mathfrak{P}^{(2)}(2p)}{\sum} d_\A^{\sharp(\gamma^{-1}\lambda)}d_\B^{\sharp(\gamma_f^{-1}\lambda)+k-|\im(f)|}. \end{equation}

First, fix $d_\A$. The argument then follows the exact same lines as in the proof of Proposition \ref{prop:gaussian-p}. Indeed, the pair partitions $\lambda\in\mathfrak{P}^{(2)}(2p)$ contributing to the dominant order in $d_B$ in the expansion \eqref{eq:p-dAdB} are the $\mathrm{Cat}_p$ non-crossing pair partitions $\lambda\in NC^{(2)}(2p)$, for which $\sharp(\gamma^{-1}\lambda)=p+1$. Moreover, for each of these $\lambda$, the functions $f:[2p]\rightarrow[k]$ contributing to the dominant order in $d_B$ in the expansion \eqref{eq:p-dAdB} are the $k^p$ functions which are such that $f\circ\lambda=f$, for which $\sharp(\gamma_f^{-1}\lambda)+k-|\im(f)|=p+k$. So we eventually get
\[ \mathbf{E}_{G_{\A\B}\sim GUE(d_\A d_\B)}\, \mathrm{Tr}\left[\left(\sum_{j=1}^k\widetilde{G}_{\A\B^k}(j)\right)^{2p}\right] \underset{d_\B\rightarrow+\infty}{\sim} \mathrm{Cat}_p\,k^pd_\A^{p+1}d_\B^{p+k}. \]

Now, fix $d_\B$. Again, the pair partitions $\lambda\in\mathfrak{P}^{(2)}(2p)$ contributing to the dominant order in $d_\A$ in the expansion \eqref{eq:p-dAdB} are the $\mathrm{Cat}_p$ non-crossing pair partitions $\lambda\in NC^{(2)}(2p)$, for which $\sharp(\gamma^{-1}\lambda)=p+1$.
So consider one of these $\lambda$. Observe that, for any $0\leq\delta\leq \lfloor(p+k)/2\rfloor$, if $f:[2p]\rightarrow[k]$ is such that there are exactly $\delta$ pair blocks of $\lambda$ on which $f$ takes two values, then necessarily  $\sharp(\gamma_f^{-1}\lambda)+k-|\im(f)|\geq p+k-2\delta$. Indeed, the case $\delta=0$ is already known. So let us describe precisely what happens in the case $\delta=1$, i.e.~when there is exactly $1$ pair block of $\lambda$ on which $f$ takes $2$ values.\\
$\bullet$ If amongst these $2$ values, at least $1$ of them is also taken on another pair block of $\lambda$, then there exist transpositions $\tau,\tau'$ and a function $g$ satisfying $g\circ\lambda=g$, such that $\gamma_f=\gamma_g\tau\tau'$ and $|\im(f)|=|\im(g)|$. Hence,
\[ \sharp(\gamma_f^{-1}\lambda) = \sharp(\tau'\tau\gamma_f^{-1}\lambda) \geq \sharp(\gamma_g^{-1}\lambda) -2 = p+|\im(g)|-2 = p+|\im(f)|-2. \]\\
$\bullet$ If none of these $2$ values is also taken on another pair block of $\lambda$, then there exist a transposition $\tau$ and a function $g$ satisfying $g\circ\lambda=g$, such that $\gamma_f=\gamma_g\tau$ and $|\im(f)|=|\im(g)|+1$. Hence,
\[ \sharp(\gamma_f^{-1}\lambda) = \sharp(\tau\gamma_f^{-1}\lambda) \geq \sharp(\gamma_g^{-1}\lambda) -1 = p+|\im(g)|-1 = p+|\im(f)|-2. \]
And this generalizes in a similar way to $\delta>1$. Yet, for a given $0\leq\delta\leq \lfloor(p+k)/2\rfloor$, there are ${p \choose \delta}k^p(k-1)^{\delta}$ functions which take $2$ values on exactly $\delta$ pair blocks of $\lambda$ (assuming of course that $p\geq k$). So we eventually get
\begin{align*} \mathbf{E}_{G_{\A\B}\sim GUE(d_\A d_\B)}\, \mathrm{Tr}\left[\left(\sum_{j=1}^k\widetilde{G}_{\A\B^k}(j)\right)^{2p}\right] & \geq \mathrm{Cat}_p\,k^p\left(\sum_{\delta=0}^{\lfloor(p+k)/2\rfloor}{p \choose \delta}(k-1)^{\delta}d_\B^{-2\delta}\right)d_\B^{p+k}d_\A^{p+1}\\
& \geq \mathrm{Cat}_p\,k^p\left(\sum_{\delta=0}^{\lfloor(p+k)/2\rfloor}{\lfloor(p+k)/2\rfloor \choose \delta}\left(\frac{k-1}{d_\B^2}\right)^{\delta}\right)d_\B^{p+k}d_\A^{p+1}\\
& = \mathrm{Cat}_p\,k^p\left(1+\frac{k-1}{d_\B^2}\right)^{\lfloor(p+k)/2\rfloor}d_\B^{p+k}d_\A^{p+1}. \end{align*}

One can then argue as in the proof of the derivation of Proposition \ref{prop:gaussian-infty} from Proposition \ref{prop:gaussian-p} that we additionally have $\mathbf{E}\left\| \sum_{j=1}^{k} \widetilde{G}_{\A\B^k}\right\|_{\infty}\sim\lim_{p\rightarrow+\infty} \left( \mathbf{E}\, \mathrm{Tr}\left[ \left( \sum_{j=1}^{k} \widetilde{G}_{\A\B^k}\right)^{2p} \right] \right)^{1/2p}$, when either $d_\B\rightarrow+\infty$ or $d_\A\rightarrow+\infty$. This automatically yields the two announced statements on the mean width of $\mathcal{E}_k(\C^{d_\A}{:}\C^{d_\B})$.
\end{proof}

\begin{remark}
In the situation where $d_\B$ is fixed, if we had an exact expression
\[ \forall\ p\in\N,\ \mathbf{E}_{G_{\A\B}\sim GUE(d_\A d_\B)}\, \mathrm{Tr}\left[\left(\sum_{j=1}^k\widetilde{G}_{\A\B^k}(j)\right)^{2p}\right] \underset{d_\A\rightarrow+\infty}{\sim} M(d_\A,p), \]
then we would be able to conclude without any further argument that
\[ \mathbf{E}\left\| \sum_{j=1}^{k} \widetilde{G}_{\A\B^k}\right\|_{\infty} \underset{d_\A\rightarrow+\infty}{\sim} \lim_{p\rightarrow+\infty} M(d_\A,p)^{1/2p}. \]
This would indeed follow from the convergence result of \cite{HT} for non-commutative polynomials in multi-variables with matrix coefficients (in our case, $d_\B^2$ variables and $d_\B^k\times d_\B^k$ coefficients).
\end{remark}

The asymmetry in the definition of $k$-extendibility appears more strikingly in this unbalanced setting. Indeed, for a finite $k\in\N$, a given state on $\mathrm{A}\otimes\mathrm{B}$ may be $k$-extendible with respect to $\mathrm{B}$ but not $k$-extendible with respect to $\mathrm{A}$. It is only in the limit $k\rightarrow +\infty$ that there is equivalence between the two notions: a state on $\mathrm{A}\otimes\mathrm{B}$ is $k$-extendible with respect to $\mathrm{B}$ for all $k\in\N$ if and only if it is $k$-extendible with respect to $\mathrm{A}$ for all $k\in\N$ (and if and only if it is separable).

However, what Theorem \ref{th:unbalanced} stipulates is that, even for a finite $k\in\N$, when both subsystems grow, being $k$-extendible with respect to either one or the other are two constraints which are, on average, equivalently restricting. On the contrary, what Theorem \ref{th:unbalanced'} shows is that when only one subsystem grows, and the other remains of fixed size, being $k$-extendible with respect to the bigger one is, on average, a tougher constraint than being $k$-extendible with respect to the smaller one (as one would have probably expected).

This is to be put in perspective with some of the original observations made in \cite{DPS}. It was indeed noticed that checking whether a state on $\C^{d}\otimes\C^{d'}$ is $k$-extendible with respect to $\C^{d'}$ requires space resources which scale as $\left[{d'+k-1 \choose k}d\right]^2$ when implemented. It was therefore advised that in the unbalanced situation of $d$ ``big'' and $d'$ ``small'', one should check $k$-extendibility with respect to $\C^{d'}$ rather than $\C^d$, the former being much more economical. On the other hand, it comes out from our study that, in this case, an entangled state is likely to fail passing the $k$-extendibility test for a smaller $k$ when the extension is searched with respect to $\C^d$ than when it is searched with respect to $\C^{d'}$. But understanding the precise trade-off seems out of reach at the moment.

\section{Miscellaneous questions}
\label{section:miscellaneous}

\subsection{What about the mean width of the set of $k$-extendible states for ``big'' $k$?}

All the statements proven sofar, regarding either the $k$-extendibility of random-induced states or the mean width of the set of $k$-extendible states, converge towards the same (expected) conclusion: for any given $k\in\N$, the $k$-extendibility criterion becomes a very weak necessary condition for separability when the dimension of the considered bipartite system increases. So the natural question at that point is: what can be said about the $k$-extendibility criterion on $\C^d\otimes\C^d$ when $k\equiv k(d)$ is allowed to grow in some way with $d$? Unfortunately, most of the results we established rely at some point on the assumption that $k$ is a fixed parameter, and therefore do not seem to be directly generalizable to the case where $k$ depends on $d$.

There is at least one estimate though that remains valid in this setting, which is the lower bound on the mean width of $k$-extendible states.

\begin{theorem} \label{th:w(k(d)-ext)-lower}
There exist positive constants $c_d\rightarrow_{d\rightarrow+\infty}1$ such that, for any $k(d)\in\N$, the mean width of the set of $k(d)$-extendible states on $\C^d\otimes\C^d$ satisfies
\[ w\left(\mathcal{E}_{k(d)}(\C^d{:}\C^d)\right) \geq c_d\, \frac{2}{\sqrt{k(d)}}\frac{1}{d}. \]
\end{theorem}

\begin{proof}
Let $d,k(d)\in\N$. For any $p\in\N$, the exact expression for the $2p$-order moment established in Proposition \ref{prop:gaussian-p-preliminary} of course remains true. So by the same arguments as in the proof of Proposition \ref{prop:gaussian-p}, we still have in that case at least the lower bound
\begin{align*}
\mathbf{E}_{G_{\A\B}\sim GUE(d^2)} \mathrm{Tr}\left[\left(\sum_{j=1}^{k(d)}\widetilde{G}_{\A\B^{k(d)}}(j)\right)^{2p}\right] & = \underset{f:[2p]\rightarrow[k(d)]}{\sum} \underset{\lambda\in\mathfrak{P}^{(2)}(2p)}{\sum} d^{\sharp(\gamma^{-1}\lambda)+\sharp(\gamma_f^{-1}\lambda)+k(d)-|\im(f)|}\\
& \geq \mathrm{Cat}_p\,k(d)^pd^{2p+k(d)+1}.
\end{align*}
This lower bound on moments in turn guarantees, as explained in the derivation of Proposition \ref{prop:gaussian-infty} from Proposition \ref{prop:gaussian-p}, that there exist positive constants $c_d\rightarrow_{d\rightarrow+\infty}1$ such that we have the inequality
\[ \mathbf{E}_{G_{\A\B}\sim GUE(d^2)} \left\| \underset{j=1}{\overset{k(d)}{\sum}}\widetilde{G}_{\A\B^{k(d)}}(j) \right\|_{\infty} 
\geq c_d\, 2\sqrt{k(d)}d ,\]
which yields the announced lower bound for the mean width of $\mathcal{E}_{k(d)}(\C^d{:}\C^d)$.
\end{proof}

Theorem \ref{th:w(k(d)-ext)-lower} only provides a lower bound on the asymptotic mean width of $\mathcal{E}_{k}(\C^d{:}\C^d)$ when $k$ is allowed to depend on $d$. It is nevertheless already an interesting piece of information. Indeed, as mentioned in Section \ref{section:w-comparison}, we know from \cite{AS} that the mean width of the set of separable states on $\C^d\otimes\C^d$ is of order $1/d^{3/2}$. Theorem \ref{th:w(k(d)-ext)-lower} therefore asserts that, on $\C^d\otimes\C^d$, one has to go at least to $k$ of order $d$ to obtain a set of $k$-extendible states whose mean width scales as the one of the set of separable states.

Furthermore, it may be worth mentioning that the proof of Proposition \ref{prop:gaussian-infty} actually provides additional information, namely an upper bound on the mean width of $k$-extendible states which remains valid for a quite wide range of $k$.

\begin{theorem} \label{th:w(k(d)-ext)-upper}
For any $d,k(d)\in\N$, provided that $k(d)<d^{2/7}$ and $d$ is big enough, the mean width of the set of $k(d)$-extendible states on $\C^d\otimes\C^d$ satisfies
\[ w\big(\mathcal{E}_{k(d)}(\C^d{:}\C^d)\big) \leq \frac{2}{\sqrt{k(d)}}\frac{1}{d}\exp\left(\frac{k(d)^{7/5}\ln d}{d^{2/5}}\right). \]
\end{theorem}

\begin{proof}
Let $d,k(d)\in\N$ with $k(d)<d^{2/7}$. Taking $p_d=(d/k(d))^{2/5}$ in equation \eqref{eq:upper-bound-gaussian} (which is indeed, as required, bigger than $k(d)$ and smaller than $(2d/k(d))^{1/2}$ for $d$ big enough) we get
\[ \mathbf{E}_{G_{\A\B}\sim GUE(d^2)} \left\|\sum_{j=1}^{k(d)}\widetilde{G}_{\A\B^k}(j)\right\|_{\infty} \leq 2\sqrt{k(d)}\,d\, \exp\left(\frac{k(d)^{7/5}}{2d^{2/5}}\ln d +\frac{k(d)^{2/5}}{2d^{2/5}}\left[\ln d + \ln \left(\frac{5}{4}\right)\right] \right). \]
The latter quantity is smaller than $2\sqrt{k(d)}\,d\,\exp\left(k(d)^{7/5}\ln d /d^{2/5}\right)$ for $d$ big enough, which yields the advertised upper bound for the mean width of $\mathcal{E}_{k(d)}(\C^d{:}\C^d)$.
\end{proof}

Of course, the upper bound provided by Theorem \ref{th:w(k(d)-ext)-upper} is interesting only for $k(d)<k_0(d)=d^{2/7}/(\ln d)^{5/7}$. Nevertheless, since the set of $k$-extendible states contains the set of $k'$-extendible states for all $k'\geq k$, we also have as a (potentially weak) consequence of Theorem \ref{th:w(k(d)-ext)-upper} that for $k(d)\geq k_0(d)$ and $d$ big enough,
\[ w\big(\mathcal{E}_{k(d)}(\C^d{:}\C^d)\big) \leq w\big(\mathcal{E}_{k_0(d)}(\C^d{:}\C^d)\big) \leq \frac{2e\,(\ln d)^{5/14}}{d^{8/7}} . \]

Theorems \ref{th:w(k(d)-ext)-lower} and \ref{th:w(k(d)-ext)-upper} together imply in particular the following: in the regime where $k(d)$ grows with $d$ slower than $d$ itself, both the ratio $w\big(\mathcal{S}(\C^d{:}\C^d)\big)/w\big(\mathcal{E}_{k(d)}(\C^d{:}\C^d)\big)$ and the ratio $w\big(\mathcal{E}_{k(d)}(\C^d{:}\C^d)\big)/w\big(\mathcal{D}(\C^d\otimes\C^d)\big)$ are unbounded. To rephrase it, for $k(d)$ having this growth rate, the set of $k(d)$-extendible states lies ``strictly in between'' the set of separable states and the set of all states from an asymptotic size point of view.

\subsection{When is a random-induced state with high probability $k$-extendible?}

The result provided by Theorem \ref{th:not-kext} is only one-sided: it tells us that if $s<s(k,d)=d^2(k-1)^2/4k$, then a random mixed state on $\C^d\otimes\C^d$ obtained by partial tracing on $\C^s$ a uniformly distributed pure state on $\C^d\otimes\C^d\otimes\C^s$ is with high probability not $k$-extendible. But what can be said about the case $s>s(k,d)$? Or more generally, can one find a reasonable $s'(k,d)$ such that if $s>s'(k,d)$, then a random mixed state on $\C^d\otimes\C^d$ obtained by partial tracing on $\C^s$ a uniformly distributed pure state on $\C^d\otimes\C^d\otimes\C^s$ is with high probability $k$-extendible?

By the arguments discussed in extensive depth in \cite{ASY}, one can assert at least that there exists a universal constant $c>0$ such that $ckd^2\log^2d$ is a possible value for such $s'(k,d)$. We will not repeat the whole reasoning here, but let us still give the key ideas underlying it.

Define $\overline{\mathcal{E}}_k(\C^d{:}\C^d)$ the translation of $\mathcal{E}_k(\C^d{:}\C^d)$ by its center of mass, the maximally mixed state $\Id/d^2$, i.e.~\[ \overline{\mathcal{E}}_k(\C^d{:}\C^d) =\left\{ \rho-\frac{\Id}{d^2} \st \rho\in\mathcal{E}_k(\C^d{:}\C^d) \right\}. \]
Define also $\overline{\mathcal{E}}_k(\C^d{:}\C^d)^{\circ}$ the convex body polar to $\overline{\mathcal{E}}_k(\C^d{:}\C^d)$, i.e.~\[ \overline{\mathcal{E}}_k(\C^d{:}\C^d)^{\circ} = \left\{ \Delta\in\mathcal{H}(d^2) \st \forall\ X\in\overline{\mathcal{E}}_k(\C^d{:}\C^d),\ \tr\left(\Delta X\right)\leq 1 \right\}. \]
What then has to be specifically determined is (see \cite{ASY}, Section 2, for further comments)
\[ \left\{ s\in\N \st \E_{\rho\sim\mu_{d^2,s}}\sup_{\Delta\in\overline{\mathcal{E}}_k(\C^d{:}\C^d)^{\circ}}\tr\left(\left(\rho-\frac{\Id}{d^2}\right)\Delta\right)<1 \right\}, \]
One can first of all use the fact that, roughly speaking, when $d,s\rightarrow+\infty$, the random matrix $\rho-\Id/d^2$ for $\rho\sim\mu_{d^2,s}$ ``looks the same as'' the random matrix $G/d^2\sqrt{s}$ for $G\sim GUE(d^2)$ (see \cite{ASY}, Proposition 3.1 and Remark 3.2 as well as Appendices A and B, for precise majorization statements and proofs). In particular, there exists a constant $C>0$ such that, for all $d,s\in\N$ with (say) $d^2\leq s\leq d^3$, we have the upper bound
\[ \E_{\rho\sim\mu_{d^2,s}}\sup_{\Delta\in\overline{\mathcal{E}}_k(\C^d{:}\C^d)^{\circ}}\tr\left(\left(\rho-\frac{\Id}{d^2}\right)\Delta\right)\leq \frac{C}{d^2\sqrt{s}} \E_{G\sim GUE(d^2)}\sup_{\Delta\in\overline{\mathcal{E}}_k(\C^d{:}\C^d)^{\circ}}\tr\left(G\Delta\right). \]
Next, due to the fact that, again putting it vaguely, the convex body $\overline{\mathcal{E}}_k(\C^d{:}\C^d)$ is ``sufficiently well-balanced'' (see \cite{ASY}, Section 4 as well as Appendices C and D, for a complete exposition of the $\ell$-position argument), we know that there exists a constant $C'>0$ such that, for all $d\in\N$, we have the upper bound
\[ \left(\E_{G\sim GUE(d^2)}\sup_{\Delta\in\overline{\mathcal{E}}_k(\C^d{:}\C^d)^{\circ}}\tr\left(G\Delta\right)\right) \left(\E_{G\sim GUE(d^2)}\sup_{\Delta\in\overline{\mathcal{E}}_k(\C^d{:}\C^d)}\tr\left(G\Delta\right)\right) \leq C'd^4\log d. \]
Now, $\E_{G\sim GUE(d^2)}\sup_{\Delta\in\overline{\mathcal{E}}_k(\C^d{:}\C^d)}\tr\left(G\Delta\right)$ is nothing else than the Gaussian mean width of $\overline{\mathcal{E}}_k(\C^d{:}\C^d)$, which is the same as the Gaussian mean width of $\mathcal{E}_k(\C^d{:}\C^d)$, so for which we have an estimate thanks to Theorem \ref{th:w(k-ext)}, namely $w_G\left(\overline{\mathcal{E}}_k(\C^d{:}\C^d)\right) \sim_{d\rightarrow+\infty} 2d/\sqrt{k}$.

Putting everything together, we see that
\[ \E_{\rho\sim\mu_{d^2,s}}\sup_{\Delta\in\overline{\mathcal{E}}_k(\C^d{:}\C^d)^{\circ}}\tr\left(\left(\rho-\frac{\Id}{d^2}\right)\Delta\right)\leq \widetilde{C}\frac{\sqrt{k}d\log d}{\sqrt{s}}, \]
for some constant $\widetilde{C}>0$ independent of $d,s,k\in\N$, which implies as claimed that if $s> \widetilde{C}^2kd^2\log^2d$, then $\E_{\rho\sim\mu_{d^2,s}}\sup_{\Delta\in\overline{\mathcal{E}}_k(\C^d{:}\C^d)^{\circ}}\tr\left(\left(\rho-\Id/d^2\right)\Delta\right)<1$.

\begin{remark}
Let us briefly comment on a notable difference, from a convex geometry point of view, between the $k$-extendibility criterion and other common separability criteria. In the case of $k$-extendibility, computing the support function of $\overline{\mathcal{E}}_k$ is easier than computing the support function of its polar $\overline{\mathcal{E}}_k^{\circ}$, while for other separability relaxations it is usually the opposite. Indeed, for a given traceless unit Hilbert-Schmidt norm Hermitian $\Delta$ on $\C^d\otimes\C^d$, we have for instance the closed formulas
\[ h_{\overline{\mathcal{E}}_k(\C^d{:}\C^d)}(\Delta) = \sup \left\{ \tr(\Delta X) \st X\in\overline{\mathcal{E}}_k(\C^d{:}\C^d) \right\} = \left\|\widetilde{\Delta}\right\|_{\infty}, \]
\[ h_{\overline{\mathcal{P}}(\C^d{:}\C^d)^{\circ}}(\Delta) = \sup \left\{ \tr(\Delta X) \st X\in\overline{\mathcal{P}}(\C^d{:}\C^d)^{\circ} \right\} = d^2\left\|\Delta^{\Gamma}\right\|_{\infty}, \]
whereas the dual quantities $h_{\overline{\mathcal{E}}_k(\C^d{:}\C^d)^{\circ}}(\Delta)$ and $h_{\overline{\mathcal{P}}(\C^d{:}\C^d)}(\Delta)$ cannot be written in such a simple way.

This explains why in the case of $\mathcal{E}_k$ it is the mean width that can be exactly computed, contrary to the threshold value which can only be approximated, while for other approximations of $\mathcal{S}$ the reverse generally happens.
\end{remark}

\section*{Acknowledgements}

This research project originated from a question raised by Fernando Brand\~{a}o, Toby Cubitt and Ashley Montanaro. I am therefore extremely grateful to them, first of all of course for this enthusiastic launching, but even more so for their unfailing interest and support afterwards. I would also like to thank Beno\^{\i}t Collins and Ion Nechita for their most insightful comments at various key stages of this work. 
And last but not least, many thanks to Guillaume Aubrun for his numerous helpful remarks and his careful proof-reading.

This research was supported by the ANR projects OSQPI 11-BS01-0008 and Stoq 14-CE25-0033, and by the ERC advanced grant IRQUAT 2010-AdG-267386. It was initiated during the ``Intensive month on operator algebras and quantum information'' taking place at the ICMAT in Madrid in the summer 2013, and pursued during the thematic programme ``Mathematical challenges in quantum information'' taking place at the Isaac Newton Institute in Cambridge in the fall 2013. The hospitality of both institutions and the work of both organising teams are gratefully acknowledged as well.

\appendix

\section{Combinatorics of permutations and partitions: short summary of standard facts}
\label{appendix:combinatorics}

Let $p\in\N$. We denote by $\mathfrak{S}(p)$ the set of permutations on $\{1,\ldots,p\}$. For any $\pi\in\mathfrak{S}(p)$, we denote by $\sharp(\pi)$ the number of cycles in the decomposition of $\pi$ into a product of disjoint cycles, and by $|\pi|$ the minimal number of transpositions in the decomposition of $\pi$ into a product of transpositions. We also define $\gamma\in\mathfrak{S}(p)$ as the canonical full cycle $(p\,\ldots\,1)$. More generally, we shall say that $c$ is the canonical full cycle on a set $\{i_1,\ldots,i_p\}$ with $i_1<\cdots<i_p$ if $c=(i_p\,\ldots\,i_1)$.

Some standard results related to $\mathfrak{S}(p)$ are gathered below (see e.g.~\cite{NS}, Lectures 9 and 23, for more details).

\begin{lemma}
\label{lemma:cycles-transpositions}
For any $\pi\in\mathfrak{S}(p)$, $\sharp(\pi)+|\pi|=p$.
\end{lemma}

\begin{lemma}
\label{lemma:distance-general}
$d:(\pi,\varsigma)\in\mathfrak{S}(p)\times\mathfrak{S}(p)\mapsto|\pi^{-1}\varsigma|$ defines a distance on $\mathfrak{S}(p)$, so that for any $\pi,\varsigma\in\mathfrak{S}(p)$,
\begin{equation}
\label{eq:geodesic}
|\varsigma^{-1}\pi|+|\pi|=d(\varsigma,\pi)+d(\id,\pi)\geq d(\id,\varsigma) =|\varsigma|,
\end{equation}
with equality in \eqref{eq:geodesic} if and only if $\pi$ lies on the geodesic between $\id$ and $\varsigma$. And whenever this is not the case, there exists $\delta\in\{1,\ldots, p-1\}$ such that $|\varsigma^{-1}\pi|+|\pi|= |\varsigma|+2\delta$.
\end{lemma}

\begin{definition}
A partition of $\{1,\ldots,p\}$ is a family $\lambda=\{I_1,\ldots,I_L\}$ of disjoint non-empty subsets of $\{1,\ldots,p\}$ whose union is $\{1,\ldots,p\}$. The sets $I_1,\ldots,I_L$ are called the blocks of $\lambda$. If each of them contains exactly $2$ elements, $\lambda$ is said to be a pair partition of $\{1,\ldots,p\}$. We shall denote by $\mathfrak{P}(p)$ the set of partitions of $\{1,\ldots,p\}$, and by $\mathfrak{P}^{(2)}(p)$ the set of pair partitions of $\{1,\ldots,p\}$. Note that $\mathfrak{P}^{(2)}(p)=\emptyset$ if $p$ is odd. Remark also that, whenever $p$ is even, the set of pair partitions of $\{1,\ldots,p\}$ is in bijection with the set of pairings on $\{1,\ldots,p\}$ (i.e.~ the set of permutations on $\{1,\ldots,p\}$ which are a product of $p/2$ disjoint transpositions). We shall therefore make no distinction between both.

A partition of $\{1,\ldots,p\}$ is said to be non-crossing if there does not exist $i<j<k<l$ in $\{1,\ldots,p\}$ such that $i,k$ belong to the same block, $j,l$ belong to the same block, and $i,j$ belong to different blocks. We shall denote by $NC(p)$ the set of non-crossing partitions of $\{1,\ldots,p\}$, and by $NC^{(2)}(p)$ the set of pair non-crossing partitions of $\{1,\ldots,p\}$. Note that $NC^{(2)}(p)=\emptyset$ if $p$ is odd.
\end{definition}

A well-known combinatorial result regarding non-crossing partitions is the following.

\begin{lemma} \label{lemma:catalan}
The number of non-crossing partitions of $\{1,\ldots,p\}$ and the number of pair non-crossing partitions of $\{1,\ldots,2p\}$ are both equal to the $p^{th}$ Catalan number
\[ \mathrm{Cat}_p=\frac{1}{p+1}{2p \choose p}. \]
More precisely, for any $1\leq m\leq p$, the number of non-crossing partitions of $\{1,\ldots,p\}$ which are composed of exactly $m$ blocks is equal to the $(p,m)^{th}$ Narayana number
\[ \mathrm{Nar}_p^m=\frac{1}{p+1}{p+1 \choose m}{p-1 \choose m-1}. \]
Obviously, these numbers are such that $\sum_{m=1}^p\mathrm{Nar}_p^m=\mathrm{Cat}_p$.
\end{lemma}

With these definitions in mind, we can now state a special case of particular interest of Lemma \ref{lemma:distance-general}.

\begin{lemma}
\label{lemma:distance}
Denote by $\gamma$ the canonical full cycle on $\{1,\ldots,p\}$. Then, for any $\pi\in\mathfrak{S}(p)$,
\begin{equation}
\label{eq:geodesic'}
|\gamma^{-1}\pi|+|\pi|\geq |\gamma|= p-1,
\end{equation}
with equality in \eqref{eq:geodesic'} if and only if $\pi$ lies on the geodesic between $\id$ and $\gamma$. The latter subset of $\mathfrak{S}(p)$ is in bijection with the set of non-crossing partitions of $\{1,\ldots,p\}$ (by the mapping which associates to a given partition the product of the canonical full cycles on each of its blocks). We shall thus write $\pi\in NC(p)$ in such case, not distinguishing a geodesic permutation from its corresponding non-crossing partition.

More generally, let $\{I_1,\ldots,I_L\}$ be a partition of $\{1,\ldots,p\}$ and denote by $\gamma_1,\ldots,\gamma_L$ the canonical full cycles on $I_1,\ldots,I_L$. Then, for any $\pi\in\mathfrak{S}(p)$,
\begin{equation}
\label{eq:geodesic''}
|(\gamma_1\cdots\gamma_L)^{-1}\pi|+|\pi|\geq |\gamma_1\cdots\gamma_L|= p-L,
\end{equation}
with equality in \eqref{eq:geodesic''} if and only if $\pi$ lies on the geodesic between $\id$ and $\gamma_1\cdots\gamma_L$. The latter subset of $\mathfrak{S}(p)$ is in bijection with the set of non-crossing partitions of $\{1,\ldots,p\}$ which are finer than $I_1\sqcup\cdots\sqcup I_L$, which itself is in bijection with $NC(|I_1|)\times\cdots\times NC(|I_L|)$.
\end{lemma}

Combining Lemma \ref{lemma:distance} with Lemma \ref{lemma:catalan}, we can in fact say the following: Let $\varsigma\in\mathfrak{S}(p)$ and assume that its decomposition into disjoint cycles is $\varsigma=c_1\cdots c_L$ where, for each $1\leq i\leq L$, $c_i$ is of length $p_i$ (hence with $p_1+\cdots+p_L=p$). Then, for any $\pi\in\mathfrak{S}(p)$, $|\varsigma^{-1}\pi|+|\pi|\geq p-L$, and
\[ \left|\left\{ \pi\in\mathfrak{S}(p) \st |\varsigma^{-1}\pi|+|\pi| = p-L \right\}\right| = \mathrm{Cat}_{p_1}\times\cdots\times\mathrm{Cat}_{p_L}. \]
Having this easy observation in mind might be useful later on.



\section{Computing moments of Gaussian matrices: Wick formula and genus expansion}
\label{appendix:wick}

When computing expectations of Gaussian random variables, a useful tool is the Wick formula (see e.g.~\cite{Zvonkin} or \cite{NS}, Lecture 22, for a proof).

\begin{lemma}[Gaussian Wick formula] Let $X_1,\ldots,X_q$ be jointly Gaussian centered random variables (real or complex).
\begin{align*}
& \text{If}\ q=2p+1\ \text{is odd},\ \text{then}\ \E\left[X_1\cdots X_q\right] = 0.\\
& \text{If}\ q=2p\ \text{is even},\ \text{then}\ \E\left[X_1\cdots X_q\right] = \underset{\{\{i_1,j_1\},\ldots,\{i_p,j_p\}\}\in\mathfrak{P}^{(2)}(2p)}{\sum} \,\underset{m=1}{\overset{p}{\prod}}\E\left[X_{i_m}X_{j_m}\right].
\end{align*}
\label{lemma:wick}
\end{lemma}

\subsection{Moments of GUE matrices}
\label{appendix:wick-gaussian}

A first important application of Lemma \ref{lemma:wick} is to the computation of the moments of matrices from the Gaussian Unitary Ensemble. Indeed, for any $q\in\N$, we have
\[ \E_{G\sim GUE(n)} \tr\left(G^q\right) = \underset{1\leq l_1,\ldots,l_q\leq n}{\sum} \E\left[G_{l_1,l_2}\cdots G_{l_q,l_1}\right], \]
where the $G_{i,j}$, $1\leq i,j\leq n$, are centered Gaussian random variables satisfying $\E[G_{i,j}G_{i',j'}]=\delta_{i=j',j=i'}$. So what we get applying the Wick formula is that, for any $p\in\N$,
\begin{align*}
& \E_{G\sim GUE(n)} \tr\left(G^{2p+1}\right) = 0 \\
& \E_{G\sim GUE(n)} \tr\left(G^{2p}\right) = \underset{\lambda\in\mathfrak{P}^{(2)}(2p)}{\sum} n^{\flat(\lambda)},
\end{align*}
where for each pair partition $\lambda=\{\{i_1,j_1\},\ldots,\{i_p,j_p\}\}$ of $\{1,\ldots,2p\}$, $\flat(\lambda)$ is the number of free parameters $l_1,\ldots,l_{2p}\in\{1,\ldots,n\}$ when imposing that $\forall\ 1\leq m\leq p,\ l_{i_m+1}=l_{j_m},\ l_{j_m+1}=l_{i_m}$. Identifying the pair partition $\{\{i_1,j_1\},\ldots,\{i_p,j_p\}\}$ with the pairing $(i_1\,j_1)\ldots(i_p\,j_p)$ and denoting by $\gamma$ the canonical full cycle $(2p\,\ldots\,1)$, the latter condition can be written as $\forall\ 1\leq i\leq 2p,\ l_{\gamma^{-1}\lambda(i)}=l_{i}$. So in fact, $\flat(\lambda)= \sharp(\gamma^{-1}\lambda)$ and the expression above becomes
\[ \E_{G\sim GUE(n)} \tr\left(G^{2p}\right) = \underset{\lambda\in\mathfrak{P}^{(2)}(2p)}{\sum} n^{\sharp(\gamma^{-1}\lambda)}. \]
We thus have the so-called \textit{genus expansion} (see e.g.~\cite{NS}, Lecture 22)
\[ \E_{G\sim GUE(n)} \tr\left(G^{2p}\right) = \sum_{\delta=0}^{\lfloor p/2\rfloor} P(2p,\delta)n^{p+1-2\delta}, \]
where for each $0\leq\delta\leq\lfloor p/2\rfloor$, we defined $P(2p,\delta)$ as the number of pairings of $\{1,\ldots,2p\}$ having \textit{genus} $\delta$, i.e.~\[ P(2p,\delta)=\left|\{\lambda\in\mathfrak{P}^{(2)}(2p) \st \sharp(\gamma^{-1}\lambda)=p+1-2\delta\}\right|= \left|\{\lambda\in\mathfrak{P}^{(2)}(2p) \st \sharp(\gamma^{-1}\lambda)+\sharp(\lambda)=2p+1-2\delta\}\right|. \]
Equivalently, $P(2p,\delta)$ is the number of pairings of $\{1,\ldots,2p\}$ having a defect $2\delta$ of being on the geodesics between $\id$ and $\gamma$. Hence, $P(2p,0)$ is the number of pairings of $\{1,\ldots,2p\}$ lying exactly on the geodesics between $\id$ and $\gamma$, i.e.~ the number of non-crossing pair partitions of $\{1,\ldots,2p\}$. So $P(2p,0)=\left|NC^{(2)}(2p)\right|=\mathrm{Cat}_p$, and we recover the well-known asymptotic estimate
\[ \E_{G\sim GUE(n)} \tr\left(G^{2p}\right) \underset{n\rightarrow+\infty}{\sim} \mathrm{Cat}_p\, n^{p+1} .\]


\subsection{Moments of Wishart matrices}
\label{appendix:wick-wishart}

A second important application of Lemma \ref{lemma:wick} is to the computation of the moments of matrices from the Wishart Ensemble. In such case, a graphical way of visualising the Wick formula has been developed in \cite{CN}, to which the reader is referred for further details and proofs, a brief summary only being provided here.

In the graphical formalism, a matrix $X:\C^m\rightarrow\C^n$ is represented by a ``box'' with two ``gates'', one specifying the size $m$ at its entrance and the other specifying the size $n$ at its exit. For $X:\C^m\rightarrow\C^n$  and $Y:\C^n\rightarrow\C^m$, the product $XY:\C^n\rightarrow\C^n$ is represented by a wire connecting the exit of $Y$ to the entrance of $X$. For $Z:\C^m\rightarrow\C^m$, the trace $\tr(Z)$ is represented by a wire connecting the exit and the entrance of $Z$.

Let $W$ be a $(n,s)$-Wishart matrix, i.e.~$W=GG^{\dagger}$ with $G$ a $n\times s$ matrix with independent complex normal entries. Representing by $\blacklozenge$ a $n$-dimensional gate and by $\blacktriangledown$ a $s$-dimensional gate, the quantity $\tr(W^p)$ is then graphically represented by $p$ boxes $G$ and $p$ boxes $G^{\dagger}$ connected by wires in the following way.

\begin{center}
\begin{tikzpicture} [scale=0.8]
\node[draw=lightgray, minimum height=0.8cm, minimum width=0.8cm, fill=lightgray, =white] (G1) at (0,0) {$G$};
\node[draw=lightgray, minimum height=0.8cm, minimum width=0.8cm, fill=lightgray] (G1*) at (1.5,0) {$G^{\dagger}$};
\node[draw=lightgray, minimum height=0.8cm, minimum width=0.8cm, fill=lightgray] (Gp) at (5.5,0) {$G$};
\node[draw=lightgray, minimum height=0.8cm, minimum width=0.8cm, fill=lightgray] (Gp*) at (7,0) {$G^{\dagger}$};
\draw(G1.west) node{$\blacklozenge$}; \draw(G1.east) node{$\blacktriangledown$};
\draw(G1*.east) node{$\blacklozenge$}; \draw(G1*.west) node{$\blacktriangledown$};
\draw(Gp.west) node{$\blacklozenge$}; \draw(Gp.east) node{$\blacktriangledown$};
\draw(Gp*.east) node{$\blacklozenge$}; \draw(Gp*.west) node{$\blacktriangledown$};
\draw (G1.east) -- (G1*.west); \draw (Gp.east) -- (Gp*.west);
\draw (G1.west) to[out=155,in=25] (Gp*.east);
\draw (G1*.east) -- (2.5,0);
\draw (Gp.west) -- (4.5,0);
\draw[dotted] (3,0) -- (4,0);
\end{tikzpicture}
\end{center}

For any $\alpha\in\mathfrak{S}(p)$, we will denote by $\mathcal{G}_{\alpha}$ the diagram obtained from the one above by ``erasing'' the boxes, just keeping their gates, and then connecting, for each $1\leq i\leq p$, the entrance of the $i^{th}$ box $G$ to the exit of the $\alpha(i)^{th}$ box $G^{\dagger}$, and the exit of the $i^{th}$ box $G$ to the entrance of the $\alpha(i)^{th}$ box $G^{\dagger}$. Doing so, $\sharp(\gamma^{-1}\alpha)$ loops connecting $n$-dimensional gates and $\sharp(\alpha)$ loops connecting $s$-dimensional gates are obtained. And the graphical version of the Wick formula tells us that
\[ \E_{W\sim\mathcal{W}_{n,s}}\tr(W^p) = \sum_{\alpha\in\mathfrak{S}(p)} \mathcal{D}_{\alpha} = \sum_{\alpha\in\mathfrak{S}(p)} n^{\sharp(\gamma^{-1}\alpha)}s^{\sharp(\alpha)} .\]

In the special case where $s=n$, this can be rewritten as a so-called \textit{genus expansion} (see e.g.~\cite{CN})
\[ \E_{W\sim\mathcal{W}_{n,n}}\tr(W^p) = \sum_{\delta=0}^{\lfloor p/2\rfloor}S(p,\delta)n^{p+1-2\delta}, \]
where for each $0\leq\delta\leq\lfloor p/2\rfloor$, we defined $S(p,\delta)$ as the number of permutations on $\{1,\ldots,p\}$ having \textit{genus} $\delta$, i.e.~$S(p,\delta)= \left| \{\alpha\in\mathfrak{S}(p) \st \sharp(\gamma^{-1}\alpha)+\sharp(\alpha)=p+1-2\delta\} \right|$. Since $\{\alpha\in\mathfrak{S}(p) \st \sharp(\gamma^{-1}\alpha)+\sharp(\alpha)=p+1\}=NC(p)$, we have $S(p,0)=\mathrm{Cat}_p$ and hence recover the well-known asymptotic estimate
\[ \E_{W\sim\mathcal{W}_{n,n}} \tr\left(W^{p}\right) \underset{n\rightarrow+\infty}{\sim} \mathrm{Cat}_p\, n^{p+1} .\]


\section{One needed combinatorial fact: relating the number of cycles in some specific permutations on either $[p]\times[k]$ or $[p]$}
\label{appendix:technical}

Let $\alpha\in\mathfrak{S}(p)$ and $f:[p]\rightarrow[k]$. We define $\hat{\alpha}_f$ on $[p]\times[k]$ as
\[ \forall\ (i,r)\in[p]\times[k],\ \hat{\alpha}_f(i,r)=
\begin{cases} (\alpha(i),f(\alpha(i)))\ \text{if}\ r=f(i)\\
(i,r)\ \text{if}\ r\neq f(i)\end{cases}. \]
We also define $\hat{\gamma}$ on $[p]\times[k]$ as $(\gamma,\id)$, where $\gamma\in\mathfrak{S}(p)$ is the canonical full cycle $(p\,\ldots\,1)$.

We would like to understand what is the number of cycles in $\hat{\gamma}^{-1}\hat{\alpha}_f$. For that, it will be convenient to do a bit of rewriting. Let us first extend the definition of $\hat{\alpha}_f$ and $\hat{\gamma}$ to $[p]\times\left(\{0\}\cup[k]\right)$. We shall denote by $\bar{\alpha}_f$ and $\bar{\gamma}$ the respective extensions. Note that since $f$ takes values in $[k]$, we have $\bar{\alpha}_f(i,0)=(i,0)$ for all $i\in[p]$.

We will now make two easy observations.

\begin{fact}
\label{fact:alpha_f}
For any $f:[p]\rightarrow[k]$, define for each $i\in[p]$, $\bar{\tau}_f^{(i)}$ as the transposition on $[p]\times\left(\{0\}\cup[k]\right)$ which swaps $(i,0)$ and $(i,f(i))$, and set $\bar{\beta}_f=\bar{\tau}_f^{(1)}\cdots\bar{\tau}_f^{(p)}$. We then have, for any $\alpha\in\mathfrak{S}(p)$,
\[ \bar{\alpha}_f=\bar{\beta}_f^{-1}\bar{\alpha}'\bar{\beta}_f,\ \text{where}\ \forall\ (i,r)\in[p]\times\left(\{0\}\cup[k]\right),\ \bar{\alpha}'(i,r)=\begin{cases} (\alpha(i),r)\ \text{if}\ r=0\\
(i,r)\ \text{if}\ r\neq 0 \end{cases}. \]
\end{fact}

The advantage of expressing $\bar{\alpha}_f$ in this way is that $\bar{\alpha}'$ is particularly simple: it acts as $\alpha\times\id$ on $[p]\times\{0\}$ and does nothing on $[p]\times[k]$. Furthermore, due to the cyclicity of $\sharp(\cdot)$, a direct consequence of Fact \ref{fact:alpha_f} is that
$\sharp(\bar{\gamma}^{-1}\bar{\alpha}_f) = \sharp(\bar{\gamma}_f^{-1}\bar{\alpha}')$, where $\bar{\gamma}_f=\bar{\beta}_f\bar{\gamma}\bar{\beta}_f^{-1}$. It may then be easily checked that $\bar{\gamma}_f$ decomposes into $k+1$ disjoint cycles as stated in Fact \ref{fact:gamma_f} below.

\begin{fact}
\label{fact:gamma_f}
For any $f:[p]\rightarrow[k]$, we have
\[ \bar{\gamma}_f=\bar{c}_1\cdots\bar{c}_k\bar{c}, \]
with $\bar{c}=\left((p,f(p))\ldots(1,f(1))\right)$, and for each $r\in[k]$, $\bar{c}_r=\left((p,s_r(p))\ldots(1,s_r(1))\right)$, where for each $i\in[p]$, $s_r(i)=0$ if $f(i)=r$ and $s_r(i)=r$ if $f(i)\neq r$.
\end{fact}

\begin{example} For the sake of concreteness, let us have a look at a simple example. In the case where $p=4$, $k=3$, and $f$ is defined by $f(1)=f(2)=f(4)=1$, $f(3)=2$, we obtain that the cycles in $\bar{\gamma}_f$ are $\bar{c}_1=((4,0)(3,1)(2,0)(1,0))$, $\bar{c}_2=((4,2)(3,0)(2,2)(1,2))$, $\bar{c}_3=((4,3)(3,3)(2,3)(1,3))$, $\bar{c}=((4,1)(3,2)(2,1)(1,1))$. This is schematically represented in Figure \ref{fig:bargamma_f}, where the elements in $\bar{c}_i$ are marked by ``$i$'', for $i\in\{1,2,3\}$, and the elements in $\bar{c}$ are marked by ``$\bullet$''.

\begin{figure}[h]
\caption{$f:[4]\rightarrow[3]$ such that $f^{-1}(1)=\{1,2,4\}$, $f^{-1}(2)=\{3\}$, $f^{-1}(3)=\emptyset$. }
\begin{center}
\begin{tikzpicture} [scale=0.7]
\draw (6,4.3) node {$ $};
\draw [<->] (0,-0.5) -- (4,-0.5); \draw [<->] (-0.5,0) -- (-0.5,4);
\draw (2,-1) node {$i\in[4]$}; \draw (-2,2) node {$r\in\{0\}\cup[3]$};
\draw (0,0) -- (4,0); \draw (0,1) -- (4,1); \draw (0,2) -- (4,2); \draw (0,3) -- (4,3); \draw (0,4) -- (4,4);
\draw (0,0) -- (0,4); \draw (1,0) -- (1,4); \draw (2,0) -- (2,4); \draw (3,0) -- (3,4); \draw (4,0) -- (4,4);
\draw (0.5,0.5) node {$1$}; \draw (1.5,0.5) node {$1$}; \draw (2.5,0.5) node {$2$}; \draw (3.5,0.5) node {$1$};
\draw (0.5,1.5) node {$\bullet$}; \draw (1.5,1.5) node {$\bullet$}; \draw (2.5,1.5) node {$1$}; \draw (3.5,1.5) node {$\bullet$};
\draw (0.5,2.5) node {$2$}; \draw (1.5,2.5) node {$2$}; \draw (2.5,2.5) node {$\bullet$}; \draw (3.5,2.5) node {$2$};
\draw (0.5,3.5) node {$3$}; \draw (1.5,3.5) node {$3$}; \draw (2.5,3.5) node {$3$}; \draw (3.5,3.5) node {$3$};
\end{tikzpicture}
\end{center}
\label{fig:bargamma_f}
\end{figure}
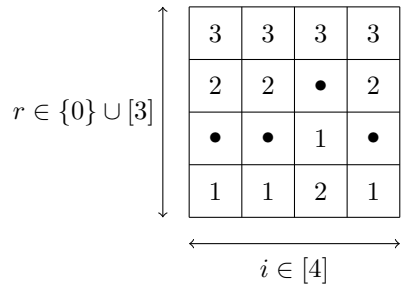

\end{example}

\begin{lemma}
\label{lemma:bargamma_f-gamma_f}
Let $f:[p]\rightarrow[k]$ and define $\gamma_f\in\mathfrak{S}(p)$ as $\gamma_f=\gamma_{f=1}\cdots\gamma_{f=k}$, where for each $r\in[k]$, $\gamma_{f=r}$ is the canonical full cycle on $f^{-1}(r)$. Then, for any $\alpha\in\mathfrak{S}(p)$,
\begin{equation}\label{eq:bargamma_f-gamma_f} \sharp(\bar{\gamma}_f^{-1}\bar{\alpha}') = \sharp(\gamma_f^{-1}\alpha) + 1 + k - |\im(f)|.\end{equation}
\end{lemma}

\begin{proof} In $\bar{\gamma}_f^{-1}\bar{\alpha}'$ there are, first of all:\\
$\bullet$ $k-|\im(f)|$ cycles of the form $((1,r)\ldots(p,r))$ for $r\in[k]\setminus\im (f)$, because for any $(i,r)\in[p]\times\left([k]\setminus\im (f)\right)$, $\bar{\gamma}_f^{-1}\bar{\alpha}'(i,r)=\bar{\gamma}_f^{-1}(i,r)=(\gamma^{-1}(i),r)$.\\
$\bullet$ $1$ cycle $((1,f(1))\ldots(p,f(p)))$, because for any $i\in[p]$, $\bar{\gamma}_f^{-1}\bar{\alpha}'(i,f(i))=\bar{\gamma}_f^{-1}(i,f(i))=(\gamma^{-1}(i),f(\gamma^{-1}(i)))$.\\
For the cycles in $\bar{\gamma}_f^{-1}\bar{\alpha}'$ which belong to none of these two categories, there are two crucial observations to be made. First, for any $i,j\in[p]$, $(i,0)$ and $(j,0)$ belong to the same cycle of $\bar{\gamma}_f^{-1}\bar{\alpha}'$ if and only if $i$ and $j$ belong to the same cycle of $\gamma_f^{-1}\alpha$. And second, for each $i\in[p]$ and $r\in[k]\setminus\{f(\alpha(i))\}$, there exists $j\in[p]$ such that $(i,r)$ belongs to the same cycle of $\bar{\gamma}_f^{-1}\bar{\alpha}'$ as $(j,0)$. Indeed, for any $i\in[p]$, we have on the one hand
\[ \gamma_f^{-1}\alpha(i)= (\gamma^{-1})^{L+1}\alpha(i),\ \text{with}\ L=\inf\{ l\geq 0 \st f((\gamma^{-1})^{l+1}\alpha(i))= f(\alpha(i))\}. \]
While we have on the other hand,
\begin{align*}
f(\gamma^{-1}\alpha(i))=f(\alpha(i))\ & \Rightarrow\ \bar{\gamma}_f^{-1}\bar{\alpha}'(i,0) =(\gamma^{-1}\alpha(i),0) =(\gamma_f^{-1}\alpha(i),0), \\
f(\gamma^{-1}\alpha(i))\neq f(\alpha(i))\ & \Rightarrow\ \begin{cases}\forall\ 0\leq l\leq L-1,\ (\bar{\gamma}_f^{-1}\bar{\alpha}')^l(i,0)=((\gamma^{-1})^{l}\gamma^{-1}\alpha(i),f(\alpha(i)))\\
(\bar{\gamma}_f^{-1}\bar{\alpha}')^{L}(i,0)=((\gamma^{-1})^{L}\gamma^{-1}\alpha(i),0) =(\gamma_f^{-1}\alpha(i),0) \end{cases}.
\end{align*}
So there are in fact exactly $\sharp(\gamma_f^{-1}\alpha)$ remaining cycles in $\bar{\gamma}_f^{-1}\bar{\alpha}'$.
\end{proof}

\begin{example} Looking at the same example as before, namely $p=4$, $k=3$, and $f$ such that $f^{-1}(1)=\{1,2,4\}$, $f^{-1}(2)=\{3\}$, $f^{-1}(3)=\emptyset$, we see that, for $\alpha=(14)(23)$, the cycles in $\bar{\gamma}_f^{-1}\bar{\alpha}'$ are:\\
$\bullet$ $((1,3)(2,3)(3,3)(4,3))$, because $3\notin\im (f)$.\\
$\bullet$ $((1,1)(2,1)(3,2)(4,1))$, because $f(1)=1$, $f(2)=1$, $f(3)=2$ and $f(4)=1$.\\
$\bullet$ $(1,0)$ and $((2,0)(2,2)(1,2)(4,2)(3,0)(4,0)(3,1))$, corresponding to the cycles $(1)$ and $(2,3,4)$ in $\gamma_f^{-1}\alpha$.
\end{example}

Putting together these preliminary technical results, we straightforwardly obtain Proposition \ref{prop:geodesic-alpha_f} below.

\begin{proposition}
\label{prop:geodesic-alpha_f}
Let $f:[p]\rightarrow[k]$ and define $\gamma_f\in\mathfrak{S}(p)$ as $\gamma_f=\gamma_{f=1}\cdots\gamma_{f=k}$, where for each $r\in[k]$, $\gamma_{f=r}$ is the canonical full cycle on $f^{-1}(r)$. Then, for any $\alpha\in\mathfrak{S}(p)$,
\begin{equation}\label{eq:geodesic-alpha_f} \sharp(\hat{\gamma}^{-1}\hat{\alpha}_f) = \sharp(\gamma_f^{-1}\alpha) + k - |\im(f)|. \end{equation}
\end{proposition}

\begin{proof} This is a direct consequence of Lemma \ref{lemma:bargamma_f-gamma_f}, just noticing that $\sharp(\bar{\gamma}^{-1}\bar{\alpha}_f)=\sharp(\hat{\gamma}^{-1}\hat{\alpha}_f)+1$.
\end{proof}

\section{Proof of the moments expression for ``modified'' GUE matrices (mean-width of the set of $k$-extendible states)}
\label{appendix:gaussian}

The goal of this Appendix is to generalize the methodology described in Appendix \ref{appendix:wick-gaussian} in order to compute the $2p$-order moments of the matrix $\sum_{j=1}^k\widetilde{G}_{\A\B^k}(j)$. Recall that this issue arises when trying to estimate the mean width of the set of $k$-extendible states. We are thus dealing here, not with standard GUE matrices, but with $d^2$-dimensional GUE matrices which are tensorized with $d^{k-1}$-dimensional identity matrices.

For any $i_1,\ldots,i_{2p}\in[k]$, we can write
\[ \mathrm{Tr}\left[\underset{j=1}{\overset{2p}{\overrightarrow{\prod}}}\widetilde{G}_{\A\B^k}(i_j)\right] = \underset{\vec{l}_1,\ldots,\vec{l}_{2p}\in[d]^{k+1}}{\sum} \widetilde{G}_{\A\B^k}(i_1)_{\vec{l}_1,\vec{l}_2}\cdots\widetilde{G}_{\A\B^k}(i_{2p})_{\vec{l}_{2p},\vec{l}_1} ,\]
where for each $j\in[2p]$ and each $\vec{l}_j=\left(a_j,b_j^1,\ldots,b_j^k\right),\vec{l}_{j+1}=\left(a_{j+1},b_{j+1}^1,\ldots,b_{j+1}^k\right)\in[d]^{k+1}$, we have \[ \widetilde{G}_{\A\B^k}(i_j)_{\vec{l}_j,\vec{l}_{j+1}}=G_{(a_j,b_j^{i_j}),(a_{j+1},b_{j+1}^{i_j})} \delta_{\vec{b}_j\setminus b_j^{i_j}=\vec{b}_{j+1}\setminus b_{j+1}^{i_j}}. \]
Consequently, for any $f:[2p]\rightarrow[k]$, we have
\[ \mathrm{Tr}\left[\underset{i=1}{\overset{2p}{\overrightarrow{\prod}}}\widetilde{G}_{\A\B^k}(f(i))\right] = \underset{a_1,\ldots,a_{2p}\in[d]}{\sum}\ \underset{\vec{b}_1,\ldots,\vec{b}_{2p}\in I_f}{\sum} G_{(a_1,b_1^{f(1)}),(a_2,b_2^{f(1)})}\cdots G_{(a_{2p},b_{2p}^{f(2p)}),(a_1,b_1^{f(2p)})}, \]
where $I_f=\left\{\vec{b}_1,\ldots,\vec{b}_{2p}\in[d]^k \st \forall\ i\in[2p],\ \forall\ r\in[k]\setminus\{f(i)\},\ b_{i+1}^r=b_i^r\right\}$.

What we therefore get by the Wick formula for Gaussian matrices is that, for any $f:[2p]\rightarrow[k]$,
\[ \mathbf{E}_{G_{\A\B}\sim GUE(d^2)} \mathrm{Tr} \left[\underset{i=1}{\overset{2p}{\overrightarrow{\prod}}}\widetilde{G}_{\A\B^k}(f(i))\right] = \underset{\lambda\in\mathfrak{P}^{(2)}(2p)}{\sum} d^{\flat(\lambda)+\flat(\hat{\lambda}_f)}, \]
where for each pair partition $\lambda=\{\{i_1,j_1\},\ldots,\{i_p,j_p\}\}$ of $\{1,\ldots,2p\}$, $\flat(\lambda)$ is the number of free parameters $a_1,\ldots,a_{2p}\in[d]$ when imposing that $\forall\ 1\leq m\leq p,\ a_{i_m+1}=a_{j_m},\ a_{j_m+1}=a_{i_m}$, and $\flat(\hat{\lambda}_f)$ is the number of free parameters $\vec{b}_1,\ldots,\vec{b}_{2p}\in I_f$ when imposing that $\forall\ 1\leq m\leq p,\ b_{i_m+1}^{f(i_m)}=b_{j_m}^{f(j_m)},\ b_{j_m+1}^{f(j_m)}=b_{i_m}^{f(i_m)}$. As noticed before, identifying the pair partition $\{\{i_1,j_1\},\ldots,\{i_p,j_p\}\}$ with the pairing $(i_1\,j_1)\,\ldots\,(i_p\,j_p)$ and denoting by $\gamma$ the canonical full cycle $(2p\,\ldots\,1)$, the latter conditions may be written as
\[ \forall\ i\in[2p],\ a_{\gamma^{-1}\lambda(i)}=a_{i}\ \ \text{and}\ \ \forall\ r\in[k],\ \begin{cases} b_{\gamma^{-1}\lambda(i)}^{f(\lambda(i))}=b_{i}^r\ \text{if}\ r=f(i) \\ b_{\gamma^{-1}(i)}^r=b_{i}^r\ \text{if}\ r\neq f(i) \end{cases}. \]
So in fact, $\flat(\lambda)=\sharp(\gamma^{-1}\lambda)$ and $\flat(\hat{\lambda}_f)=\sharp(\hat{\gamma}^{-1}\hat{\lambda}_f)$, where
\[ \forall\ (i,r)\in[2p]\times[k],\ \hat{\gamma}(i,r)=(\gamma(i),r)\ \ \text{and}\ \ \hat{\lambda}_f(i,r)=\begin{cases}  (\lambda(i),f(\lambda(i)))\ \text{if}\ r=f(i) \\ (i,r)\ \text{if}\ r\neq f(i) \end{cases}. \]
What is more, we know by Proposition \ref{prop:geodesic-alpha_f} that $\sharp(\hat{\gamma}^{-1}\hat{\lambda}_f) = \sharp(\gamma_f^{-1}\lambda) + k - |\im(f)|$, where $\gamma_f$ is the product of the canonical full cycles on the level sets of $f$.

Let us summarize.

\begin{proposition}
\label{prop:gaussian-p-preliminary}
For any $d\in\N$ and any $p\in\N$, we have
\begin{align*} \mathbf{E}_{G_{\A\B}\sim GUE(d^2)} \mathrm{Tr}\left[\left(\sum_{j=1}^k\widetilde{G}_{\A\B^k}(j)\right)^{2p}\right] = & \sum_{f:[2p]\rightarrow[k]}\mathbf{E}_{G_{\A\B}\sim GUE(d^2)} \mathrm{Tr}\left[\underset{i=1}{\overset{2p}{\overrightarrow{\prod}}}\widetilde{G}_{\A\B^k}(f(i))\right] \\
= & \underset{f:[2p]\rightarrow[k]}{\sum} \underset{\lambda\in\mathfrak{P}^{(2)}(2p)}{\sum} d^{\sharp(\gamma^{-1}\lambda)+\sharp(\gamma_f^{-1}\lambda) + k - |\im(f)|}. \end{align*}
\end{proposition}

\section{Proof of the moments expression for partial transposition of ``modified'' GUE matrices (mean-width of the set of $k$-PPT-extendible states)}
\label{appendix:gaussian-gamma}

The goal of this Appendix is to compute the $2p$-order moments of $\sum_{j=1}^k \widetilde{G}_{\A\B^k}(j)^{\Gamma}$, where $\Gamma$ stands here for the partial transposition over the $\lceil k/2\rceil$ last $\mathrm{B}$ subsystems. Recall that this issue arises when trying to estimate the mean width of the set of $k$-PPT-extendible states.

Using the same notation as in Appendix \ref{appendix:gaussian}, and reasoning in a completely analogous way, we have that, for any $f:[2p]\rightarrow[k]$,
\[ \mathrm{Tr}\left[\underset{i=1}{\overset{2p}{\overrightarrow{\prod}}}\widetilde{G}_{\A\B^k}(f(i))^{\Gamma}\right] = \underset{a_1,\ldots,a_{2p}\in[d]}{\sum}\ \underset{\vec{b}_1,\ldots,\vec{b}_{2p}\in I_f}{\sum} G_{(a_1,b_{x_1}^{f(1)}),(a_2,b_{\bar{x}_1}^{f(1)})}\cdots G_{(a_{2p},b_{x_{2p}}^{f(2p)}),(a_1,b_{\bar{x}_{2p}}^{f(2p)})}, \]
where for each $1\leq i\leq 2p$, $x_i=\begin{cases} i\ \text{if}\ f(i)\leq \lfloor k/2\rfloor \\ i+1\ \text{if}\ f(i)> \lfloor k/2\rfloor \end{cases}$ and $\bar{x}_i=\begin{cases} i+1\ \text{if}\ f(i)\leq \lfloor k/2\rfloor \\ i\ \text{if}\ f(i)> \lfloor k/2\rfloor \end{cases}$.

What we therefore get by the Wick formula for Gaussian matrices is that, for any $f:[2p]\rightarrow[k]$,
\[ \mathbf{E}_{G_{\A\B}\sim GUE(d^2)} \mathrm{Tr}\left[\underset{i=1}{\overset{2p}{\overrightarrow{\prod}}}\widetilde{G}_{\A\B^k}(f(i))^{\Gamma}\right] = \underset{\lambda\in\mathfrak{P}^{(2)}(2p)}{\sum} d^{\flat(\lambda)+\flat(\check{\lambda}_f)}, \]
where for each pair partition $\lambda=\{\{i_1,j_1\},\ldots,\{i_p,j_p\}\}$ of $\{1,\ldots,2p\}$,  $\flat(\check{\lambda}_f)$ is the number of free parameters $\vec{b}_1,\ldots,\vec{b}_{2p}\in [d]^k$ when imposing that for all $i\in[2p]$, first $b_{\gamma^{-1}(i)}^r=b_{i}^r$ if $r\neq f(i)$, and second the one condition $b_{\gamma^{-1}\lambda(i)}^{f(\lambda(i))}=b_{i}^{f(i)}$ if $f(i),f(\lambda(i))\leq\lfloor k/2\rfloor$ or $f(i),f(\lambda(i))>\lfloor k/2\rfloor$, while the two conditions $b_{\gamma^{-1}\lambda(i)}^{f(\lambda(i))}=b_{i}^{f(i)}$ and $b_{\lambda(i)}^{f(\lambda(i))}=b_{\gamma^{-1}(i)}^{f(i)}$ if $f(i)\leq\lfloor k/2\rfloor$, $f(\lambda(i))>\lfloor k/2\rfloor$ or $f(i)>\lfloor k/2\rfloor$, $f(\lambda(i))\leq\lfloor k/2\rfloor$.

Let us rephrase what we just established. Fix $\lambda\in\mathfrak{P}^{(2)}(2p)$. For functions $f:[2p]\rightarrow\{1,\ldots,\lfloor k/2\rfloor\}\equiv\left[\lfloor k/2\rfloor\right]$ or $f:[2p]\rightarrow\{\lfloor k/2\rfloor+1,\ldots,k\}\equiv\left[\lceil k/2\rceil\right]$, the number of free parameters associated to the pair $(\lambda,f)$ is the same as the one observed in Appendix \ref{appendix:gaussian}. On the contrary, for functions $f$ which are not of this form, extra matching conditions are imposed. So these will for sure not contribute to the dominating term in the expansion of $\E\mathrm{Tr}\left[\left(\sum_{j=1}^k\widetilde{G}_{\A\B^k}(j)^{\Gamma}\right)^{2p}\right]$ into powers of $d$. Consequently, we have the asymptotic estimate
\[ \E_{G_{\A\B}\sim GUE(d^2)}\mathrm{Tr}\left[\left(\sum_{j=1}^k\widetilde{G}_{\A\B^k}(j)^{\Gamma}\right)^{2p}\right]
 \underset{d\rightarrow+\infty}{\sim} \underset{\lambda\in\mathfrak{P}^{(2)}(2p)}{\sum}\, \underset{f:[2p]\rightarrow\left[\lfloor k/2\rfloor\right]\,\text{or}\,\left[\lceil k/2\rceil\right]}{\sum} d^{\sharp(\gamma^{-1}\lambda)+\sharp(\gamma_f^{-1}\lambda)+k-|\im(f)|}. \]

\section{Proof of the moments expression for ``modified'' Wishart matrices ($k$-extendibility of random-induced states)}
\label{appendix:wishart}

The goal of this Appendix is to generalize the methodology described in Appendix \ref{appendix:wick-wishart} in order to compute the $p$-order moments of the matrix $\sum_{j=1}^k\widetilde{W}_{\A\B^k}(j)$. Recall that this issue arises when trying to characterize $k$-extendibility of random-induced states. We are thus dealing here, not with standard Wishart matrices, but with $(d^2,s)$-Wishart matrices which are tensorized with $d^{k-1}$-dimensional identity matrices.

Representing by $\bullet$ a $d$-dimensional gate and by $\blacktriangledown$ a $s$-dimensional gate, the matrix $\widetilde{W}_{\A\B^k}(1)$, for instance, may be graphically represented as in Figure \ref{fig:W_AB^k}.

\begin{figure}[h]
\caption{$\widetilde{W}_{\A\B^k}(1)=X_{\A\B^k}(1)X_{\A\B^k}(1)^{\dagger}$, with $X_{\A\B^k}(1)=G_{\A\B_1}\otimes\Id_{\B_2\ldots \B_k}$}
\begin{center}
\begin{tikzpicture} [scale=0.9]
\node[draw=lightgray, minimum height=2.4cm, minimum width=2cm, fill=lightgray] at (0,-1.6) { };
\node[draw=lightgray, minimum height=2.4cm, minimum width=2cm, fill=lightgray] at (3,-1.6) { };
\node[draw=lightgray, minimum height=0.8cm, minimum width=2cm, fill=lightgray] (X) at (0,0) { };
\node[draw=lightgray, minimum height=0.8cm, minimum width=2cm, fill=lightgray] (X*) at (3,0) { };
\draw (0,0.5) node{} node[above]{$X_{\A\B^k}(1)$}; \draw (3,0.5) node{} node[above]{$X_{\A\B^k}(1)^{\dagger}$};
\draw(X.north west) node{$\bullet$} node[left]{$\A$}; \draw(X.south west) node{$\bullet$} node[left]{$B_1$}; \draw(X.east) node{$\blacktriangledown$};
\draw (-1.1,-1.2) node{$\bullet$} node[left]{$B_2$}; \draw (-1.1,-2.8) node{$\bullet$} node[left]{$B_k$}; \draw (1.1,-1.2) node{$\bullet$}; \draw (1.1,-2.8) node{$\bullet$};
\draw(X*.north east) node{$\bullet$} node[right]{$\A$}; \draw(X*.south east) node{$\bullet$} node[right]{$B_1$}; \draw(X*.west) node{$\blacktriangledown$};
\draw (1.9,-1.2) node{$\bullet$}; \draw (1.9,-2.8) node{$\bullet$}; \draw (4.1,-1.2) node{$\bullet$} node[right]{$B_2$}; \draw (4.1,-2.8) node{$\bullet$} node[right]{$B_k$};
\draw (X.east) -- (X*.west); \draw (-1,-1.2) -- (4,-1.2); \draw (-1,-2.8) -- (4,-2.8);
\draw[dotted] (-1.3,-1.5) -- (-1.3,-2.5);  \draw[dotted] (4.3,-1.5) -- (4.3,-2.5);
\end{tikzpicture}
\end{center}
\label{fig:W_AB^k}
\end{figure}
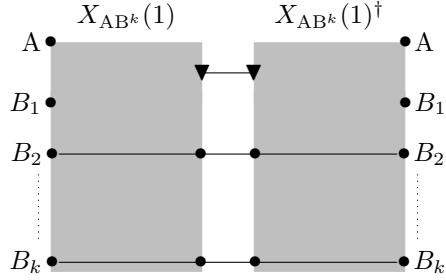

The products $\widetilde{W}_{\A\B^k}(1)\widetilde{W}_{\A\B^k}(1)$ and $\widetilde{W}_{\A\B^k}(1)\widetilde{W}_{\A\B^k}(2)$, for instance, are then obtained by the wirings represented in Figure \ref{fig:prod-W_AB^k}.

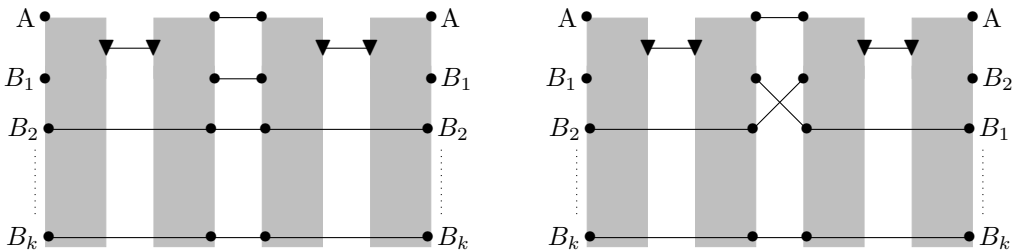
\begin{figure}[h]
\caption{$\widetilde{W}_{\A\B^k}(1)\widetilde{W}_{\A\B^k}(1)$ (on the left) and $\widetilde{W}_{\A\B^k}(1)\widetilde{W}_{\A\B^k}(2)$ (on the right)}
\begin{center}
\begin{tikzpicture} [scale=0.9]
\node[draw=lightgray, minimum height=2.4cm, minimum width=0.8cm, fill=lightgray] at (0,-1.6) { };
\node[draw=lightgray, minimum height=2.4cm, minimum width=0.8cm, fill=lightgray] at (1.6,-1.6) { };
\node[draw=lightgray, minimum height=0.8cm, minimum width=0.8cm, fill=lightgray] (X) at (0,0) { };
\node[draw=lightgray, minimum height=0.8cm, minimum width=0.8cm, fill=lightgray] (X*) at (1.6,0) { };
\draw(X.north west) node{$\bullet$} node[left]{$\A$}; \draw(X.south west) node{$\bullet$} node[left]{$B_1$}; \draw(X.east) node{$\blacktriangledown$};
\draw (-0.4,-1.2) node{$\bullet$} node[left]{$B_2$}; \draw (-0.4,-2.8) node{$\bullet$} node[left]{$B_k$};
\draw(X*.north east) node{$\bullet$}; \draw(X*.south east) node{$\bullet$}; \draw(X*.west) node{$\blacktriangledown$};
\draw (2,-1.2) node{$\bullet$}; \draw (2,-2.8) node{$\bullet$};
\draw (X.east) -- (X*.west);
\draw[dotted] (-0.6,-1.5) -- (-0.6,-2.5);
\node[draw=lightgray, minimum height=2.4cm, minimum width=0.8cm, fill=lightgray] at (3.2,-1.6) { };
\node[draw=lightgray, minimum height=2.4cm, minimum width=0.8cm, fill=lightgray] at (4.8,-1.6) { };
\node[draw=lightgray, minimum height=0.8cm, minimum width=0.8cm, fill=lightgray] (Y) at (3.2,0) { };
\node[draw=lightgray, minimum height=0.8cm, minimum width=0.8cm, fill=lightgray] (Y*) at (4.8,0) { };
\draw(Y.north west) node{$\bullet$}; \draw(Y.south west) node{$\bullet$}; \draw(Y.east) node{$\blacktriangledown$};
\draw (2.8,-1.2) node{$\bullet$}; \draw (2.8,-2.8) node{$\bullet$};
\draw(Y*.north east) node{$\bullet$} node[right]{$\A$}; \draw(Y*.south east) node{$\bullet$} node[right]{$B_1$}; \draw(Y*.west) node{$\blacktriangledown$};
\draw (5.2,-1.2) node{$\bullet$} node[right]{$B_2$}; \draw (5.2,-2.8) node{$\bullet$} node[right]{$B_k$};
\draw (Y.east) -- (Y*.west);
\draw (-0.4,-1.2) -- (5.2,-1.2); \draw (-0.4,-2.8) -- (5.2,-2.8);
\draw[dotted] (5.4,-1.5) -- (5.4,-2.5);
\draw (X*.north east) -- (Y.north west); \draw (X*.south east) -- (Y.south west);

\node[draw=lightgray, minimum height=2.4cm, minimum width=0.8cm, fill=lightgray] at (8,-1.6) { };
\node[draw=lightgray, minimum height=2.4cm, minimum width=0.8cm, fill=lightgray] at (9.6,-1.6) { };
\node[draw=lightgray, minimum height=0.8cm, minimum width=0.8cm, fill=lightgray] (X') at (8,0) { };
\node[draw=lightgray, minimum height=0.8cm, minimum width=0.8cm, fill=lightgray] (X'*) at (9.6,0) { };
\draw(X'.north west) node{$\bullet$} node[left]{$\A$}; \draw(X'.south west) node{$\bullet$} node[left]{$B_1$}; \draw(X'.east) node{$\blacktriangledown$};
\draw (7.6,-1.2) node{$\bullet$} node[left]{$B_2$}; \draw (7.6,-2.8) node{$\bullet$} node[left]{$B_k$};
\draw(X'*.north east) node{$\bullet$}; \draw(X'*.south east) node{$\bullet$}; \draw(X'*.west) node{$\blacktriangledown$};
\draw (10,-1.2) node{$\bullet$}; \draw (10,-2.8) node{$\bullet$};
\draw (X'.east) -- (X'*.west);
\draw (7.6,-1.2) -- (10,-1.2);
\draw[dotted] (7.4,-1.5) -- (7.4,-2.5);
\node[draw=lightgray, minimum height=2.4cm, minimum width=0.8cm, fill=lightgray] at (11.2,-1.6) { };
\node[draw=lightgray, minimum height=2.4cm, minimum width=0.8cm, fill=lightgray] at (12.8,-1.6) { };
\node[draw=lightgray, minimum height=0.8cm, minimum width=0.8cm, fill=lightgray] (Y') at (11.2,0) { };
\node[draw=lightgray, minimum height=0.8cm, minimum width=0.8cm, fill=lightgray] (Y'*) at (12.8,0) { };
\draw(Y'.north west) node{$\bullet$}; \draw(Y'.south west) node{$\bullet$}; \draw(Y'.east) node{$\blacktriangledown$};
\draw (10.8,-1.2) node{$\bullet$}; \draw (10.8,-2.8) node{$\bullet$};
\draw(Y'*.north east) node{$\bullet$} node[right]{$\A$}; \draw(Y'*.south east) node{$\bullet$} node[right]{$B_2$}; \draw(Y'*.west) node{$\blacktriangledown$};
\draw (13.2,-1.2) node{$\bullet$} node[right]{$B_1$}; \draw (13.2,-2.8) node{$\bullet$} node[right]{$B_k$};
\draw (Y'.east) -- (Y'*.west);
\draw (10.8,-1.2) -- (13.2,-1.2);
\draw[dotted] (13.4,-1.5) -- (13.4,-2.5);
\draw (7.6,-2.8) -- (13.2,-2.8);
\draw (X'*.north east) -- (Y'.north west);
\draw (X'*.south east) -- (10.8,-1.2); \draw (10,-1.2) -- (Y'.south west);
\end{tikzpicture}
\end{center}
\label{fig:prod-W_AB^k}
\end{figure}

So what we get by the graphical Wick formula for Wishart matrices is that for any $f:[p]\rightarrow[k]$,
\[ \mathbf{E}_{W_{\A\B}\sim\mathcal{W}_{d^2,s}} \mathrm{Tr}\left[\underset{i=1}{\overset{p}{\overrightarrow{\prod}}}\widetilde{W}_{\A\B^k}(f(i))\right] = \underset{\alpha\in\mathfrak{S}(p)}{\sum} \mathcal{D}_{f,\alpha} = \underset{\alpha\in\mathfrak{S}(p)}{\sum} d^{\sharp(\gamma^{-1}\alpha)}d^{\sharp(\hat{\gamma}^{-1}\hat{\alpha}_f)}s^{\sharp(\alpha)}, \]
where $\hat{\alpha}_f$ is defined by
\[ \hat{\alpha}_f:(i,r)\in [p]\times[k]\mapsto
\begin{cases}(\alpha(i),f(\alpha(i)))\ \text{if}\ r=f(i)\\
(i,r)\ \text{if}\ r\neq f(i)\end{cases}, \]
and where $\hat{\gamma}$ stands for $\gamma$ applied to the first argument. Indeed, for each $\alpha\in\mathfrak{S}(p)$, there are $\sharp(\gamma^{-1}\alpha)$ loops connecting the $d$-dimensional gates corresponding to $\A$, $\sharp(\hat{\gamma}^{-1}\hat{\alpha}_f)$ loops connecting the $d$-dimensional gates corresponding to $B_1,\ldots,B_k$, and $\sharp(\alpha)$ loops connecting $s$-dimensional gates. This is because for each $1\leq i\leq p$, on  subsystems $\A$ and $B_{f(i)}$, the entrances (respectively the exit) of the $i^{th}$ box $X_{\A\B^k}(f(i))$ are connected to the exits (respectively the entrance) of the $\alpha(i)^{th}$ box $X_{\A\B^k}(f(\alpha(i)))^{\dagger}$.

What happens in the special case $p=2$ and $k=2$ is detailed in Figures \ref{fig:p=2_k=2_1} and \ref{fig:p=2_k=2_2} below as an illustration.

\begin{figure}[h]
\caption{$f(1)=f(2)=1$. On the left, $\alpha=\id$: $\mathcal{D}_{f,\alpha}=d^3s^2$. On the right, $\alpha=(1\,2)$: $\mathcal{D}_{f,\alpha}=d^5s$.}
\begin{center}
\begin{tikzpicture} [scale=0.7]
\draw (-1,0) node {$B_2$}; \draw (-1,1) node {$B_1$}; \draw (-1,2) node {$\A$};
\draw (8,0) node {$B_2$}; \draw (8,1) node {$B_1$}; \draw (8,2) node {$\A$};
\draw (0,0) node {$\bullet$}; \draw (0,1) node {$\bullet$}; \draw (0,2) node {$\bullet$};
\draw (3,0) node {$\bullet$}; \draw (3,1) node {$\bullet$}; \draw (3,2) node {$\bullet$};
\draw (1,3/2) node {$\blacktriangledown$}; \draw (2,3/2) node {$\blacktriangledown$};
\draw (4,0) node {$\bullet$}; \draw (4,1) node {$\bullet$}; \draw (4,2) node {$\bullet$};
\draw (7,0) node {$\bullet$}; \draw (7,1) node {$\bullet$}; \draw (7,2) node {$\bullet$};
\draw (5,3/2) node {$\blacktriangledown$}; \draw (6,3/2) node {$\blacktriangledown$};
\draw (0,0) -- (7,0); \draw (3,1) -- (4,1); \draw (3,2) -- (4,2);
\draw (1,3/2) -- (2,3/2); \draw (5,3/2) -- (6,3/2);
\draw (0,0) to[out=-160,in=-20] (7,0);
\draw (0,1) to[out=-160,in=-20] (7,1); \draw (0,1) to[out=30,in=180] (3,1); \draw (4,1) to[out=0,in=150] (7,1);
\draw (0,2) to[out=160,in=20] (7,2); \draw (0,2) to[out=0,in=150] (3,2); \draw (4,2) to[out=30,in=180] (7,2);
\draw (1,3/2) to[bend left=90] (2,3/2); \draw (5,3/2) to[bend left=90] (6,3/2);

\draw (10,0) node {$B_2$}; \draw (10,1) node {$B_1$}; \draw (10,2) node {$\A$};
\draw (19,0) node {$B_2$}; \draw (19,1) node {$B_1$}; \draw (19,2) node {$\A$};
\draw (11,0) node {$\bullet$}; \draw (11,1) node {$\bullet$}; \draw (11,2) node {$\bullet$};
\draw (14,0) node {$\bullet$}; \draw (14,1) node {$\bullet$}; \draw (14,2) node {$\bullet$};
\draw (12,3/2) node {$\blacktriangledown$}; \draw (13,3/2) node {$\blacktriangledown$};
\draw (15,0) node {$\bullet$}; \draw (15,1) node {$\bullet$}; \draw (15,2) node {$\bullet$};
\draw (18,0) node {$\bullet$}; \draw (18,1) node {$\bullet$}; \draw (18,2) node {$\bullet$};
\draw (16,3/2) node {$\blacktriangledown$}; \draw (17,3/2) node {$\blacktriangledown$};
\draw (11,0) -- (18,0); \draw (14,1) -- (15,1); \draw (14,2) -- (15,2);
\draw (12,3/2) -- (13,3/2); \draw (16,3/2) -- (17,3/2);
\draw (11,0) to[out=-160,in=-20] (18,0);
\draw (11,1) to[out=-160,in=-20] (18,1); \draw (11,1) to[bend right=10] (18,1); \draw (14,1) to[bend right=45] (15,1);
\draw (11,2) to[out=160,in=20] (18,2); \draw (11,2) to[bend left=10] (18,2); \draw (14,2) to[bend left=45] (15,2);
\draw (12,3/2) to[out=-30, in=180] (16,3/2); \draw (13,3/2) to[out=0, in=150] (17,3/2);
\end{tikzpicture}
\end{center}
\label{fig:p=2_k=2_1}
\end{figure}

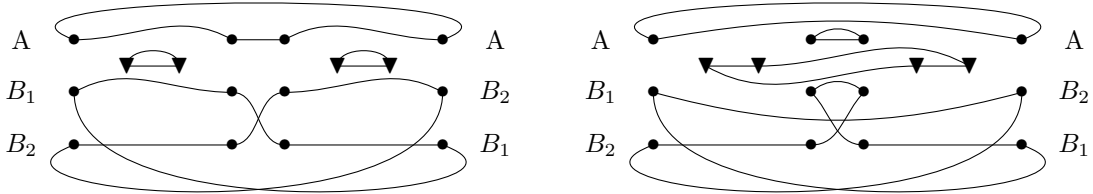
\begin{figure}[h]
\caption{$f(1)=1$, $f(2)=2$. On the left, $\alpha=\id$: $\mathcal{D}_{f,\alpha}=d^3s^2$. On the right, $\alpha=(1\,2)$: $\mathcal{D}_{f,\alpha}=d^3s$.}
\begin{center}
\begin{tikzpicture} [scale=0.7]
\draw (-1,0) node {$B_2$}; \draw (-1,1) node {$B_1$}; \draw (-1,2) node {$\A$};
\draw (8,0) node {$B_1$}; \draw (8,1) node {$B_2$}; \draw (8,2) node {$\A$};
\draw (0,0) node {$\bullet$}; \draw (0,1) node {$\bullet$}; \draw (0,2) node {$\bullet$};
\draw (3,0) node {$\bullet$}; \draw (3,1) node {$\bullet$}; \draw (3,2) node {$\bullet$};
\draw (1,3/2) node {$\blacktriangledown$}; \draw (2,3/2) node {$\blacktriangledown$};
\draw (4,0) node {$\bullet$}; \draw (4,1) node {$\bullet$}; \draw (4,2) node {$\bullet$};
\draw (7,0) node {$\bullet$}; \draw (7,1) node {$\bullet$}; \draw (7,2) node {$\bullet$};
\draw (5,3/2) node {$\blacktriangledown$}; \draw (6,3/2) node {$\blacktriangledown$};
\draw (0,0) -- (3,0); \draw (4,0) -- (7,0); \draw (3,2) -- (4,2);
\draw (1,3/2) -- (2,3/2); \draw (5,3/2) -- (6,3/2);
\draw (0,0) to[out=-160,in=-90] (7,1); \draw (0,1) to[out=-90,in=-20] (7,0);
\draw (3,0) to[out=0, in=180] (4,1); \draw (3,1) to[out=0, in=180] (4,0);
\draw (0,1) to[out=30, in=180] (3,1); \draw (4,1) to[out=0, in=150] (7,1);
\draw (0,2) to[out=160,in=20] (7,2); \draw (0,2) to[out=0,in=150] (3,2); \draw (4,2) to[out=30,in=180] (7,2);
\draw (1,3/2) to[bend left=90] (2,3/2); \draw (5,3/2) to[bend left=90] (6,3/2);

\draw (10,0) node {$B_2$}; \draw (10,1) node {$B_1$}; \draw (10,2) node {$\A$};
\draw (19,0) node {$B_1$}; \draw (19,1) node {$B_2$}; \draw (19,2) node {$\A$};
\draw (11,0) node {$\bullet$}; \draw (11,1) node {$\bullet$}; \draw (11,2) node {$\bullet$};
\draw (14,0) node {$\bullet$}; \draw (14,1) node {$\bullet$}; \draw (14,2) node {$\bullet$};
\draw (12,3/2) node {$\blacktriangledown$}; \draw (13,3/2) node {$\blacktriangledown$};
\draw (15,0) node {$\bullet$}; \draw (15,1) node {$\bullet$}; \draw (15,2) node {$\bullet$};
\draw (18,0) node {$\bullet$}; \draw (18,1) node {$\bullet$}; \draw (18,2) node {$\bullet$};
\draw (16,3/2) node {$\blacktriangledown$}; \draw (17,3/2) node {$\blacktriangledown$};
\draw (11,0) -- (14,0); \draw (15,0) -- (18,0); \draw (14,2) -- (15,2);
\draw (12,3/2) -- (13,3/2); \draw (16,3/2) -- (17,3/2);
\draw (11,0) to[out=-160, in=-90] (18,1); \draw (11,1) to[out=-90, in=-20] (18,0);
\draw (14,0) to[out=0, in=-135] (15,1); \draw (14,1) to[out=-45, in=180] (15,0);
\draw (11,1) to[bend right=15] (18,1); \draw (14,1) to[bend left=45] (15,1);
\draw (11,2) to[out=160,in=20] (18,2); \draw (11,2) to[bend left=10] (18,2); \draw (14,2) to[bend left=45] (15,2);
\draw (12,3/2) to[out=-30, in=180] (16,3/2); \draw (13,3/2) to[out=0, in=150] (17,3/2);
\end{tikzpicture}
\end{center}
\label{fig:p=2_k=2_2}
\end{figure}

Finally, we also know by Proposition \ref{prop:geodesic-alpha_f} that for any $\alpha\in\mathfrak{S}(p)$ and $f:[p]\rightarrow[k]$, denoting by $\gamma_f$ the product of the canonical full cycles on the level sets of $f$, we have $\sharp(\hat{\gamma}^{-1}\hat{\alpha}_f) = \sharp(\gamma_f^{-1}\alpha) + k - |\im(f)|$.

Putting everything together, we eventually come to the result summarized in Proposition \ref{prop:wishart-p-preliminary} below.

\begin{proposition} \label{prop:wishart-p-preliminary} For any $d,s\in\N$ and any $p\in\N$, we have
\begin{align*} \mathbf{E}_{W_{\A\B}\sim\mathcal{W}_{d^2,s}} \mathrm{Tr}\left[\left(\sum_{j=1}^k\widetilde{W}_{\A\B^k}(j)\right)^p\right] = & \sum_{f:[p]\rightarrow[k]}\mathbf{E}_{W_{\A\B}\sim\mathcal{W}_{d^2,s}} \mathrm{Tr}\left[\underset{i=1}{\overset{p}{\overrightarrow{\prod}}}\widetilde{W}_{\A\B^k}(f(i))\right] \\
= & \underset{f:[p]\rightarrow[k]}{\sum}\underset{\alpha\in\mathfrak{S}(p)}{\sum} d^{\sharp(\gamma^{-1}\alpha)+\sharp(\gamma_f^{-1}\alpha)+k-|\im(f)|}s^{\sharp(\alpha)}.
\end{align*}
\end{proposition}

\section{Counting geodesics vs non-geodesics pairings and permutations}
\label{appendix:geodesics}

Let us recall once and for all two notation that we will use repeatedly in this section, and that were introduced in Lemma \ref{lemma:catalan}. For any $p,m\in\N$ with $m\leq p$, we denote by $\mathrm{Cat}_p=\frac{1}{p+1}{2p \choose p}$ the $p^{th}$ Catalan number, and by $\mathrm{Nar}_p^m=\frac{1}{p+1}{p+1 \choose m}{p-1 \choose m-1}$ the $(p,m)^{th}$ Narayana number.

\subsection{Number of pairings of $2p$ elements which are not on the geodesics between the identity and the canonical full cycle}

\begin{lemma} \label{lemma:number-pairings-defect}
Let $p\in\N$ and denote by $\gamma$ the canonical full cycle on $\{1,\ldots,2p\}$. For any $0\leq\delta\leq\lfloor p/2\rfloor$, define the set of pairings having a defect $2\delta$ of being on the geodesics between $\id$ and $\gamma$ as
\[ \mathfrak{P}^{(2)}_{\delta}(2p) = \{\lambda\in\mathfrak{P}^{(2)}(2p) \st \sharp(\gamma^{-1}\lambda)=p+1-2\delta\}. \]
Then, the cardinality of $\mathfrak{P}^{(2)}_{\delta}(2p)$ is upper bounded by $\mathrm{Cat}_p \left(p^4/4\right)^{\delta}$.
\end{lemma}

To prove Lemma \ref{lemma:number-pairings-defect} (and later on Lemma \ref{lemma:number-functions,pairings-defect}) we will need the simple observation below. Roughly speaking, it will allow us to assume without loss of generality that, in the decomposition of an element of $\mathfrak{P}^{(2)}_{\delta}(2p)$ into $p$ disjoint transpositions, the ones ``creating'' the $2\delta$ geodesic defects are the $2\delta$ first ones.

\begin{fact} \label{fact:ordering-defects}
Let $\varsigma$ be a permutation on $\{1,\ldots,q\}$ and $\tau_1,\tau_2,\tau_3$ be three disjoint transpositions on $\{1,\ldots,q\}$, for some integer $q\geq 6$. Define $\varsigma^{(1)}=\varsigma\,\tau_1$, $\varsigma^{(2)}=\varsigma\,\tau_1\tau_2$, $\varsigma^{(3)}=\varsigma\,\tau_1\tau_2\tau_3$, and assume that
\begin{equation} \label{eq:not-canonical} \sharp(\varsigma^{(1)})=\sharp(\varsigma)+1,\ \sharp(\varsigma^{(2)})=\sharp(\varsigma)+2,\  \sharp(\varsigma^{(3)})=\sharp(\varsigma)+1. \end{equation}
Then, there exists a permutation $\pi$ of the three indices $\{1,2,3\}$ such that, defining this time $\varsigma_{\pi}^{(1)}=\varsigma\,\tau_{\pi(1)}$, $\varsigma_{\pi}^{(2)}=\varsigma\,\tau_{\pi(1)}\tau_{\pi(2)}$, $\varsigma_{\pi}^{(3)}=\varsigma\,\tau_{\pi(1)}\tau_{\pi(2)}\tau_{\pi(3)}$, we have
\[ \sharp(\varsigma_{\pi}^{(1)})=\sharp(\varsigma)+1,\ \sharp(\varsigma_{\pi}^{(2)})=\sharp(\varsigma),\  \sharp(\varsigma_{\pi}^{(3)})=\sharp(\varsigma)+1. \]
\end{fact}

\begin{proof}
Assume that $\varsigma$ and $\tau_1=(i_1\,j_1),\tau_2=(i_2\,j_2),\tau_3=(i_3\,j_3)$ satisfy equation \eqref{eq:not-canonical}. This means that $i_1,j_1$ belong to the same cycle of $\varsigma$, $i_2,j_2$ belong to the same cycle of $\varsigma^{(1)}$, and $i_3,j_3$ belong to two different cycles of $\varsigma^{(2)}$. So let us inspect all the scenarios which may occur.\\
$\bullet$ $c_1 \overset{(i_1\,j_1)}{\rightarrow} c_1^xc_1^y$ and $c_2 \overset{(i_2\,j_2)}{\rightarrow} c_2^xc_2^y$, with $c_1,c_2$ two different cycles of $\varsigma$: If $i_3\in c_1^x$ and $j_3\in c_1^y$ then the re-ordering $1,3,2$ is suitable. If $i_3\in c_2^x$ and $j_3\in c_2^y$ then the re-ordering $2,3,1$ is suitable. If $i_3\in c_1^a$ and $j_3\in c_2^b$, for $a,b\in\{x,y\}$, then both re-orderings $1,3,2$ and $2,3,1$ are suitable. And similarly when the roles of $i_3$ and $j_3$ are exchanged.\\
$\bullet$ $c \overset{(i_1\,j_1)}{\rightarrow} c'c'' \overset{(i_2\,j_2)}{\rightarrow} c^xc^yc^z$, with $c$ a cycle of $\varsigma$, while $c^z=c''$ and $c' \overset{(i_2\,j_2)}{\rightarrow} c^xc^y$: If $i_3\in c^x$ and $j_3\in c^y$ then the re-ordering $2,3,1$ is suitable. If $i_3\in c^a$, for $a\in\{x,y\}$, and $j_3\in c^z$ then the re-ordering $1,3,2$ is suitable. And similarly when the roles of $i_3$ and $j_3$ are exchanged.
\end{proof}

As an immediate consequence of Fact \ref{fact:ordering-defects}, we have the following: Let $\varsigma\in\mathfrak{S}(2p)$ and $\lambda=\tau_1\cdots\tau_p\in\mathfrak{P}^{(2)}(2p)$. Define for each $1\leq q\leq p$, $\varsigma^{(q)}=\varsigma\,\tau_1\cdots\tau_q$, as well as $\varsigma^{(0)}=\varsigma$. Assume next that, for some $0\leq\delta\leq\lfloor (p+\sharp(\varsigma))/2\rfloor$,
\[ \sharp(\varsigma^{(p)})=\sharp(\varsigma)+p-2\delta. \]
Then, there exists a permutation $\pi$ of the $p$ indices $\{1,\ldots,p\}$ such that, defining this time for each $1\leq q\leq p$, $\varsigma_{\pi}^{(q)}=\varsigma\,\tau_{\pi(1)}\cdots\tau_{\pi(q)}$, as well as $\varsigma_{\pi}^{(0)}=\varsigma$, we have
\begin{equation} \label{eq:canonical} \forall\ 1\leq q\leq p,\ \begin{cases} \sharp(\varsigma_{\pi}^{(q)})=\sharp(\varsigma_{\pi}^{(q-1)})-1\ \text{if}\ q\in\{2\epsilon \st 1\leq\epsilon\leq\delta\} \\  \sharp(\varsigma_{\pi}^{(q)})=\sharp(\varsigma_{\pi}^{(q)})+1\ \text{if}\ q\notin\{2\epsilon \st 1\leq\epsilon\leq\delta\} \end{cases}. \end{equation}
Since $\lambda=\tau_1\cdots\tau_p=\tau_{\pi(1)}\cdots\tau_{\pi(p)}$, we see that we may always assume without loss of generality that, given $\varsigma$, the transpositions $\tau_1,\ldots,\tau_p$ in the decomposition of $\lambda$ are ordered so that $\lambda$ is under the canonical form \eqref{eq:canonical}. The behaviour of the function $q\in[p]\mapsto \sharp(\varsigma^{(q)})$ under this hypothesis, depending on the value of $\delta$, is represented in Figure \ref{fig:pairings} (in the special case $p=6$ and $\sharp(\varsigma)=1$).

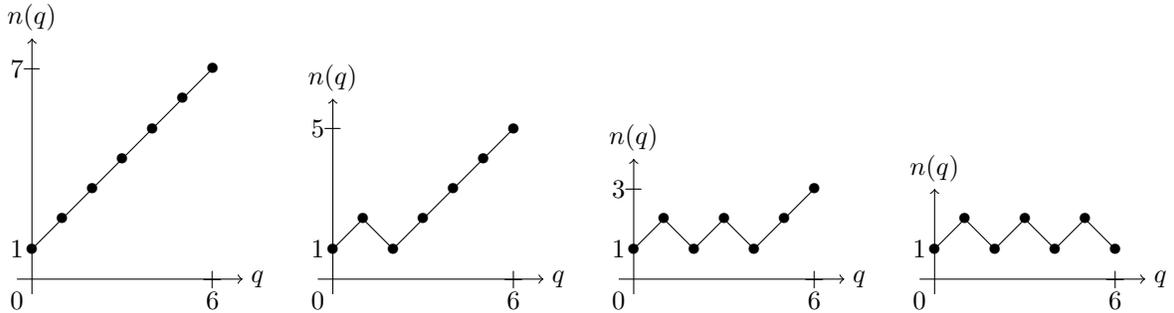
\begin{figure}[h]
\caption{Case $p=6$ and $\sharp(\varsigma)=1$. From left to right: $\delta=0$, $\delta=1$, $\delta=2$ and $\delta=3$.}
\begin{center}
\begin{tikzpicture} [scale=0.4]
\draw [->] (-0.5,0) -- (7,0); \draw [->] (0,-0.5) -- (0,8);
\draw (0,8.7) node {$n(q)$}; \draw (7.5,0) node {$q$};
\draw (-0.5,-0.7) node {$0$}; \draw (-0.5,1) node {$1$}; \draw (0,7) node {$-$}; \draw (-0.5,7) node {$7$}; \draw (6,0) node {$+$}; \draw (6,-0.7) node {$6$};
\draw (0,1) -- (6,7);
\draw (0,1) node {$\bullet$}; \draw (1,2) node {$\bullet$}; \draw (2,3) node {$\bullet$}; \draw (3,4) node {$\bullet$}; \draw (4,5) node {$\bullet$}; \draw (5,6) node {$\bullet$}; \draw (6,7) node {$\bullet$};

\draw [->] (9.5,0) -- (17,0); \draw [->] (10,-0.5) -- (10,6);
\draw (10,6.7) node {$n(q)$}; \draw (17.5,0) node {$q$};
\draw (9.5,-0.7) node {$0$}; \draw (9.5,1) node {$1$}; \draw (10,5) node {$-$}; \draw (9.5,5) node {$5$}; \draw (16,0) node {$+$}; \draw (16,-0.7) node {$6$};
\draw (10,1) -- (11,2); \draw (11,2) -- (12,1); \draw (12,1) -- (16,5);
\draw (10,1) node {$\bullet$}; \draw (11,2) node {$\bullet$}; \draw (12,1) node {$\bullet$}; \draw (13,2) node {$\bullet$}; \draw (14,3) node {$\bullet$}; \draw (15,4) node {$\bullet$}; \draw (16,5) node {$\bullet$};

\draw [->] (19.5,0) -- (27,0); \draw [->] (20,-0.5) -- (20,4);
\draw (20,4.7) node {$n(q)$}; \draw (27.5,0) node {$q$};
\draw (19.5,-0.7) node {$0$}; \draw (19.5,1) node {$1$}; \draw (20,3) node {$-$}; \draw (19.5,3) node {$3$}; \draw (26,0) node {$+$}; \draw (26,-0.7) node {$6$};
\draw (20,1) -- (21,2); \draw (21,2) -- (22,1); \draw (22,1) -- (23,2); \draw (23,2) -- (24,1); \draw (24,1) -- (26,3);
\draw (20,1) node {$\bullet$}; \draw (21,2) node {$\bullet$}; \draw (22,1) node {$\bullet$}; \draw (23,2) node {$\bullet$}; \draw (24,1) node {$\bullet$}; \draw (25,2) node {$\bullet$}; \draw (26,3) node {$\bullet$};

\draw [->] (29.5,0) -- (37,0); \draw [->] (30,-0.5) -- (30,3);
\draw (30,3.7) node {$n(q)$}; \draw (37.5,0) node {$q$};
\draw (29.5,-0.7) node {$0$}; \draw (29.5,1) node {$1$}; \draw (36,0) node {$+$}; \draw (36,-0.7) node {$6$};
\draw (30,1) -- (31,2); \draw (31,2) -- (32,1); \draw (32,1) -- (33,2); \draw (33,2) -- (34,1); \draw (34,1) -- (35,2); \draw (35,2) -- (36,1);
\draw (30,1) node {$\bullet$}; \draw (31,2) node {$\bullet$}; \draw (32,1) node {$\bullet$}; \draw (33,2) node {$\bullet$}; \draw (34,1) node {$\bullet$}; \draw (35,2) node {$\bullet$}; \draw (36,1) node {$\bullet$};
\end{tikzpicture}
\end{center}
\label{fig:pairings}
\end{figure}

With this result in mind, let us now turn to the proof of Lemma \ref{lemma:number-pairings-defect}.

\begin{proof}[Proof of Lemma \ref{lemma:number-pairings-defect}]
Given $\lambda=(i_1\,j_1)\cdots (i_p\,j_p)\in\mathfrak{P}^{(2)}(2p)$, we will always assume from now that the transpositions $(i_1\,j_1),\ldots,(i_p\,j_p)$ in its decomposition are ordered so that $\lambda$ is under the canonical form \eqref{eq:canonical} for $\gamma^{-1}$. This means the following: defining, for each $1\leq q\leq p$, the permutation $\widetilde{\lambda}^{(q)}=\gamma^{-1}(i_1\,j_1)\cdots (i_q\,j_q)$ and the integer $n(q)=\sharp(\widetilde{\lambda}^{(q)})$, as well as $\widetilde{\lambda}^{(0)}=\gamma^{-1}$ and $n(0)=\sharp(\widetilde{\lambda}^{(0)})=1$, we have, for any $0\leq\delta\leq\lfloor p/2\rfloor$,
\[ \lambda\in\mathfrak{P}^{(2)}_{\delta}(2p)\ \Leftrightarrow\ \forall\ 1\leq q\leq p,\ \begin{cases} n(q)=n(q-1)-1\ \text{if}\ q\in\{2\epsilon \st 1\leq\epsilon\leq\delta\} \\  n(q)=n(q-1)+1\ \text{if}\ q\notin\{2\epsilon \st 1\leq\epsilon\leq\delta\} \end{cases}. \]
In particular, $\lambda\in NC^{(2)}(2p)\ \Leftrightarrow\ \forall\ 1\leq q\leq p,\ n(q)=n(q-1)+1$, and we know that there are precisely $\mathrm{Cat}_p$ possibilities to build such pairing $\lambda$.
This implies that, for each $1\leq\delta\leq\lfloor p/2\rfloor$, there are necessarily less than
${2p \choose 2}\cdots{2(p-\delta+1) \choose 2}\times\mathrm{Cat}_{p-2\delta}$ possibilities to build a $\lambda\in\mathfrak{P}^{(2)}_{\delta}(2p)$.
Indeed, for the choice of the $2\delta$ first disjoint transpositions we can use the trivial upper bound that would consist in picking them completely arbitrarily, while the $p-2\delta$ last ones have to be chosen so that they form a non-crossing pairing of the $2p-4\delta$ not yet selected indices. Now, we just have to observe that
\begin{align*} {2p \choose 2}\cdots{2(p-2\delta+1) \choose 2}\times \mathrm{Cat}_{p-2\delta} =\, & \frac{2p\cdots(2p-4\delta+1)}{2^{2\delta}}\times\frac{(2p-4\delta)!}{(p-2\delta)!(p-2\delta+1)!} \\
=\, & \frac{1}{2^{2\delta}}\times\frac{p!(p+1)!}{(p-2\delta)!(p-2\delta+1)!}\times \frac{(2p)!}{p!(p+1)!} \\
\leq\, & \frac{p^{4\delta}}{2^{2\delta}}\times\mathrm{Cat}_p, \end{align*}
which completes the proof.
\end{proof}

\subsection{One needed generalization: bounding the number of pairings of $2p$ elements which are not on the geodesic path between the identity and a product of (few) cycles}

The proof of Proposition \ref{prop:gaussian-infty} crucially relies at some point on a statement of the same kind as the one appearing in Lemma \ref{lemma:number-pairings-defect}. Nevertheless, what we actually need there is a slight generalization of the latter. More specifically, we have to bound the number of pairings which have some defect of lying on the geodesics between the identity and, not only a full cycle, but also a product of (few) cycles. So let us give the following extension of Lemma \ref{lemma:number-pairings-defect}, which is really directed towards the application that we have in mind.

\begin{lemma} \label{lemma:number-functions,pairings-defect}
Let $p\in\N$. For any $f:[2p]\rightarrow[k]$ and any $0\leq\delta\leq \left\lfloor \left(p+|\im(f)|\right)/2\right\rfloor$, define the set of pairings having a defect $2\delta$ of being on the geodesics between $\id$ and $\gamma_f$ (the product of the canonical full cycles on each of the $|\im(f)|$ level sets of $f$) as
\[ \mathfrak{P}^{(2)}_{f,\delta}(2p) = \{\lambda\in\mathfrak{P}^{(2)}(2p) \st \sharp(\gamma_f^{-1}\lambda)=p+|\im(f)|-2\delta\}. \]
Then, for any $0\leq\delta\leq\lfloor p/2\rfloor$, we have the upper bound
\[ \left|\left\{(f,\lambda) \st \lambda\in\mathfrak{P}^{(2)}_{f,\delta}(2p) \right\} \right| \leq k^{p+2\delta}\mathrm{Cat}_p \left(\frac{p^4}{4}\right)^{\delta}. \]
\end{lemma}

\begin{proof}
We will follow the same strategy and employ the same notation as in the proof of Lemma \ref{lemma:number-pairings-defect}. Given $\lambda=(i_1\,j_1)\cdots(i_p\,j_p)\in\mathfrak{P}^{(2)}(2p)$ and $f:[2p]\rightarrow[k]$, we will always assume that the transpositions $(i_1\,j_1),\ldots,(i_p\,j_p)$ in the decomposition of $\lambda$ are ordered so that $\lambda$ is under the canonical form \eqref{eq:canonical} for $\gamma_f^{-1}$. This means the following: defining, for each $1\leq q\leq p$, $\widetilde{\lambda}^{(q)}=\gamma_f^{-1}(i_1\,j_1)\cdots (i_q\,j_q)$ and $n(q)=\sharp(\widetilde{\lambda}^{(q)})$, as well as $\widetilde{\lambda}^{(0)}=\gamma_f^{-1}$ and $n(0)=\sharp(\widetilde{\lambda}^{(0)})=|\im(f)|$, we have, for any $0\leq\delta\leq\lfloor (p+|\im(f)|)/2\rfloor$,
\begin{equation} \label{eq:b} \lambda\in\mathfrak{P}^{(2)}_{f,\delta}(2p)\ \Leftrightarrow\ \forall\ 1\leq q\leq p,\ \begin{cases} n(q)=n(q-1)-1\ \text{if}\ q\in\{2\epsilon \st 1\leq\epsilon\leq\delta\} \\  n(q)=n(q-1)+1\ \text{if}\ q\notin\{2\epsilon \st 1\leq\epsilon\leq\delta\} \end{cases}. \end{equation}

In particular, $\lambda\in \mathfrak{P}^{(2)}_{f,0}(2p)\ \Leftrightarrow\ \forall\ 1\leq q\leq p,\ n(q)=n(q-1)+1$, and we know that there are precisely $k^p\,\mathrm{Cat}_p$ possibilities to build a pair $(f,\lambda)$ satisfying this condition (because the latter holds if and only if both constraints $\lambda\in NC^{(2)}(2p)$ and $f\circ\lambda=f$ are fulfilled).

In the case $1\leq\delta\leq\lfloor p/2\rfloor$, notice that after the $2\delta$ first steps, we are left with a permutation $\overline{\varsigma}$ having $|\im(f)|$ cycles, and we have to impose that the partial pairing $\overline{\lambda}=(i_{2\delta+1}\,j_{2\delta+1})\cdots(i_{2p}\,j_{2p})$ lies on the geodesics between $\id$ and $\overline{\varsigma}$. Now, the number of such partial pairings is the same as the number of partial pairings lying on the geodesics between $\id$ and $\gamma_{\overline{f}}$, for any function $\overline{f}:[2p]\rightarrow[k]$ whose level sets are the supports of the cycles of $\overline{\varsigma}$. Hence, to build a pair $(\overline{f},\lambda)$ meeting our requirements, we have at most $k^{4\delta}{2p \choose 2}\cdots{2(p-\delta+1) \choose 2} \times k^{p-2\delta}\,\mathrm{Cat}_{p-2\delta}$ possibilities. Indeed, for the $2\delta$ first disjoint transpositions we can use the trivial upper bound that would consist in picking them, as well as the values of $\overline{f}$ on them, completely arbitrarily, while for the $p-2\delta$ last ones we have to impose that they are non-crossing and that $\overline{f}$ takes only one value on a given transposition. Now, we know from the proof of Lemma \ref{lemma:number-pairings-defect} that ${2p \choose 2}\cdots{2(p-\delta+1) \choose 2}\mathrm{Cat}_{p-2\delta}\leq(p^4/4)^{\delta}\mathrm{Cat}_p$. So we get as announced that there are less than $k^{p+2\delta}\,\mathrm{Cat}_p(p^4/4)^{\delta}$ pairs $(f,\lambda)$ satisfying condition \eqref{eq:b}.
\end{proof}

\subsection{One needed adaptation: bounding the number of permutations of $p$ elements which are not on the geodesic path between the identity and a product of (few) cycles}

The proof of Proposition \ref{prop:wishart-infty} requires a statement analogous to the one appearing in Lemma \ref{lemma:number-functions,pairings-defect}, but for permutations instead of pairings. In order to derive it, we need first to explicit a bit how an element of $\mathfrak{S}(p)$ can be put in one-to-one correspondence with an element of $\mathfrak{P}^{(2)}(2p)$ whose pairs are all composed of one even integer and one odd integer.

To a full cycle $c=(i_l\,\ldots\,i_1)$ on $\{1,\ldots,l\}$ we associate the pairing $\lambda_c=(2i_1\,2i_2-1)\cdots(2i_l\,2i_1-1)$ on $\{1,\ldots,2l\}$. The reverse operation is obtained by collapsing the two elements $2i$ and $2i-1$ to a single element $i$ for each $1\leq i\leq p$. Then as expected, we associate to a general permutation $\alpha=c_1\cdots c_m\in\mathfrak{S}(p)$ the pairing $\lambda_{\alpha}=\lambda_{c_1}\cdots\lambda_{c_m}\in\mathfrak{P}^{(2)}(2p)$.

Observe that, denoting by $\gamma$ the canonical full cycle either on $\{1,\ldots,p\}$ or on $\{1,\ldots,2p\}$, we have
\begin{equation} \label{eq:permutation-pairing} \forall\ \alpha\in\mathfrak{S}(p),\ \sharp(\alpha) + \sharp(\gamma^{-1}\alpha) = \sharp(\gamma^{-1}\lambda_{\alpha}). \end{equation}
Indeed, the cycles of $\gamma^{-1}\lambda_{\alpha}$ are precisely cycles of the form $(2i_i\,\ldots\,2i_l)$ for $(i_l\,\ldots\,i_1)$ a cycle of $\alpha$ (supported on even integers) and of the form $(2i_{l'}-1\,\ldots\,2i_1-1)$ for $(i_{l'}\,\ldots\,i_1)$ a cycle of $\gamma^{-1}\alpha$ (supported on odd integers). So what equation \eqref{eq:permutation-pairing} shows is that the elements of $\mathfrak{S}(p)$ having a given geodesic defect are in bijection with the elements of $\mathfrak{P}^{(2)}(2p)$ with even-odd pairs only and having the same geodesic defect (between $\id$ and $\gamma$ in both cases). In particular, we recover the well-known bijection between $NC(p)$ and $NC^{(2)}(2p)$ (because a non-crossing pairing is necessarily composed of even-odd pairs only).

Next, for any function $g$, either from $[p]$ to $[k]$ or from $[2p]$ to $[k]$, we will denote by $\gamma_g$ the permutation, either on $\{1,\ldots,p\}$ or on $\{1,\ldots,2p\}$, which is the product of the canonical full cycles on the level sets of $g$. For any function $f:[p]\rightarrow[k]$, we define the function $\widetilde{f}:[2p]\rightarrow[k]$ by $\widetilde{f}(2i)=\widetilde{f}(2i-1)=f(i)$ for each $1\leq i\leq p$. It is then easy to see that we have more generally
\[ \forall\ f:[p]\rightarrow[k],\ \forall\ \alpha\in\mathfrak{S}(p),\ \sharp(\alpha) + \sharp(\gamma_f^{-1}\alpha) = \sharp(\gamma_{\widetilde{f}}^{-1}\lambda_{\alpha}). \]

This simple observation will allow us to derive, as a slight adaptation of Lemma \ref{lemma:number-functions,pairings-defect}, a corresponding estimate for permutations instead of pairings.

\begin{lemma} \label{lemma:number-functions,permutations-defect}
Let $p\in\N$.
For any $f:[p]\rightarrow[k]$, any $0\leq\delta\leq \left\lfloor \left(p+|\im(f)|\right)/2\right\rfloor$, and any $1\leq m\leq p-2\delta$, define the set of permutations which are composed of $m$ disjoint cycles and which have a defect $2\delta$ of being on the geodesics between $\id$ and $\gamma_f$ (the product of the canonical full cycles on each of the $|\im(f)|$ level sets of $f$) as
\[ \mathfrak{S}_{f,\delta,m}(p) = \{\alpha\in\mathfrak{S}(p) \st \sharp(\alpha)=m\ \ \text{and}\ \ \sharp(\gamma_f^{-1}\alpha)+\sharp(\alpha)=p+|\im(f)|-2\delta\}. \]
Then, for any $0\leq\delta\leq\lfloor p/2\rfloor$ and any $1\leq m\leq p-2\delta$, we have the upper bound
\[ \left|\big\{(f,\alpha) \st \alpha\in\mathfrak{S}_{f,\delta,m}(p) \big\} \right| \leq \left(4k^4p^4\right)^{\delta}\sum_{\epsilon=0}^{2\delta}k^{m-\epsilon}\mathrm{Nar}_p^{m-\epsilon}. \]
\end{lemma}

\begin{proof}
We just observed that, for any $0\leq\delta\leq\lfloor p/2\rfloor$ and $1\leq m\leq p-2\delta$, the following equivalence holds
\begin{equation} \label{eq:c} \alpha\in\mathfrak{S}_{f,\delta,m}(p)\ \Leftrightarrow\ \sharp(\alpha)=m\ \ \text{and}\ \ \lambda_{\alpha}\in\mathfrak{P}^{(2)}_{\widetilde{f},\delta}, \end{equation}
where $\mathfrak{P}^{(2)}_{\widetilde{f},\delta}$ denotes the set of pairings having a defect $2\delta$ of lying on the geodesics between $\id$ and $\gamma_{\widetilde{f}}$, as defined in Lemma \ref{lemma:number-functions,pairings-defect}.

In particular, $\alpha\in\mathfrak{S}_{f,0,m}(p)\ \Leftrightarrow\ \sharp(\alpha)=m\ \ \text{and}\ \ \lambda_{\alpha}\in\mathfrak{P}^{(2)}_{\widetilde{f},0}$, and we know that there are precisely $k^m\mathrm{Nar}_p^m$ possibilities to build a pair $(f,\alpha)$ satisfying this condition (because the latter holds if and only if the three constraints $\sharp(\alpha)=m$, $\alpha\in NC(p)$ and $f\circ\alpha=f$ are fulfilled).

For the case $1\leq\delta\leq\lfloor p/2\rfloor$, we will mimic the proof of Lemma \ref{lemma:number-functions,pairings-defect}. So let $(f,\alpha)$ be such that $\alpha\in\mathfrak{S}_{f,\delta,m}(p)$ and assume without loss of generality that the transpositions $(i_1\,j_1),\ldots,(i_p\,j_p)$ in $\lambda_{\alpha}$ are ordered so that $\lambda_{\alpha}$ is under the canonical form \eqref{eq:canonical} for $\gamma_{\widetilde{f}}$. This means that the partial pairing $(i_{2\delta+1}\,j_{2\delta+1})\cdots(i_p\,j_p)$ is on the geodesics between $\id$ and some $\overline{\varsigma}$ with $|\im(f)|$ cycles, and the number of such partial pairings is the same as the number of partial pairings being on the geodesics between $\id$ and some $\gamma_{\overline{f}}$ with $|\im(\overline{f})|=|\im(f)|$. Hence, to count how many ways there are of constructing what happens on $\{i_1,j_1,\ldots,i_{2\delta},j_{2\delta}\}$, we have the trivial upper bound that would arise if picking the $2\delta$ first transpositions in $\lambda_{\alpha}$, as well as the values of $\overline{f}$ on them, completely arbitrarily. This yields a number of possibilities of at most $k^{4\delta}{2p \choose 2}\cdots{2(p-\delta+1) \choose 2}$. While on $\{i_{2\delta+1},j_{2\delta+1},\ldots,i_{2p},j_{2p}\}$, we have to impose that the $p-2\delta$ last transpositions in $\lambda_{\alpha}$ are non-crossing, and that, when collapsed into a permutation of $p-2\delta$ elements, the latter has between $m-2\delta$ and $m$ cycles and the function $\overline{f}$ takes only one value on each of them. This leaves us with a number of possibilities of at most $\sum_{\epsilon=0}^{2\delta}k^{m-\epsilon}\mathrm{Nar}_{p-2\delta}^{m-\epsilon}$. Putting everything together, we see that the number of pairs $(f,\alpha)$ satisfying condition \eqref{eq:c} is less than
\[ k^{4\delta}{2p \choose 2}\cdots{2(p-\delta+1) \choose 2}\sum_{\epsilon=0}^{2\delta}k^{m-\epsilon}\mathrm{Nar}_{p-2\delta}^{m-\epsilon} \leq k^{4\delta}\left(2p^2\right)^{2\delta}\sum_{\epsilon=0}^{2\delta}k^{m-\epsilon}\mathrm{Nar}_p^{m-\epsilon}, \]
which is exactly what we wanted to show.
\end{proof}

\begin{remark}
The upper bound we established in Lemma \ref{lemma:number-functions,permutations-defect} is probably far from optimal (e.g.~it is likely that the exponent $4\delta$ in the polynomial pre-factor in $k$ and $p$ can be improved). But this does not really matter for our specific goal. Nonetheless, in the special case of non-geodesic permutations between $\id$ and $\gamma$ on $\{1,\ldots,p\}$, it is in fact quite easy to obtain an upper bound which scales as $p^{3\delta}$ for the ratio between the number of $2\delta$ non-geodesic permutations with a given number of cycles and the number of geodesic permutations with the same number of cycles. We present the result in Lemma \ref{lemma:number-permutations-defect} below, the problem being that the proof method does not seem to generalize so straightforwardly to the case that we truly need, that is the one of non-geodesic permutations between $\id$ and $\gamma_f$.

Note also that very similar looking upper bounds had previously been derived regarding the cardinality of the set $\mathfrak{S}_{\delta}(p) = \{\alpha\in\mathfrak{S}(p) \st \sharp(\gamma^{-1}\alpha)+\sharp(\alpha)=p+1-2\delta\}$, which is the union of the sets $\mathfrak{S}_{\delta,m}(p)$ defined in Lemma \ref{lemma:number-permutations-defect}, for $1\leq m\leq p$. In particular, it was established in \cite{Montanaro}, Lemma 12, that for any $0\leq\delta\leq\lfloor p/2\rfloor$,
\[ \big|\mathfrak{S}_{\delta}(p)\big| \leq \big|\mathfrak{S}_{0}(p)\big|\,p^{3\delta} = \mathrm{Cat}_p \,p^{3\delta}.\]
However, this is definitely even less enough for our purpose: the latter really requires an upper bound on the number of permutations which have a given defect and a given number of cycles in terms of the number of permutations which have no defect and the same (or a related) number of cycles.

On the other hand, one may have hoped for a stronger result than these simply counting ones. For instance something like
\[ d(\id,\alpha)+d(\alpha,\gamma)=d(\id,\gamma)+2\delta\ \Rightarrow\ \exists\ \alpha' \st d(\alpha,\alpha')=2\delta'\ \ \text{and}\ \ d(\id,\alpha')+d(\alpha',\gamma)=d(\id,\gamma), \]
with $\delta'\leq \theta\delta$ and with the mapping $\phi:\alpha\mapsto\alpha'$ satisfying $\left|\phi^{-1}(\alpha')\right|\leq p^{\kappa\delta}$, for some coefficients $\theta,\kappa$. However, determining whether this kind of statement holds or not seems to remain an open question.
\end{remark}

\begin{lemma} \label{lemma:number-permutations-defect}
Let $p\in\N$ and denote by $\gamma$ the canonical full cycle on $\{1,\ldots,p\}$. For any $0\leq\delta\leq\lfloor p/2\rfloor$ and $1\leq m\leq p-2\delta$, define the set of permutations which are composed of $m$ disjoint cycles and which are $2\delta$-away from the geodesics between $\id$ and $\gamma$ as
\[ \mathfrak{S}_{\delta,m}(p) = \{\alpha\in\mathfrak{S}(p) \st \sharp(\alpha)=m\ \ \text{and}\ \ \sharp(\gamma^{-1}\alpha)+\sharp(\alpha)=p+1-2\delta\}. \]
Then, the cardinality of $\mathfrak{S}_{\delta,m}(p)$ is upper bounded in terms of the cardinality of $\mathfrak{S}_{0,m}(p)$ as
\[ \big|\mathfrak{S}_{\delta,m}(p)\big| \leq \big|\mathfrak{S}_{0,m}(p)\big|\left(\frac{p^3}{2}\right)^{\delta}. \]
\end{lemma}

\begin{proof} Let $p\in\N$ and $1\leq m\leq p$. We know from \cite{GS}, Theorems 4.1 and 4.2, that there exist polynomials $P_{q}$ of degree $q$, for $0\leq q\leq \lfloor p/2\rfloor$, such that for any $0\leq\delta\leq \lfloor(p-m)/2\rfloor$,
\begin{equation} \label{eq:S_delta,m} \big|\mathfrak{S}_{\delta,m}(p)\big| = \frac{p!}{2^{2\delta}(2\delta)!} {p+1-2\delta \choose m} \sum_{\epsilon=0}^{\delta} {p-1 \choose m-1+2\epsilon}P_{\epsilon}(m)P_{\delta-\epsilon}(p+1-m-2\delta) .\end{equation}
What is more, one can check from the explicit expression provided there for the polynomials $P_q$ that, for any $x\geq 0$, $P_q(x)\leq (2x)^q$. As a particular instance of equation \eqref{eq:S_delta,m}, we have
\[ \big|\mathfrak{S}_{0,m}(p)\big| = p!\, {p+1 \choose m}{p-1 \choose m-1}. \]
And as a consequence, we get by a brutal upper bounding that, for any $1\leq\delta\leq \lfloor(p-m)/2\rfloor$,
\begin{align*}
\frac{\big|\mathfrak{S}_{\delta,m}(p)\big|}{\big|\mathfrak{S}_{0,m}(p)\big|} = & \,\frac{1}{2^{2\delta}(2\delta)!}\prod_{i=0}^{m-1}\frac{p+1-2\delta-i}{p+1-i} \sum_{\epsilon=0}^{\delta}\prod_{j=0}{2\epsilon-1}\frac{p-m-j}{m+j} P_{\epsilon}(m)P_{\delta-\epsilon}(p+1-m-2\delta) \\
\leq  & \,\frac{1}{2^{2\delta}(2\delta)!} \sum_{\epsilon=0}^{\delta}\left(\frac{p-m}{m}\right)^{2\epsilon}\, 2^{\epsilon}\,m^{\epsilon}\, 2^{\delta-\epsilon}\,(p+1-m-2\delta)^{\delta-\epsilon} \\
\leq  & \,\frac{1}{2^{2\delta}(2\delta)!} \times (\delta+1)\,2^{\delta}\,p^{3\delta}\\
\leq & \,\left(\frac{p^3}{2}\right)^{\delta}.
\end{align*}
And this is precisely the claimed upper bound.
\end{proof}

\begin{remark} There is a close link between the problem we are concerned with and the one of finding tractable expressions for the so-called \textit{connection coefficients} of the symmetric group (the reader is referred e.g.~to \cite{GJ} for more on that topic). Closed formulas are actually known for the connection coefficients of $\mathfrak{S}(p)$, involving the characters of its irreducible representations. But unfortunately, they are not really handleable in there full generality. And it seems it is only in some specific cases that more manageable forms can been obtained (i.e.~in the first place as a sum of positive terms, so that one can see more easily what its order of magnitude is). The two situations which are well-understood are, on the one hand when the function $f$ is constant (which corresponds to the case where $\gamma_f$ is the canonical full cycle, and hence has a particularly simple cycle type, that is treated e.g.~in \cite{GS}), and on the other hand when the defect $2\delta$ is $0$ (which corresponds to the case of so-called \textit{top connection coefficients}).
\end{remark}

\section{Extra remarks on the convergence of the studied random matrix ensembles}
\label{appendix:convergence}

For any Hermitian $M$ on $\C^n$, we shall denote by $\lambda_1(M),\ldots,\lambda_n(M)\in\R$ its eigenvalues, and by $N_M$ its eigenvalue distribution, i.e.~the probability measure on $\R$ defined by
\[ N_M = \frac{1}{n}\sum_{i=1}^n\delta_{\lambda_i(M)}. \]
In words, for any $I\subset\R$, $N_M(I)$ is the proportion of eigenvalues of $M$ which belong to $I$.

\subsection{``Modified'' Wishart ensemble}

Fix $k\in\N$ and $c>0$. Then, for each $d\in\N$, let $W\sim\mathcal{W}_{d^2,cd^2}$ and define the random positive semidefinite matrix $W_d$ on $(\C^d)^{\otimes k+1}$ by
\begin{equation} \label{eq:W_d} W_d=\frac{1}{d^2}\sum_{j=1}^k\widetilde{W}(j). \end{equation}
Proposition \ref{prop:wishart-p} establishes that when $d\rightarrow+\infty$, the eigenvalue distribution of $W_d$ converges in moments towards a Mar\v{c}enko-Pastur distribution of parameter $ck$. But a stronger result actually holds, namely that there is convergence in probability of $N_{W_d}$ towards $\mu_{MP(ck)}$. What is meant is made precise in Theorem \ref{th:convergenceW} below.

\begin{theorem} \label{th:convergenceW}
For any $I\in\R$ and any $\e>0$,
\[ \underset{d\rightarrow+\infty}{\lim} \P_{W\sim\mathcal{W}_{d^2,cd^2}}\left(\left|N_{W_d}(I)-\mu_{MP(ck)}(I)\right|>\e\right) =0 ,\]
where the matrix $W_d$ is as defined in equation \eqref{eq:W_d}.
\end{theorem}

Theorem \ref{th:convergenceW} is a direct consequence of the estimate on the $p$-order moments $\mathbf{E}\, \mathrm{Tr}\left[\left(\sum_{j=1}^k\widetilde{W}_{\A\B^k}(j)\right)^p\right]$ from Proposition \ref{prop:wishart-p}, combined with the estimate on the $p$-order variances $\mathbf{Var}\, \mathrm{Tr}\left[\left(\sum_{j=1}^k\widetilde{W}_{\A\B^k}(j)\right)^p\right]$ from Proposition \ref{prop:wishart-p-var} below. The proof, which follows a quite standard procedure, may be found detailed in \cite{AGZ} and sketched in \cite{Aubrun1}.

\begin{proposition}\label{prop:wishart-p-var} Let $p\in\N$. For any constant $c>0$,
\[ \mathbf{Var}_{W_{\A\B}\sim\mathcal{W}_{d^2,cd^2}}\, \mathrm{Tr}\left[\left(\underset{j=1}{\overset{k}{\sum}}\widetilde{W}_{\A\B^k}(j)\right)^p\right] \underset{d\rightarrow+\infty}{=} o\left(d^{2p+k+1}\right). \]
\end{proposition}

\begin{proof}
Let $p\in\N$. We already know that $\left(\mathbf{E}\, \mathrm{Tr}\left[\left(\sum_{j=1}^k\widetilde{W}_{\A\B^k}(j)\right)^p\right]\right)^2 \sim_{d\rightarrow +\infty} \left(\mathrm{M}_{MP(ck)}^{(p)}d^{2p+k+1} \right)^2$ thanks to Proposition \ref{prop:wishart-p}. Consequently, the only thing that remains to be shown in order to establish Proposition \ref{prop:wishart-p-var} is that we also have $\mathbf{E}\,\left( \mathrm{Tr}\left[\left(\sum_{j=1}^k\widetilde{W}_{\A\B^k}(j)\right)^p\right]\right)^2 \sim_{d\rightarrow +\infty} \left(\mathrm{M}_{MP(ck)}^{(p)}d^{2p+k+1} \right)^2$. The combinatorics involved in the proof of the latter estimate is very similar to the one already appearing in the proof of the former. We will therefore skip some of the details here.

To begin with, let us fix a few additional notation. We define $\gamma_1=(p\,\ldots\,1)$ and $\gamma_2=(2p\,\ldots\,p+1)$ as the canonical full cycles on $\{1,\ldots,p\}$ and $\{p+1,\ldots,2p\}$ respectively. Also, for each functions $f_1:\{1,\ldots,p\}\rightarrow[k]$, $f_2:\{p+1,\ldots,2p\}\rightarrow[k]$, we define the function $f_{1,2}:[2p]\rightarrow[k]$ by $f_{1,2}=f_1$ on $\{1,\ldots,p\}$ and $f_{1,2}=f_2$ on $\{p+1,\ldots,2p\}$. Then, by a slight generalization of Proposition \ref{prop:geodesic-alpha_f} we have that, for any $\alpha\in\mathfrak{S}(2p)$,
\[ \sharp((\hat{\gamma}_1\hat{\gamma}_2)^{-1}\hat{\alpha}_{f_{1,2}})=\sharp((\gamma_{1\,f_1}\gamma_{2\,f_2})^{-1}\alpha)+ 2k -|\im(f_1)|-|\im(f_2)|. \]
We can thus derive from the graphical calculus for Wishart matrices (in complete analogy to the way formula \eqref{eq:p-moment} was obtained) that, for any $d\in\N$,
\begin{equation} \label{eq:varp} \mathbf{E}_{W_{\A\B}\sim\mathcal{W}_{d^2,cd^2}} \left(\mathrm{Tr}\left[\left(\sum_{j=1}^k\widetilde{W}_{\A\B^k}(j)\right)^p\right]\right)^2 = \sum_{\substack{f_1:\{1,\ldots,p\}\rightarrow[k]\\f_2:\{p+1,\ldots,2p\}\rightarrow[k]}}\sum_{\alpha\in\mathfrak{S}(2p)} c^{\sharp(\alpha)} d^{n(\alpha,f_1,f_2)}, \end{equation}
where for each $\alpha\in\mathfrak{S}(2p)$ and $f_1:\{1,\ldots,p\}\rightarrow[k]$, $f_2:\{p+1,\ldots,2p\}\rightarrow[k]$,
\[ n(\alpha,f_1,f_2)= 2\sharp(\alpha)+ \sharp((\gamma_1\gamma_2)^{-1}\alpha)+ \sharp((\gamma_{1\,f_1}\gamma_{2\,f_2})^{-1}\alpha)+2k-|\im(f_1)|-|\im(f_2)|. \]
Yet, by Lemma \ref{lemma:cycles-transpositions} and equation \eqref{eq:geodesic''} in Lemma \ref{lemma:distance}, we get: First, for any $\alpha\in\mathfrak{S}(2p)$,
\begin{equation}
\label{eq:alpha-var}
\sharp(\alpha)+\sharp((\gamma_1\gamma_2)^{-1}\alpha)= 4p-\left(|\alpha|+|(\gamma_1\gamma_2)^{-1}\alpha|\right)\leq 4p-|\gamma_1\gamma_2|= 2p+\sharp(\gamma_1\gamma_2) =2p+2,
\end{equation}
with equality if and only if $\alpha=\alpha_1\alpha_2$ where $\alpha_1\in NC(\{1,\ldots,p\})$, $\alpha_2\in NC(\{p+1,\ldots,2p\})$. And second, for any $\alpha\in\mathfrak{S}(2p)$ and any $f_1:\{1,\ldots,p\}\rightarrow[k]$, $f_2:\{p+1,\ldots,2p\}\rightarrow[k]$,
\begin{equation}
\label{eq:alpha_f'-var}
\sharp(\alpha)+\sharp((\gamma_{1\,f_1}\gamma_{2\,f_2})^{-1}\alpha)\leq 2p+\sharp(\gamma_{1\,f_1}\gamma_{2\,f_2}) =2p+|\im(f_1)|+|\im(f_2)|,
\end{equation}
with equality if and only if $\alpha=\alpha_1\alpha_2$ where $\alpha_1\in NC(\{1,\ldots,p\})$ and $f_1\circ\alpha_1=f_1$, $\alpha_2\in NC(\{p+1,\ldots,2p\})$ and $f_2\circ\alpha_2=f_2$. So putting equations \eqref{eq:alpha-var} and \eqref{eq:alpha_f'-var} together, we get in the end that for any $\alpha\in\mathfrak{S}(2p)$ and $f_1:\{1,\ldots,p\}\rightarrow[k]$, $f_2:\{p+1,\ldots,2p\}\rightarrow[k]$,
\[ n(\alpha,f_1,f_2) \leq 4p+2k+2, \]
with equality if and only if $\alpha=\alpha_1\alpha_2$ where $\alpha_1\in NC(\{1,\ldots,p\})$ and $f_1\circ\alpha_1=f_1$, $\alpha_2\in NC(\{p+1,\ldots,2p\})$ and $f_2\circ\alpha_2=f_2$.

We thus get that, asymptotically, the dominant term in formula \eqref{eq:varp} factorizes as
\begin{align*}
\mathbf{E}\, \left(\mathrm{Tr}\left[\left(\sum_{j=1}^k\widetilde{W}_{\A\B^k}(j)\right)^p\right]\right)^2 & \underset{d\rightarrow+\infty}{\sim} d^{4p+2k+2} \sum_{\substack{\alpha_1\in NC(\{1,\ldots,p\})\\ \alpha_2\in NC(\{p+1,\ldots,2p\})}} \sum_{\substack{f_1:\{1,\ldots,p\}\rightarrow[k],\,f_1\circ\alpha_1=f_1 \\f_2:\{p+1,\ldots,2p\}\rightarrow[k],\,f_2\circ\alpha_2=f_2}}  c^{\sharp(\alpha_1\alpha_2)}\\
& \underset{d\rightarrow+\infty}{\sim} \left( d^{2p+k+1}\sum_{\alpha\in NC(p)}\sum_{\underset{f\circ\alpha=f} {f:[p]\rightarrow[k]}}  c^{\sharp(\alpha)} \right)^2,
\end{align*}
where the last equality is simply because $\sharp(\alpha_1\alpha_2)=\sharp(\alpha_1)+\sharp(\alpha_2)$. And hence,
\[ \mathbf{E}_{W_{\A\B}\sim\mathcal{W}_{d^2,cd^2}} \left(\mathrm{Tr}\left[\left(\sum_{j=1}^k\widetilde{W}_{\A\B^k}(j)\right)^p\right]\right)^2 \underset{d\rightarrow+\infty}{\sim} \left(d^{2p+k+1}\sum_{\alpha\in NC(p)}(ck)^{\sharp(\alpha)}\right)^2 = \left(d^{2p+k+1}\mathrm{M}_{MP(ck)}^{(p)}\right)^2, \]
which is exactly what we needed to conclude the proof.
\end{proof}

Let us illustrate the result stated in Theorem \ref{th:convergenceW} in the simplest case of $2$-extendibility and uniformly distributed mixed states. In Figure \ref{fig:wishart}, the spectral distribution of $W_d=\left(W_{\A\B_1}\otimes\Id_{\B_2}+W_{\A\B_2}\otimes\Id_{\B_1}\right)/d^2$, for $W_{\A\B}\sim\mathcal{W}_{d^2,d^2}$, and a Mar\v{c}enko-Pastur distribution of parameter $2$ are plotted together. The empirical eigenvalue histogram is done in dimension $d=12$, from $100$ repetitions.

\begin{figure}[h] \caption{Spectral distribution of $\left(W_{\A\B_1}\otimes\Id_{\B_2}+W_{\A\B_2}\otimes\Id_{\B_1}\right)/d^2$, for $W_{\A\B}\sim\mathcal{W}_{d^2,d^2}$ vs Mar\v{c}enko-Pastur distribution of parameter $2$.}
\label{fig:wishart}
\begin{center}
\includegraphics[width=12cm]{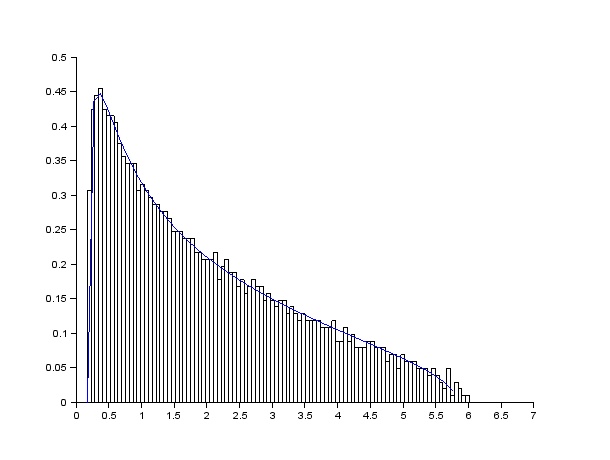}
\end{center}
\end{figure}

\subsection{``Modified'' GUE ensemble}

Fix $k\in\N$. Then, for each $d\in\N$, let $G\sim GUE(d^2)$ and define the random Hermitian matrix $G_d$ on $(\C^d)^{\otimes k+1}$ by
\begin{equation} \label{eq:G_d} G_d=\frac{1}{d}\sum_{j=1}^k\widetilde{G}(j). \end{equation}
In complete analogy to what was explained in the case of Wishart matrices, Proposition \ref{prop:gaussian-p} establishes that when $d\rightarrow+\infty$, the eigenvalue distribution of $G_d$ converges in moments towards a centered semicircular distribution of parameter $k$. But here again, there is in fact convergence in probability of $N_{G_d}$ towards $\mu_{SC(k)}$, which is made precise in Theorem \ref{th:convergenceG} below.

\begin{theorem} \label{th:convergenceG}
For any $I\in\R$ and any $\e>0$,
\[ \underset{d\rightarrow+\infty}{\lim} \P_{G\sim GUE(d^2)}\left(\left|N_{G_d}(I)-\mu_{SC(k)}(I)\right|>\e\right) =0, \]
where the matrix $G_d$ is as defined in equation \eqref{eq:G_d}.
\end{theorem}

As already explained in the Wishart case, this follows directly from the moment's estimate in Proposition \ref{prop:gaussian-p}, together with the variance's estimate, for all $p\in\N$,
\begin{equation} \label{eq:gaussian-p-var} \mathbf{Var}_{G_{\A\B}\sim GUE(d^2)}\, \mathrm{Tr}\left[ \left( \underset{j=1}{\overset{k}{\sum}}\widetilde{G}_{\A\B^k}(j) \right)^{2p} \right] \underset{d\rightarrow+\infty}{=} o\left(d^{2p+k+1}\right) .\end{equation}
The proof follows the exact same lines as the one of Proposition \ref{prop:wishart-p-var} and is not repeated here.

Let us illustrate the result stated in Theorem \ref{th:convergenceG} in the simplest case of $2$-extendibility. In Figure \ref{fig:gaussian}, the spectral distribution of $G_d=\left(G_{\A\B_1}\otimes\Id_{\B_2}+G_{\A\B_2}\otimes\Id_{\B_1}\right)/d^2$, for $G_{\A\B}\sim GUE(d^2)$, and a centered semicircular distribution of parameter $2$ are plotted together. The empirical eigenvalue histogram is done in dimension $d=10$, from $100$ repetitions.

\begin{figure}[h] \caption{Spectral distribution of $\left(G_{\A\B_1}\otimes\Id_{\B_2}+G_{\A\B_2}\otimes\Id_{\B_1}\right)/d^2$, for $G_{\A\B}\sim GUE(d^2)$ vs Centered semicircular distribution of parameter $2$.}
\label{fig:gaussian}
\begin{center}
\includegraphics[width=12cm]{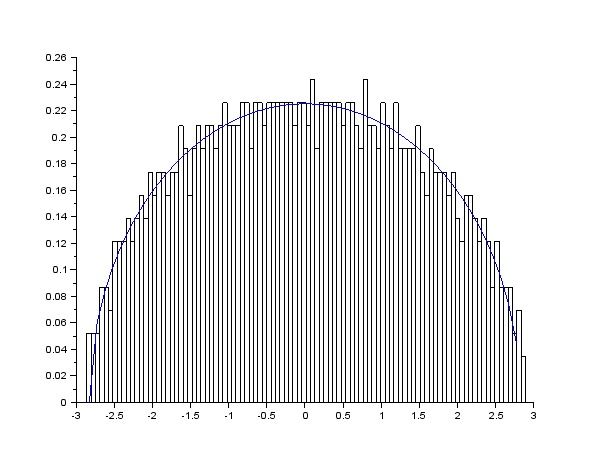}
\end{center}
\end{figure}

\begin{remark}
We may in fact say even more on the convergence of the random matrix sequences $(W_d)_{d\in\N}$ and $(G_d)_{d\in\N}$ defined by equations \eqref{eq:W_d} and \eqref{eq:G_d} respectively. Namely,
\[ N_{W_d}\overset{a.s.}{\underset{d\rightarrow+\infty}{\rightarrow}} \mu_{MP(ck)}\ \ \text{and}\ \  N_{G_d}\overset{a.s.}{\underset{d\rightarrow+\infty}{\rightarrow}} \mu_{SC(k)}.\]
To establish this almost sure convergence result, the only thing that has to be verified is that, for any $p\in\N$, the series of variances
\begin{equation} \label{eq:summable} \sum_{d=1}^{+\infty} \mathbf{Var}\left[\frac{1}{d^{k+1}}\mathrm{Tr}\left(W_d^p\right)\right]\ \text{and} \ \sum_{d=1}^{+\infty} \mathbf{Var}\left[\frac{1}{d^{k+1}}\mathrm{Tr}\left(G_d^{2p}\right)\right] \end{equation}
are summable. Indeed, almost sure convergence will then automatically follow from a standard application of the Chebyshev inequality and the Borel--Cantelli lemma. And condition \eqref{eq:summable} actually holds, as a consequence of the fact that, for any $p\in\N$,
\[ \mathbf{Var}\left[\frac{1}{d^{k+1}}\mathrm{Tr}\left(W_d^p\right)\right]= O\left(d^{-2}\right)\ \ \text{and}\ \ \mathbf{Var}\left[\frac{1}{d^{k+1}}\mathrm{Tr}\left(G_d^{2p}\right)\right]= O\left(d^{-2}\right). \]
\end{remark}

\subsection{Asymptotic freeness of certain Gaussian matrices}

Let us fix a few definitions and notation. Given $n\in\N$ and $[\Omega,P]$ a classical probability space, we define the free probability space $\left[\mathcal{M}_n(L^{\infty}[\Omega,P]),\varphi_n\right]$, where $\mathcal{M}_n(L^{\infty}[\Omega,P])$ is the set of $n\times n$ matrices with entries in $L^{\infty}[\Omega,P]$ and $\varphi_n(\cdot)=\E\tr(\cdot)/n$ is the normalized trace function on $\mathcal{M}_n(L^{\infty}[\Omega,P])$. The two particular examples we shall focus on in the sequel are the ones we have already been extensively dealing with, namely GUE and Wishart matrices.


\begin{lemma} \label{lem:phi-G}
Given two finite-dimensional Hilbert spaces $\A\equiv\C^{d_\A}$, $\B\equiv\C^{d_\B}$, and $G$ a random GUE matrix on $\A\otimes\B$, we define the following random matrices on $\A\otimes\B_1\otimes\B_2$:
\[ \widetilde{G}_{1}=\frac{1}{\sqrt{d_\A d_\B}}\,G_{\A\B_1}\otimes\Id_{\B_2}\ \ \text{and}\ \ \widetilde{G}_{2}=\frac{1}{\sqrt{d_\A d_\B}}\,G_{\A\B_2}\otimes\Id_{\B_1}. \]
Then, for any $p\in\N$ and any function $f:[2p]\rightarrow[2]$,
\begin{equation} \label{eq:phi-G} \lim_{d_\A\leq d_\B\rightarrow+\infty} \varphi_{d_\A d_\B^2}\left(\widetilde{G}_{f(1)}\ldots\widetilde{G}_{f(2p)}\right) = \left|\left\{\lambda\in NC^{(2)}(2p) \st f\circ\lambda =f \right\} \right|. \end{equation}
\end{lemma}

\begin{proof}
We know from the proof of Proposition 2.3 (and using the same notation as those employed there) that
\[ \varphi_{d_\A d_\B^2}\left(\widetilde{G}_{f(1)}\ldots\widetilde{G}_{f(2p)}\right) = \sum_{\lambda\in\mathfrak{P}^{(2)}(2p)} d_\A^{\sharp(\gamma^{-1}\lambda)-p-1}d_B^{\sharp(\gamma_f^{-1}\lambda)-p-|\im(f)|}. \]
Now, as explained there as well, for any $\lambda\in\mathfrak{P}^{(2)}(2p)$, on the one hand $\sharp(\gamma^{-1}\lambda)\leq p+1$ with equality iff $\lambda\in NC^{(2)}(2p)$, and on the other hand $\sharp(\gamma_f^{-1}\lambda)\leq p+ |\im(f)|$ with equality iff $\lambda\in NC^{(2)}(2p)$ and $f\circ\lambda=f$. The asymptotic estimate \eqref{eq:phi-G} therefore immediately follows.
\end{proof}

\begin{theorem} \label{th:free-G}
Let $G_{\A\B}$ be a random GUE matrix on $\A\otimes\B$. Then, the random matrices $G_{\A\B_1}\otimes\Id_{\B_2}$ and $G_{\A\B_2}\otimes\Id_{\B_1}$ on $\A\otimes\B_1\otimes\B_2$ are asymptotically free.
\end{theorem}

\begin{proof}
Theorem \ref{th:free-G} is a direct consequence of Lemma \ref{lem:phi-G} (see e.g.~\cite{NS}, proof of Proposition 22.22, for an entirely analogous argument). Indeed, as $d_\A,d_\B\rightarrow+\infty$, the two empirical spectral distributions $\mu_{\widetilde{G}_{1}}$ and $\mu_{\widetilde{G}_{2}}$ both converge to the semicircular distribution with mean $0$ and variance $1$. And equation \eqref{eq:phi-G} is exactly the rule for computing mixed moments in two free such semicircular distributions (see e.g.~\cite{NS}, Lecture 12).
\end{proof}

\begin{lemma} \label{lem:phi-W}
Given two finite-dimensional Hilbert spaces $\A\equiv\C^{d_\A}$, $\B\equiv\C^{d_\B}$, and $W$ a random Wishart matrix on $\A\otimes\B$ with parameter $c d_\A d_\B\in\N$, we define the following random matrices on $\A\otimes\B_1\otimes\B_2$:
\[ \widetilde{W}_{1}=\frac{1}{d_\A d_\B}\,W_{\A\B_1}\otimes\Id_{\B_2}\ \ \text{and}\ \ \widetilde{W}_{2}=\frac{1}{d_\A d_\B}\,W_{\A\B_2}\otimes\Id_{\B_1}. \]
Then, for any $p\in\N$ and any function $f:[p]\rightarrow[2]$,
\begin{equation} \label{eq:phi-W} \lim_{d_\A\leq d_\B\rightarrow+\infty} \varphi_{d_\A d_\B^2} \left(\widetilde{W}_{f(1)}\ldots\widetilde{W}_{f(p)}\right) = \sum_{\underset{f\circ\alpha=f}{\alpha\in NC(p)}}c^{\sharp(\alpha)}. \end{equation}
\end{lemma}

\begin{proof}
We know from the proof of Proposition 6.2 (and using the same notation as those employed there) that
\[ \varphi_{d_\A d_\B^2} \left(\widetilde{W}_{f(1)}\ldots\widetilde{W}_{f(p)}\right) = \sum_{\alpha\in\mathfrak{S}(p)} c^{\sharp(\alpha)}d_\A^{\sharp(\alpha)+\sharp(\gamma^{-1}\alpha)-p-1}d_\B^{\sharp(\alpha)+\sharp(\gamma_f^{-1}\lambda)-p-|\im(f)|}. \]
Now, as explained there as well, for any $\alpha\in\mathfrak{S}(p)$, on the one hand $\sharp(\alpha)+\sharp(\gamma^{-1}\alpha)\leq p+1$ with equality iff $\alpha\in NC(p)$, and on the other hand $\sharp(\alpha)+\sharp(\gamma_f^{-1}\alpha)\leq p+ |\im(f)|$ with equality iff $\alpha\in NC(p)$ and $f\circ\alpha=f$. The asymptotic estimate \eqref{eq:phi-W} therefore immediately follows.
\end{proof}

\begin{theorem} \label{th:free-W}
Let $W_{\A\B}$ be a random Wishart matrix on $\A\otimes\B$ with parameter $c d_\A d_\B\in\N$. Then, the random matrices $W_{\A\B_1}\otimes\Id_{\B_2}$ and $W_{\A\B_2}\otimes\Id_{\B_1}$ on $\A\otimes\B_1\otimes\B_2$ are asymptotically free.
\end{theorem}

\begin{proof}
Theorem \ref{th:free-W} is a direct consequence of Lemma \ref{lem:phi-W} (see e.g.~\cite{NS}, proof of Proposition 22.22, for an entirely analogous argument). Indeed, as $d_\A,d_\B\rightarrow+\infty$, the two empirical spectral distributions $\mu_{\widetilde{W}_{1}}$ and $\mu_{\widetilde{W}_{2}}$ both converge to the Mar\v{c}enko-Pastur distribution with parameter $c$. And equation \eqref{eq:phi-W} is exactly the rule for computing mixed moments in two free such Mar\v{c}enko-Pastur distributions (see e.g.~\cite{NS}, Lectures 12 and 13).
\end{proof}


\end{document}